\documentclass[graybox,vecarrow,envcountsect,envcountsame]{svmult}

\usepackage[bottom]{footmisc}

\usepackage{newtxtext}
\usepackage[varvw]{newtxmath}

\usepackage{mathtools}
\usepackage{enumitem}

\newcommand{\nin}{\not\in}

\newcommand{\defn}{\coloneqq}
\newcommand{\nfed}{\eqqcolon}

\newcommand{\one}{{{\mathchoice {\rm 1\mskip-4mu l} {\rm 1\mskip-4mu l} {\rm 1\mskip-4.5mu l} {\rm 1\mskip-5mu l}}}}

\newcommand\numberthis{\addtocounter{equation}{1}\tag{\theequation}}

\newcommand{\dif}{\mathop{}\!\mathrm{d}}
\newcommand{\e}{\mathrm{e}}
\newcommand{\im}{\mathrm{i}}

\newcommand{\field}[1][K]{{\mathbb{#1}}}

\newcommand{\RR}{{\field[R]}}
\newcommand{\CC}{{\field[C]}}
\newcommand{\CCR}{\mathrm{CCR}}
\newcommand{\RS}{\mathrm{RS}}
\renewcommand{\Re}{\operatorname{Re}}
\renewcommand{\Im}{\operatorname{Im}}

\DeclareMathOperator{\supp}{supp}

\DeclareMathOperator{\Dom}{\cD}
\DeclareMathOperator{\Ran}{\cR}
\DeclareMathOperator{\Ker}{\cN}

\DeclareMathOperator*{\slim}{s-lim}

\DeclareMathOperator{\sgn}{sgn}

\newcommand{\loc}{\mathrm{loc}}

\newcommand{\conj}[1]{\overline{#1}}
\DeclarePairedDelimiter{\abs}{\lvert}{\rvert}
\DeclarePairedDelimiter{\norm}{\lVert}{\rVert}

\DeclarePairedDelimiterX{\cinner}[2]{(}{)}{#1\mkern2mu\delimsize\vert\mkern2mu\mathopen{}#2}
\DeclarePairedDelimiterX{\rinner}[2]{\langle}{\rangle}{#1\mkern2mu\delimsize\vert\mkern2mu\mathopen{}#2}

\newcommand{\cl}{\mathrm{cl}}

\newcommand{\T}{\mathrm{T}}

\newcommand{\Feyn}{\mathrm{F}}
\newcommand{\aFeyn}{{\overline{\mathrm{F}}}}
\newcommand{\PJ}{\mathrm{PJ}}
\newcommand{\en}{\mathrm{en}}

\newcommand{\dyn}{\mathrm{dyn}}
\newcommand{\s}{{\mathrm s}}

\newcommand{\ri}{\mathrm{i}}
\renewcommand{\sc}{\mathrm{sc}}

\makeatletter
\newcommand{\aotimes}{{\mathop{\otimes}\limits^{
  \vbox to .15ex {\kern-2\ex@\hbox{\tiny alg}\vss}}}}
\makeatother
\newcommand{\Sol}{\cW}
\renewcommand{\bar}{\overline}

\newcommand{\cW}{{\mathcal W}}
\newcommand{\cD}{{\mathcal D}}
\newcommand{\cK}{{\mathcal K}}
\newcommand{\cV}{{\mathcal V}}
\newcommand{\cX}{{\mathcal X}}
\newcommand{\cY}{{\mathcal Y}}

\newcommand{\cZ}{{\mathcal Z}}
\newcommand{\cR}{{\mathcal R}}
\newcommand{\cN}{{\mathcal N}}
\newcommand{\torder}[1]{\mathrm{T}\{#1\}}
\newcommand{\atorder}[1]{\bar{\mathrm{T}}\{#1\}}
\newcommand{\dGamma}{\mathop{\dif\Gamma}}

\newcommand{\itembullet}{\hspace\leftmargin\llap{\textbullet\hspace\labelsep}}

\newcommand{\bes}{\begin{subequations}}
\newcommand{\ees}{\end{subequations}}

\spnewtheorem{conjecture}[theorem]{Conjecture}{\bfseries}{\itshape}
\spnewtheorem{assumption}[theorem]{Assumption}{\bfseries}{\rmfamily}

\begin{document}

\title*{An Evolution Equation Approach to Linear Quantum Field Theory}

\author{Jan Dereziński and Daniel Siemssen}

\institute{Jan Dereziński \at Department of Mathematical Methods in Physics, Faculty of Physics, University of Warsaw, Pasteura~5, 02-093 Warszawa, Poland, \email{jan.derezinski@fuw.edu.pl}
\and Daniel Siemssen \at Department of Mathematics and Informatics, University of Wuppertal, Gaußstraße 20, 42119 Wuppertal, Germany.}

\maketitle

\abstract{
  In the first part of our paper we analyze (non-autonomous) pseudo-unitary evolutions on Krein spaces.
  We describe various bisolutions and inverses of the corresponding Cauchy data operator.
  We prove that with boundary conditions given by a maximal positive and maximal negative space we can always associate an inverse, which can be viewed as a generalization of the usual \emph{Feynman propagator}.
  These constructions can be applied to a large class of globally hyperbolic manifolds.
  In this context, the Feynman propagator turns out to be a distinguished inverse of the Klein--Gordon operator.
  We discuss when the Feynman propagator can be viewed as the inverse of the Klein--Gordon operator in the operator-theoretic sense.
  In particular, if the conditions of asymptotic stationarity and stability are satisfied, we give heuristic arguments why we expect the Feynman propagator to be the boundary value of the resolvent of a self-adjoint realization of the Klein--Gordon operator.
  All these constructions find an important application in Quantum Field Theory on curved spacetimes.
  The Feynman propagator yields the expectation values of time-ordered products of fields between the ``in vacuum'' and the ``out vacuum'' -- the basic ingredient of Feynman diagrams.
  Hence it is sometimes called the \emph{in-out Feynman propagator}.
}

\renewcommand{\theequation}{\thesection.\arabic{equation}}
\numberwithin{equation}{section}

\section{Introduction}

The wave equation or, more generally, the Klein--Gordon equation on curved Lorentzian manifolds is one of the classic topics of linear partial differential equations \cite{leray,friedlander,bar}.
One could expect that it is difficult to find new important concepts in this subject.
However, the present paper analyzes a few natural objects associated to the Klein--Gordon equation, which we believe are rather fundamental and to a large extent were overlooked in the mathematical literature until quite recently.
These objects include the \emph{(in-out) Feynman and anti-Feynman inverse} (or \emph{propagator}), and various well-posed realizations of the \emph{Klein--Gordon operator}.
A proof of the existence of the Feynman and anti-Feynman propagators under rather mild assumptions is probably the main result of our paper.

The Feynman and anti-Feynman propagators play a central role when we compute the scattering operator in Quantum Field Theory (QFT) on curved spacetimes.
In fact, the Feynman propagator is associated to each internal line of a Feynman diagram.

The Feynman propagator is thus crucial in the global approach to QFT, involving the whole spacetime, when one wants to compute the scattering operator.
Our results are thus complementary to a large mathematical literature devoted to QFT on curved spacetimes involving the \emph{local} approach \cite{brunetti-fredenhagen,bar-fredenhagen,hollands-wald1,hollands}.

\subsection{Propagators on the Minkowski spacetime}

As indicated above, the main motivation of our paper comes from \emph{curved} spacetimes.
However, it is natural to start from the \emph{flat} Minkowski spacetime.
Let us recall the most important ``propagators'' or ``two-point functions'' used in QFT, known from many textbooks, e.g., Appendix 2 of Bogoliubov--Shirkov \cite{bogoliubov} and Appendix C of Bjorken--Drell \cite{bjorken}:
\begin{subequations}
  \label{pro}
  \begin{align}
    \intertext{\itembullet the forward/backward or retarded/advanced propagator}
    G^{\vee/\wedge}(x,y) & \defn \frac{1}{(2\uppi)^4} \int\frac{\e^{-\im(x-y)\cdot p}}{p^2+m^2\pm\im 0\sgn(p_0)} \dif p, \label{pro1}
    \\[.0em]
    \intertext{\itembullet the Feynman/anti-Feynman propagator}
    G^{\mathrm{F}/\bar{\mathrm{F}}}(x,y) & \defn \frac{1}{(2\uppi)^4} \int\frac{\e^{-\im(x-y)\cdot p}}{p^2+m^2\mp\im0} \dif p, \label{pro2}
    \\[.0em]
    \intertext{\itembullet the Pauli--Jordan propagator}
    G^{\mathrm{PJ}}(x,y) & \defn \frac{\im}{(2\uppi)^3} \int\e^{-\im(x-y)\cdot p}\sgn(p_0)\delta(p^2+m^2) \dif p, \label{pro3}
    \\[.0em]
    \intertext{\itembullet the positive/negative frequency or particle/antiparticle bisolution}
    G^{(\pm)}(x,y) & \defn \frac{1}{(2\uppi)^3} \int\e^{-\im(x-y)\cdot p}\theta(\pm p_0)\delta(p^2+m^2) \dif p.
    \label{pro4}
  \end{align}
\end{subequations}
Mathematically,~\eqref{pro1}, \eqref{pro2} are distinguished inverses and~\eqref{pro3}, \eqref{pro4} are distinguished bisolutions of the Klein--Gordon operator $-\Box+m^2$.
We will call them jointly ``propagators''.
They satisfy a number of identities:
\begin{subequations}
  \label{iden}
  \begin{align}
    G^\PJ              & = G^\vee - G^\wedge                                \label{idi1} \\
                       & = \im G^{(+)} - \im G^{(-)},                       \label{idi2} \\
    G^\Feyn - G^\aFeyn & = \im G^{(+)} +\im G^{(-)},                        \label{idi3} \\
    G^\Feyn + G^\aFeyn & = G^\vee + G^\wedge,                               \label{idi4} \\
    G^\Feyn            & = \im G^{(+)} + G^\wedge = \im G^{(-)} + G^\vee,   \label{idi5} \\
    G^\aFeyn           & = -\im G^{(+)} + G^\vee = -\im G^{(-)} + G^\wedge. \label{idi6}
  \end{align}
\end{subequations}

Let us describe applications of these propagators to quantum field theory.
We will restrict ourselves to scalar fields.
We will consider both basic formalisms for scalar free fields -- the \emph{neutral} or \emph{real} and \emph{charged} or \emph{complex} formalism.

In the neutral formalism the basic object is a self-adjoint operator-valued distribution on spacetime $\hat\phi(x)$ satisfying the Klein--Gordon equation
\begin{equation*}
  (-\Box + m^2) \hat\phi(x) = 0.
\end{equation*}
Here are various ``two-point functions'' of these fields:
\begin{subequations}
  \begin{align*}
    \intertext{\itembullet commutation relations}
    [\hat\phi(x),\hat\phi(y)] & = -\im G^\PJ(x,y)\one,
    \\[.5em]
    \intertext{\itembullet vacuum expectation of products of fields}
    \cinner{\Omega}{\hat\phi(x)\hat\phi(y)\Omega} & = G^{(+)}(x,y),
    \\[.5em]
    \intertext{\itembullet vacuum expectation of direct/reverse time-ordered products of fields}
    \cinner[\big]{\Omega}{\torder{\hat\phi(x)\hat\phi(y)}\Omega} & = -\im G^{\Feyn}(x,y), \\
    \cinner[\big]{\Omega}{\atorder{\hat\phi(x)\hat\phi(y)}\Omega} & = \im G^{\aFeyn}(x,y).
  \end{align*}
\end{subequations}

In the charged formalism the field is non-self-adjoint.
It will be denoted with a different letter: $\hat\psi(x)$.
It also satisfies the Klein--Gordon equation
\begin{equation*}
  (-\Box + m^2) \hat\psi(x) = (-\Box + m^2) \hat\psi^*(x) = 0.
\end{equation*}
The ``two-point functions'' of the charged field are slightly more rich than in the neutral case:
\begin{subequations}
  \label{two}
  \begin{align}
    [\hat\psi(x),\hat\psi^*(y)]                                     & = -\im G^\PJ(x,y)\one,  \label{two1}  \\[.5em]
    \cinner{\Omega}{\hat\psi(x)\hat\psi^*(y)\Omega}                 & = G^{(+)}(x,y),         \label{two2a} \\
    \cinner{\Omega}{\hat\psi^*(x)\hat\psi(y)\Omega}                 & = G^{(-)}(x,y),         \label{two2b} \\[.5em]
    \cinner[\big]{\Omega}{\torder{\hat\psi(x)\hat\psi^*(y)}\Omega}  & = -\im G^{\Feyn}(x,y),  \label{two3}  \\[.5em]
    \cinner[\big]{\Omega}{\atorder{\hat\psi(x)\hat\psi^*(y)}\Omega} & = \im G^{\aFeyn}(x,y). \label{two4}
  \end{align}
\end{subequations}

We will use the name \emph{classical propagators} as the joint name for $G^\PJ$, $G^\vee$ and $G^\wedge$.
The functions $G^{(+)}$, $G^{(-)}$, $G^\Feyn$ and $G^\aFeyn$ express vacuum expectation values, therefore they will be jointly called \emph{non-classical propagators}.

\subsection{The Klein--Gordon equation on a curved spacetime and its classical propagators}
\label{The Klein-Gordon equation}

In the first part of the introduction (until
Subsect.~\ref{sub-Abstract Klein-Gordon operator}) we will discuss the
Klein--Gordon equation and related objects in a purely mathematical
setting, without a direct reference to  classical or quantum fields.

Consider a globally hyperbolic manifold $M$ equipped with a metric tensor $g=[g_{\mu\nu}]$ and its inverse $g^{-1}=[g^{\mu\nu}]$, an electromagnetic potential $A=[A_\mu]$ and a scalar potential $Y$.
Throughout most of the introduction we will assume that $g,A,Y$ are smooth -- this assumption will not be necessary in the rest of our paper.

Let $D_\mu \defn -\im\partial_\mu$.
Our paper is devoted to the \emph{Klein--Gordon operator}:
\begin{align}
  \label{eq:klein-gordon-half-}
  K & \defn -\abs{g}^{-\frac14} (D_\mu - A_\mu) \abs{g}^\frac12 g^{\mu\nu} (D_\nu - A_\nu)\abs{g}^{-\frac14} - Y.
\end{align}
(Note that we use the so-called half-density formalism, see Subsect.~\ref{Half-densities on a pseudo-Riemannian manifold}.)
The equation
\begin{equation}
  \label{KGG}
  Ku=0
\end{equation}
will be called the (homogeneous) \emph{Klein--Gordon equation}.

It has been shown by many authors that there exist unique distributions $G^\vee$ and $G^\wedge$ on $M\times M$ with the following properties.
If $f\in C_\mathrm{c}^\infty(M)$, then
\begin{align*}
  u^\vee(x)=\int G^\vee(x,y) f(y)\dif y,\quad u^\wedge(x)=\int G^\wedge(x,y) f(y)\dif y
\end{align*}
satisfy
\begin{align*}
  Ku^\vee=K u^\wedge=f
\end{align*}
and $\supp u^\vee$, $\supp u^\wedge$ are contained in the future, resp.\ past causal shadow of $\supp f$.
We can also generalize the Pauli--Jordan propagator \eqref{pro3} by using the identity~\eqref{idi1}:
\[ G^\PJ = G^\vee - G^\wedge. \]

Thus the classical propagators $G^\PJ$, $G^\vee$ and $G^\wedge$ are well defined (and also well known) for general Klein--Gordon equations on globally hyperbolic manifolds.

\subsection{Pseudounitary structure}
\label{Pseudounitary structure}

We will denote by $\cW_\sc$ the space of smooth space-compact functions $M\ni x\mapsto u(x)\in\mathbb{C}$ solving \eqref{KGG}.

Let $u,v\in \cW_\sc$.
Then it is easy to  see that
\begin{align*}
  \notag j^\mu(x;\bar u,v) \defn & -\bar{u(x)} g^{\mu\nu}(x)|g|^{\frac14}(x)\bigl(D_\nu-A_\nu(x)\bigr)|g|^{-\frac14}(x)v(x) \\  & -\bar{\bigl( D_\nu -A_\nu(x)\bigr)|g|^{-\frac14}(x)u(x)}g^{\mu\nu}(x)|g|^{\frac14}(x)v(x).
\end{align*}
is a conserved current, that is $\partial_\mu j^\mu=0$.
Hence a Hermitian form on $\cW_\sc$ called sometimes the \emph{charge}
\begin{equation}
  \label{symplec0}
  \cinner{u}{Qv} \defn \int_{\mathcal{S}} j^\mu(x,\bar
  u,v)\dif\mathrm{s}_\mu(x)
\end{equation}
does not depend on the choice of a Cauchy surface $\mathcal{S}$ (where $\dif\mathrm{s}_\mu(x)$ denotes the natural measure on $\mathcal{S}$ times the normal vector).

Note that many authors instead of the charge prefer to use the symplectic form on $\cW_\sc$ given by the imaginary part of \eqref{symplec0}.

The charge form $Q$ is not positive definite: it contains vectors with positive and negative charge.
The space $\cW_\sc$ can be decomposed in many ways in a direct sum of a maximally positive space and a maximally negative space, both orthogonal to one another in the sense of the charge form.
Every such a decomposition can be encoded with help of an \emph{admissible involution} $S_\bullet$: an operator on $\cW_\sc$ satisfying
\begin{equation}
  \label{admis}
  S_\bullet^2=\one,\quad \cinner{u}{QS_\bullet v} = \cinner{S_\bullet u}{Q v}\quad \text{is positive}.
\end{equation}
As every involution, $S_\bullet$, determines a pair projections, so that
\begin{equation*}
  \Pi_\bullet^{(+)}+\Pi_\bullet^{(-)}=\one,\quad S_\bullet=\Pi_\bullet^{(+)}-\Pi_\bullet^{(-)}.
\end{equation*}
The ranges of $\Pi_\bullet^{(+)}$ and $\Pi_\bullet^{(-)}$ are $Q$-orthogonal and maximally positive, resp.\ negative.
They will be used to define Fock states in QFT.

\subsection{Non-classical propagators on curved spacetimes}
\label{Non-classical propagators on curved spacetimes}

In the literature it is often claimed that it makes no sense to ask for distinguished non-classical propagators on generic spacetimes.
The main message of our paper disputes this statement.
We will argue that for a large class of non-stationary spacetimes there exist physically relevant distinguished non-classical propagators.

We will consider two types of spacetimes.
\begin{enumerate}
  \item the \emph{slab geometry case}, where $M$ can be identified with $[t_-,t_+]\times\Sigma$, where $[t_-,t_+]$ is a finite interval describing time;
  \item the \emph{unrestricted time case}, where $M$ is a globally hyperbolic spacetime with no boundary, equipped with the Klein--Gordon equation which is \emph{asymptotically stationary} and \emph{stable} in the future and in the past.
\end{enumerate}

By the asymptotic stationarity we will mean that one can identify $M$ with $\mathbb{R}\times\Sigma$ such that $g$, $A$ and $Y$ converge as $t\to\pm\infty$ to limiting values sufficiently fast.
Often, for simplicity we will just assume that there exists $T>0$ such that for $\pm t>T$ the Klein--Gordon operator is stationary.

By stability we mean the positivity of the Hamiltonian.
Thus in Case (2) we assume that for large $t$ the Hamiltonian is positive.

Case (2) is probably more interesting both physically and mathematically.
Nevertheless, it is instructive to compare Case (1) with Case (2).

In our opinion the assumption of asymptotic stationarity and stability in Case (2) is natural from the physical point of view.
Asymptotic stationarity is a necessary condition to apply the ideas of scattering theory, which is the main means of extracting useful information from QFT.
Stability is satisfied in typical physics applications.

Non-classical propagators are associated with special boundary conditions in the past and future.
After quantization, they will be used to encode Fock representations of Canonical Commutation Relation.
In this subsection we will describe them without a reference to quantum fields, as a part of an operator-theoretic analysis of the Klein--Gordon operator.

Consider first Case (1).
In order to define non-classical propagators at time $t_+$ and $t_-$ we fix a pair of admissible involutions $S_+$, resp.\ $S_-$.
They lead to corresponding projections $\Pi_\pm^{(+)}$ and $\Pi_\pm^{(-)}$.
The choice of $S_\pm$ is to a large extent arbitrary, although one can
argue that  those satisfying the so-called \emph{Hadamard property} \cite{radzikowski,fewster-verch:necessary,wald:backreaction} are more physical than the others.

In Case (2)  it is natural to select the  admissible involutions $S_\pm$ given by the sign of the generator of the dynamics at $t\to\pm\infty$.
(Recall that we assume that the evolution is asymptotically stationary and stable.)
Thus $\Pi_\pm^{(+)}$ is the projection onto ``in/out positive frequency modes'' and $\Pi_\pm^{(-)}$ onto ``in/out negative frequency modes''.

The projections $\Pi_\pm^{(+)}$ and $\Pi_\pm^{(-)}$ naturally define two pairs of bisolutions of the Klein--Gordon equation, $G_\pm^{(+)}$ and $G_\pm^{(-)}$.
The identity~\eqref{idi2} now splits into two independent identities
\begin{equation}
  \label{idi2-}
  G^\PJ = \im G_\pm^{(+)} - \im G_\pm^{(-)},
\end{equation}

It is less obvious that the Feynman and anti-Feynman propagators also possess natural unique generalizations.
The Feynman propagator $G^{\Feyn}$ can be described as the inverse of the Klein--Gordon operator corresponding to the Cauchy data in $\Pi_+^{(+)}$ for $t=t_+$ or $t\to+\infty$, and in $\Pi_-^{(-)}$ for $t=t_-$ or $t\to+\infty$.
The anti-Feynman propagator $G^{\aFeyn}$ is the inverse of the Klein--Gordon operator corresponding to the Cauchy data in $\Pi_+^{(-)}$ for $t=t_+$ or $t\to+\infty$, and in $\Pi_-^{(+)}$ for $t=t_-$ or $t\to-\infty$.
Clearly, in Case (1) $G^{\Feyn}$ and $G^{\aFeyn}$ depend on the choice of $S_+, S_-$.
In Case (2) they are defined uniquely.

Note that $G^{\Feyn}$ are $G^{\aFeyn}$ are sometimes called the
\emph{  in-out}, resp.\ \emph{out-in Feynman propagators} \cite{fukuma}, to distinguish them from some other, non-canonical proposals, such as those mentioned below in (\ref{idi5.}).
We will sometimes use  these terms to stress their physical meaning.
However, in our opinion, when one writes \emph{the Feynman propagator} using the definite article \emph{the}, there should be no doubt that $G^\Feyn$ is meant.

Note that in the generic case the relations \eqref{iden} are not satisfied, except for \eqref{idi1} and the two versions of \eqref{idi2}, see also \eqref{idi2-}.

\subsection{Well-posedness/self-adjointness of the Klein--Gordon operator}
\label{Self-adjointness of the Klein-Gordon operator}

Formally, the Feynman and anti-Feynman propagators are inverses of the Klein--Gordon operator.
One can ask whether this can be interpreted in a more precise operator-theoretic sense.
We will see that this is often true, however the situation is quite different in Case (1) and (2).

It is easy to see that the Klein--Gordon operator $K$ is Hermitian (symmetric) on, say, $C_\mathrm{c}^\infty(M)$.
In Case (1) $K$ is obviously not essentially self-adjoint -- it possesses many extensions parametrized by boundary conditions at $t=t_+$ and $t=t_-$.
The admissible involutions $S_+$ and $S_-$ determine special boundary conditions that lead to closed realizations $K^\Feyn=(K^\aFeyn)^*$, so that we have
\begin{equation*}
  G^\Feyn=(K^\Feyn)^{-1},\quad G^\aFeyn=(K^\aFeyn)^{-1}=(G^\Feyn)^*.
\end{equation*}

Note that $K^\Feyn$ and $K^\aFeyn$ are not self-adjoint.
Clearly, they are invertible, and hence \emph{well-posed}.
\footnote{%
  An operator which has a non-empty resolvent set is  called \emph{well-posed}, see \cite{EE}.
  For instance, self-adjoint operators are well-posed.
}

In Case (2) there seems to be no need for boundary conditions and one can expect that $K$ is often essentially self-adjoint on $C_\mathrm{c}^\infty(M)$.
Suppose that $K$ is essentially self-adjoint and let us denote by $K^{\mathrm{s.a.}}$ its self-adjoint extension.
Then we can expect that
\begin{equation}
  \label{KG2}
  G^\Feyn =\lim_{\epsilon\searrow0} (K^{\mathrm{s.a.}}-\im\epsilon)^{-1},
  \quad
  G^\aFeyn = \lim_{\epsilon\searrow0}(K^{\mathrm{s.a.}}+\im\epsilon)^{-1},
\end{equation}
in the sense of quadratic forms on an appropriate weighted space, e.g., $\langle t\rangle^{-s} L^2(M)$ with $s>\frac12$.

The above conjectures are obviously true on the Minkowski space.
They also hold in the stationary case.
In the absence of the electrostatic potential this is straightforward, with the electrostatic potential it requires some work, see \cite{derezinski-siemssen:static}.
There exists also recent interesting papers by Vasy~\cite{vasy:selfadjoint} and Taira--Nakamura \cite{nakamura,nakamura2,nakamura3}, where all these properties are proven for some classes of spacetimes, mostly assuming the asymptotic Minkowskian property and non-trapping conditions.

The question of the self-adjointness of the Klein--Gordon operator is beyond the scope of our paper.
It is much more difficult to answer and most probably requires additional assumptions (like non-trapping conditions).
However, as the analysis of our paper shows, the Feynman inverse is well-defined for essentially all asymptotically stable and stationary spacetimes.

Note that if asymptotic stability and stationarity does not hold, but $K$ can be interpreted as a self-adjoint operator, then one can try to use \eqref{KG2} as the \emph{definition} of the Feynman/anti-Feynman propagators.

\subsection{Reduction to a 1st order  equation for the Cauchy data}
\label{Reduction to a 1st order  equation for the Cauchy data}

In order to compute non-classical (actually, also classical) propagators, it is useful to convert the Klein--Gordon equation into a 1st order evolution equation on the phase space describing Cauchy data.
To this end, we fix a decomposition $M=I\times\Sigma$, where $I=[t_-,t_+]$ or $I=\mathbb{R}$.
We assume that $M$ is Lorentzian and $\Sigma$ is Riemannian.
We will use Latin letters for spatial indices.
We introduce
\begin{align*}
  h=[h_{ij}]=[g_{ij}], & \quad h^{-1}=[h^{ij}], \\ \beta^j \defn g_{0i}h^{ij}, & \quad \alpha^2 \defn g_{0i}h^{ij}g_{j0}-g_{00}.
\end{align*}
Note that $[h_{ij}]$, $[h^{ij}]$ are positive definite and
$\alpha^2>0$.
Set
\begin{align*}
  L & \defn \abs{g}^{-\frac14} (D_i - A_i) \abs{g}^\frac12 h^{ij} (D_j - A_j)\abs{g}^{-\frac14} + Y, \\ W & \defn \beta^iD_i-A_0+\beta^iA_i+\frac{\ri}{4}\abs{g}^{-1}\abs{g}_{,0} -\frac{\ri}{4}\beta^i\abs{g}^{-1}\abs{g}_{,i}.
\end{align*}
Then the Klein--Gordon operator and the charge can be written as
\begin{align}
  K              & =(D_0+W^*)\frac{1}{\alpha^2}(D_0+W)-L,\label{hilbert1a}
  \\\notag
  \cinner{u}{Qv} & =\int_\Sigma
  \bar{u(t,\vec x)}\frac{1}{\alpha^2(t,\vec x)}\bigl(D_0+W(t,\vec x)\bigr)
  v(t,\vec x)\dif\vec x                                                    \\ & \quad+
    \int_\Sigma
    \bar{\bigl(D_0+W(t,\vec x)\bigr)
      u(t,\vec x)}\frac{1}{\alpha^2(t,\vec x)}
    v(t,\vec x)  \dif\vec x. \notag
\end{align}
Therefore, the Klein--Gordon equation  $Ku=0$ can be rewritten as a 1st order equation for the Cauchy data on $\Sigma$:
\[
  \bigl(\partial_t+\im B(t)\bigr)w=0,
\]
where
\begin{align*}
  B(t)=
        \begin{bmatrix}
          B_{11}(t) & B_{12}(t) \\B_{21}(t) & B_{22}(t)
        \end{bmatrix}
   & \defn
           \begin{bmatrix}
             W(t) & \alpha^2(t) \\L(t) & W^*(t)
           \end{bmatrix}
  ,                                                                \\ w=
  \begin{bmatrix}
    w_1 \\ w_2
  \end{bmatrix}
   & \defn
           \begin{bmatrix}
             u \\ -\alpha^{-2}\bigl(-\ri\partial_t+W(t)\bigr)u
           \end{bmatrix}
  .
\end{align*}
The current preserved by the dynamics is given by the matrix
\[Q=
  \begin{bmatrix}
    0 & \one \\\one & 0
  \end{bmatrix}
  .
\]
It is natural to introduce the \emph{classical Hamiltonian}
\begin{align*}
  H(t) & = QB(t) =
  \begin{bmatrix}
    L(t) & W^*(t) \\ W(t) & \alpha^2(t)
  \end{bmatrix}
\end{align*}
and the \emph{Cauchy data operator}
\begin{equation*}
  M \defn \partial_t+\im B(t).
\end{equation*}
(The notational clash with the occasionally appearing manifold $M$ should cause no confusion.)

We will say that an operator $E$ is a \emph{bisolution}/\emph{inverse or Green's operator of $M$} if it satisfies
\begin{alignat*}{2}
  M E w & = 0,                   & \quad E M w & = 0,
  \\\text{resp.}\quad
  M E w & = w,
        & \quad
  E M w & = w
\end{alignat*}
for a large class of functions $t\mapsto w(t)=\begin{bmatrix}w_1(t)\\w_2(t)\end{bmatrix}$.
An inverse/bisolution of $M$ can be written as a $2\times2$ matrix
\begin{align*}
  E(t,s)=
          \begin{bmatrix}
            E_{11}(t,s) & E_{12}(t,s) \\E_{21}(t,s) & E_{22}(t,s)
          \end{bmatrix}
  .
\end{align*}
If we set
\begin{equation}
  \label{rightupper}
  G(t,s) \defn \ri E_{12}(t,s),
\end{equation}
then $G$ is formally a bisolution/inverse of the Klein--Gordon operator:
\begin{alignat*}{2}
  KG u & = 0,                     & \quad GK u & = 0,
  \\\text{resp.}\quad
  KG u & = u,
        & \quad
  GK u & = u,
\end{alignat*}
 for a large class of spacetime functions $x\mapsto u(x)$.
\subsection{Stationary case}
\label{Stationary case}
The theory of propagators for the Klein--Gordon equation greatly simplifies in the stable stationary case.
More precisely, suppose for the moment that $M=\mathbb{R}\times\Sigma$ and that $g$, $A$ and $Y$, hence also $B(t) \nfed B$ and $H(t) \nfed H$, do not depend on time $t$.
Assume also that $B$ is stable, which means that $H$ is positive and has a zero nullspace.

First of all, the Cauchy data can then be organized in a Hilbert space.
Actually, using the Hamiltonian $H$ and the generator $B$, one can construct a whole scale of natural Hilbert spaces $\cW_\lambda, \lambda\in\mathbb{R}$, which can be used to describe the Cauchy data.
Among them three have a special importance.
The \emph{energy space}, $\cW_{\frac12}$, has the scalar product given by the Hamiltonian $H$.
There is also the \emph{dual energy space}, which we denote by $\cW_{-\frac12}$, with the scalar product given by $(QHQ)^{-1}$.
Finally, interpolating between $\cW_{\frac12}$ and $\cW_{-\frac12}$, we obtain the \emph{dynamical space} $\cW_0$, which in addition to the scalar product has a natural pseudo-unitary structure given by the charge $Q$, and which is then used for quantization.

The operator $B$ can be interpreted as self-adjoint on all members of the scale $\cW_\lambda$.
Therefore, the dynamics is simply defined as $\e^{-\im tB}$ and preserves the scale $\cW_\lambda$.

Then we can define the propagators on the level of the Cauchy data as follows:
\begin{align*}
  E^\PJ(t,s)    & \defn \e^{-\im(t-s)B}, \\
  E^\vee(t,s)   & \defn \theta(t-s) \e^{-\im(t-s)B}, \\
  E^\wedge(t,s) & \defn -\theta(s-t) \e^{-\im(t-s)B}, \\
  E^{(+)}(t,s)  & \defn \e^{-\im(t-s)B}\one_{[0,\infty[}(B), \\
  E^{(-)}(t,s)  & \defn  \e^{-\im(t-s)B}\one_{[-\infty,0[}(B), \\
  E^\Feyn(t,s)  & \defn \e^{-\im(t-s)B}\bigl(\theta(t-s)\one_{[0,\infty[}(B) -\theta(s-t)\one_{]-\infty,0]}(B)\bigr), \\
  E^\aFeyn(t,s) & \defn \e^{-\im(t-s)B}\bigl(\theta(t-s)\one_{]-\infty,0]}(B) -\theta(s-t)\one_{[0,\infty[}(B)\bigr).
\end{align*}
At least formally, $E^\vee, E^\wedge, E^\Feyn, E^{\aFeyn}$ are inverses and $E^\PJ, E^{(+)}, E^{(-)}$ are bisolutions of $M$.

Then we set
\bes\label{pkc3}
\begin{align}
  G^\bullet & \defn \im E_{12}^\bullet,\quad \bullet=\PJ,\vee,\wedge,\Feyn,\aFeyn, \\
  G^{(+)} & \defn E_{12}^{(+)},\quad G^{(-)} \defn - E_{12}^{(-)},
\end{align}
\ees
(hence we use \eqref{rightupper} or its minor modifications) obtaining the generalizations of the propagators from the Minkowski space to the general stationary case.

Note that in the stationary case all the identities~\eqref{iden} still hold.

\subsection{Evolution on Hilbertizable spaces}

In the generic situation the generator $B(t)$ depends on time.
This leads both to technical and conceptual problems.

First, in order to do functional analysis we need topology.
However the Hilbert spaces $\cW_\lambda$ are no longer uniquely defined.
It seems reasonable to assume that Cauchy data are described by elements of a certain nested pair of \emph{Hilbertizable spaces} $\cW_1\subset\cW_0$, which does not change throughout the time.
(A Hilbertizable space is a space with a topology of a Hilbert space, but without a fixed scalar product.)

We devote the whole Sect.~\ref{sec:Evolutions} to a construction of cousins of all propagators described in Subsect.~\ref{Stationary case} in the setting of an evolution on Hilbertizable spaces (without assuming the existence of a charge form preserved by the dynamics).

The construction of the dynamics in the stationary case was straightforward.
Constructing the evolution generated by a time-dependent generator $B(t)$ is much more technical.
To this end we use an old result of Kato~\cite{kato:hyperbolic}.
In this approach one assumes that the Cauchy data are described by a nested pair of Hilbertizable spaces, and the generators are self-adjoint with respect to certain time-dependent scalar products compatible with both Hilbertizable structures.
Besides, one needs to make some technical assumptions, which essentially say that the generator of the evolution does not vary too much in time, so that all the time it acts in the same nested pair of Hilbertizable spaces.
Using this evolution it is easy to   define  $ E^\vee, E^\wedge, E^\PJ$,
which are the analogs of classical
propagators on the level of the Cauchy data operator.

In order to define ``non-classical'' propagators we need to choose the incoming and outgoing ``particle/antiparticle projections'', which as we discussed in Subsect.~\ref{Non-classical propagators on curved spacetimes} are determined by specifying involutions $S_\pm$.
This leads to a straightforward definition of ``in/out particle
bisolutions'' $E_\pm^{(+)}$ and ``in/out antiparticle bisolutions''
$E_\pm^{(-)}$,  which are two distinguished analogs of particle and antiparticle bisolutions.
(The plus/minus in the parentheses correspond to particles/antiparticles; the plus/minus without parentheses correspond to the future/past.)

What is more interesting, we can also try to define a natural Feynman and anti-Feynman propagator, denoted $E^\Feyn, E^{\aFeyn}$.
In the general Hilbertizable setting the existence of these propagators is not guaranteed and requires an extra condition that we call the \emph{asymptotic complementarity}.

\subsection{Pseudo-unitary dynamics}

As discussed above, the evolution of Cauchy data for the Klein--Gordon equation preserves the \emph{charge form} -- a natural Hermitian indefinite scalar product.
On the technical level it is convenient to assume that the charge form is compatible with the Hilbertizable structure.
More precisely, we need the structure of a \emph{Krein space}.

Note, in parenthesis, that we prefer to work in the complex setting of a Krein space instead of the real setting of a symplectic space, perhaps more common in the literature.
If the dynamics commutes with the complex conjugation, e.g., if there are no electromagnetic potentials, then by restricting our dynamics to the real space we can go back to the real symplectic setting.

In Sect.~\ref{Hilbertizable pseudo-unitary spaces} we discuss propagators in the context of a pseudo-unitary evolution on a Krein space.
We note an important property of Krein spaces: every maximally positive subspace is complementary to every maximally negative subspace.
By this property, if the boundary conditions are given by admissible involutions (see \eqref{admis}), then the condition of asymptotic complementarity is automatically satisfied.
Therefore the in-out Feynman and anti-Feynman propagator always exist.
The existence of these two propagators under rather general conditions is probably the main result of our paper.

To sum up, in the context of an evolution on Krein spaces we are able to define the whole family of ``propagators'' on the level of the Cauchy data operator:
\begin{equation}
  \label{proa1}
  E^\vee, E^\wedge, E^\PJ, E_\pm^{(+)}, E_\pm^{(-)}, E^\Feyn, E^{\aFeyn}.
\end{equation}

\subsection{Abstract Klein--Gordon operator}
\label{sub-Abstract Klein-Gordon operator}

Let $L(t)$, $\alpha(t)$ and $W(t)$ be time-dependent operators on a Hilbert space $\cK$.
We assume that $L(t)$ is positive, $\alpha(t)$ is positive and invertible, plus some additional technical assumptions.
By an \emph{abstract Klein--Gordon operator} we mean an operator of the form
\begin{equation}
  K \defn \bigl(D_t+W^*(t)\bigr)\frac{1}{\alpha^2(t)} \bigl(D_t+W(t)\bigr)-L(t),\label{kleingordon-}
\end{equation}
acting on the Hilbert space $L^2(I,\cK)\simeq L^2(I)\otimes\cK$.

In our applications, $\cK$ is the Hilbert space $L^2(\Sigma)$ (the space of functions on a spacelike Cauchy surface).
Besides, $\alpha^2(t)$ is related to the metric tensor, $W(t)$ consists mostly of $A^0$, and $L(t)$ is a magnetic Schr\"odinger operator on $\Sigma$.
The usual Klein--Gordon operator \eqref{eq:klein-gordon-half-} on the Hilbert space $L^2(M)\simeq L^2(I)\otimes L^2(\Sigma)$ has the form of \eqref{kleingordon-}, as discussed in Subsect.~\ref{Reduction to a 1st order equation for the Cauchy data}, see \eqref{hilbert1a}.

The operator (\ref{kleingordon-}) is second order in $t$.
It can be viewed as a 1-dimensional magnetic Schr\"odinger operator with operator-valued potentials.
To describe its Cauchy data, under our assumptions it is natural to introduce the scale of Hilbertizable spaces
\begin{equation}
  \cW_\lambda \defn L(t)^{-\frac\lambda2-\frac14}\cK\oplus L(t)^{-\frac\lambda2+\frac14}\cK,
\end{equation}
where usually $|\lambda|\leq\frac12$.
Note that $\cW_0$ has a natural Krein structure.
Using the formalism of Sect.~\ref{Hilbertizable pseudo-unitary spaces} we construct various propagators \eqref{proa1}.
Then, using a slight extension of \eqref{pkc3}, we pass from the Cauchy data propagators to spacetime propagators:
\begin{align*}
  G^\bullet & \defn \im E_{12}^\bullet,\quad \bullet=\PJ,\vee,\wedge,\Feyn,\aFeyn; \\ G_\pm^{(+)} & \defn E_{\pm,12}^{(+)},\quad G_\pm^{(-)} \defn - E_{\pm,12}^{(-)}.
\end{align*}

Feynman and anti-Feynman inverses of the Klein Gordon operator can be viewed as some special inverses of its well-posed realizations.

In Case (1) the Klein--Gordon operator is Hermitian (symmetric) but not self-adjoint.
It possesses many well-posed realizations defined by boundary conditions.
In particular, each Feynman-type and anti-Feynman-type inverse defines a certain well-posed realization.
Note that these realizations are always non-self-adjoint.

In Case (2) the situation is more difficult and not fully understood.
Clearly, the Feynman and anti-Feynman inverses are not bounded operators.
It is natural to conjecture that under quite general conditions they are boundary values of a certain distinguished self-adjoint realization of the abstract Klein--Gordon operator.
This conjecture has been partly proven in \cite{vasy:selfadjoint} and \cite{nakamura,nakamura2,nakamura3}.

\subsection{Bosonic quantization}
\label{Intro Bosonic quantization}

Let us now describe applications of the above mathematical analysis of the Klein--Gordon equation to Quantum Field Theory.

Various authors use different formalisms when introducing quantum fields.
These formalisms are essentially equivalent, however it may often be difficult for the reader to translate the concepts from one formalism to another.
Therefore, we start our discussion with some remarks about various approaches to quantization.
We restrict ourselves to linear bosonic theories.

Bosonic quantization can be divided into two steps:

\begin{enumerate}
  \item First we choose an algebra of observables satisfying canonical commutation relations corresponding to a classical phase space.
  \item Then we select a representation of this algebra on a Hilbert space.
\end{enumerate}

The first step can be presented in several ways, which superficially look differently.
In particular, we can use the real or complex formalism:
\begin{itemize}
  \item The \emph{real} or \emph{neutral} formalism starts from a real space equipped with an antisymmetric form (it does not have to be symplectic, that is, non-degenerate) and leads to a self-adjoint field $\hat\Phi$.
  \item The \emph{complex} or \emph{charged} formalism starts from a complex space equipped with a Hermitian form (sometimes called the charge) and leads to a pair of non-self-adjoint fields $\hat\Psi,\hat\Psi^*$.
\end{itemize}
The neutral formalism is in a sense more general, since every charged particle can be understood as a pair of neutral particles in the presence of a $U(1)$ symmetry.

In both the real and the complex approach, we can use the one-component formalism or the two-component formalism.
In the two-component formalism we split the fields into ``positions'' and ``momenta''.
This splitting is typical for Quantum Field Theory.

Thus we can distinguish four formalisms of bosonic quantization, which can be summarized in the following table:

\medskip

\begin{tabular}{llll} & Real (or neutral) fields & Complex (or charged) fields \\\hline \\ 1-component \quad & $[\hat\Phi(w), \hat\Phi(w')]=\im w\omega w'$ & $[\hat\Psi(w), \hat\Psi^*(w')]= \cinner{w}{Qw'}$\\ formalism\quad & $\omega$ is an antisymmetric form & $Q$ is a Hermitian form\\  & on a real space & on complex space\\\hline\\ 2-component\quad & $[\hat\phi(u),\hat\pi(v)]=\im \rinner{u}{v}$ & $[\hat\psi(u), \hat\eta^*(v)]=\im \cinner{u}{v}$\\ formalism\quad &  & $[\hat\psi^*(u), \hat\eta(v)]=\im \cinner{v}{u}$\\  & $\rinner{\cdot}{\cdot}$ is a bilinear scalar & $\cinner{\cdot}{\cdot}$ is a sesquilinear scalar & \\  & product on a real space & product on a complex space
\end{tabular}

\medskip

If the number of degrees of freedom is finite, by the Stone-von Neumann Theorem all irreducible representations of the CCR over a symplectic space are equivalent.
If the number of degrees of freedom is infinite this is not true, and we have to select a representation.
Usually this is done by fixing a state on the algebra of observables and going to the GNS representation.
In most applications to QFT one chooses a \emph{pure quasi-free state}, and then this representation naturally acts on a \emph{bosonic Fock space}.

There are several ways to describe pure quasi-free states (called also \emph{Fock states}).
As we mentioned above, in our paper these states are described by admissible involutions on the phase space, see \eqref{admis}.
In the real formalism one needs to assume in addition that $S_\bullet$ is \emph{anti-real} (i.e., $\conj{S_\bullet} = -S_\bullet$).

\subsection{Classical Field Theory on curved spacetimes}

Let us now briefly describe linear Classical Field Theory on curved spacetimes.
In the introduction, for brevity we restrict ourselves to complex scalar fields.
In Sects.~\ref{Bosonic quantization} and~\ref{Quantum Field Theory on curved space-times} we will consider also real scalar fields.

The phase space of our system is $\cW_\sc$ introduced in Subsect.~\ref{Pseudounitary structure}.
Then the \emph{field} $\psi$ and its complex conjugate $\psi^*$ are interpreted as the linear, resp.\ antilinear functional on $\cW_\sc$ given by
\begin{equation*}
  \rinner{\psi(x)}{u} \defn u(x),
  \quad
  \rinner{\psi^*(x)}{u} \defn \bar{u(x)},
  \quad
  u\in\cW_\sc.
\end{equation*}
Clearly, the fields $\psi$ and $\psi^*$ satisfy the Klein--Gordon equation:
\begin{equation}
  \label{KG1}
  K\psi(x)=K\psi^*(x)=0.
\end{equation}

As described in Subsect.~\ref{The Klein-Gordon equation}, the space $\cW_\sc$ is naturally a symplectic vector space.
The Poisson bracket of the fields was first computed by Peierls and can be expressed by the Pauli--Jordan propagator:
\begin{align*}
  \{\psi(x),\psi(y)\}   & = \{\psi^*(x),\psi^*(y)\}=0, \\
  \{\psi(x),\psi^*(y)\} & = - G^\PJ(x,y).
\end{align*}

\subsection{Quantum Field Theory on curved spacetimes}
\label{Propagators on curved spacetimes--introductory remarks}

Quantization of the classical theory defined by \eqref{KG1} is performed in two steps.
First we replace the classical fields $\psi(x),\psi^*(x)$ by $\hat\psi(x)$, $\hat\psi^*(x)$ interpreted as distributions with values in a $*$-algebra satisfying the following commutation relations, a generalization of the identity~\eqref{two1}:
\begin{align*}
  [\hat\psi(x),\hat\psi(y)] & =[\hat\psi^*(x),\hat\psi^*(y)]=0, \\
  [\hat\psi(x),\hat\psi^*(y)] & = -\im G^\PJ(x,y)\one.
\end{align*}
They also satisfy the Klein--Gordon equation
\begin{equation*}
  K\hat\psi(x)=K\hat\psi^*(x)=0.
\end{equation*}

The second step consists in choosing a representation of the CCR.

If the Klein--Gordon equation is stationary and stable, then there exists a natural Fock representation given by the so-called \emph{positive energy quantization} \cite{derezinski-siemssen:static,derezinski-gerard}.
As we discussed above, in the stationary case all propagators described for the Minkowski space have a natural generalization and the relations~\eqref{two} still hold.

If the Klein--Gordon equation is not necessarily stationary and we selected the involutions $S_-$ and $S_+$ (see Subsect.~\ref{Non-classical propagators on curved spacetimes}), we consider the Fock representations corresponding to $S_-$ and $S_+$, the so-called \emph{in} and \emph{out representations}.
As described above, we obtain two pairs of bisolutions $G_\pm^{(+)}$ and $G_\pm^{(-)}$, which after quantization describe the vacuum expectation values of products of fields in the in and out Fock representation.
More precisely, the identities~\eqref{two2a} and \eqref{two2b} split into two identities:
\bes
\begin{align}
  \label{two2a.}
  \cinner{\Omega_\pm}{\hat\psi(x)\hat\psi^*(y)\Omega_\pm} & = G_\pm^{(+)}(x,y),
  \\
  \label{two2b.}
  \cinner{\Omega_\pm}{\hat\psi^*(x)\hat\psi(y)\Omega_\pm} & = G_\pm^{(-)}(x,y).
\end{align}
\ees

After quantization the Feynman and anti-Feynman propagators satisfy slight modifications of the identities~\eqref{two3} and \eqref{two4}:
\bes\label{tw3}
\begin{align}
  \label{two3.}
  \frac{\cinner[\big]{\Omega_+}{\torder{\hat\psi(x)\hat\psi^*(y)}\Omega_-}}{\cinner{\Omega_+}{\Omega_-}}  & = -\im G^{\Feyn}(x,y),
  \\
  \label{two4.}
  \frac{\cinner[\big]{\Omega_-}{\atorder{\hat\psi(x)\hat\psi^*(y)}\Omega_+}}{\cinner{\Omega_-}{\Omega_+}} & = \im G^{\aFeyn}(x,y).
\end{align}
\ees

Strictly speaking, (\ref{tw3}) are true if the \emph{Shale condition}
for the in and out states holds, so that the vectors $\Omega_-$ and
$\Omega_+$ belong to the same representation of CCR.
If the Shale condition is violated, the left hand sides do not make sense.
However, the right hand sides are well defined.
This is an example of a renormalization in QFT: we are able to compute a quantity, which at the first sight is ill-defined.

In a large part of the literature the so-called Hadamard property is considered to be the main criterion for a physically satisfactory state \cite{radzikowski,fewster-verch:necessary,wald:backreaction}.
Note that if one assumes enough smoothness, then in Case 2 the two-point functions $G_\pm^{(+)}(x,y),G_\pm^{(-)}(x,y)$ automatically satisfy the Hadamard property.
This is a non-trivial fact proven by Gérard--Wrochna \cite{gerard-wrochna:inout}, see also \cite{fulling-narcowich-wald}.

In the slab geometry case we can choose non-Hadamard states for $\Omega_-$ and $\Omega_+$ if we insist.
However, there are good arguments saying that Hadamard states are ``more physical'' than the others.
Actually, one could argue that the main argument for the Hadamard condition in Case (1) is the fact that by \cite{gerard-wrochna:inout} it is automatic in the Case (2) scenario.

\subsection{Hadamardists and Feynmanists}

Oversimplifying and exaggerating, one can distinguish two approaches to QFT on curved spacetimes: let us call them the \emph{Hadamardist} and \emph{Feynmanist} approach.

The Hadamardist approach is mostly represented by researchers with a mathematical or General Relativity background.
It stresses that QFT should be considered in a local fashion, often restricting attention to a small causally convex region of a spacetime.
In such a setting it is impossible to choose a distinguished state $\Omega$ on the algebra of observables in a locally covariant way, see e.g.~\cite{fewster-verch:spass,hollands}.
This approach stresses that one has a lot of freedom in choosing a state and argues that physical states should satisfy the so-called \emph{Hadamard property}.
In the modern mathematical literature this condition is usually described in an elegant but abstract way with the help of the \emph{wave front set} \cite{radzikowski} of their two-point function $\cinner[\big]{\Omega}{\hat\psi(x)\hat\psi^*(y) \Omega}$.
Using this two-point function one can define the (formal) $*$-algebra of observables and perturbatively renormalize local polynomials of fields.
A particularly clear explanation of this philosophy can be found in Apps.~B and~D1 of \cite{hollands}.

In the Feynmanist approach the main goal is usually to compute scattering amplitudes, cross sections, etc. see e.g.~\cite{dewitt:curved}.
Such computations are typically based on path integrals and Feynman diagrams.
Clearly, in this case it is indispensable to look at the spacetime globally, so that one can define an in and out state, and the assumption of asymptotic stability and stationarity in the past and future is rather natural.
There is no need to worry about the Hadamard property. As we mentioned
above, it is automatic under rather weak assumptions thanks to the results of Fulling--Narcowich--Wald \cite{fulling-narcowich-wald} and Gérard--Wrochna \cite{gerard-wrochna:inout}.
The (in-out) Feynman propagator is needed to compute perturbatively the scattering operator in terms of Feynman diagrams.

There is no contradiction between these two approaches and both have their philosophical merits.
Our paper clearly belongs to the latter approach (in spite of being rather abstract mathematically).

Actually, in Case (1) (the slab geometry case) it is natural to use the hybrid point of view, which reconciles the Hadamardist and Feynmanist philosophy.
If we want to compute the scattering operator between time $t_-$ and $t_+$ it is natural to choose $S_-$ and $S_+$ both satisfying the Hadamard condition, and then to use the corresponding (in-out) Feynman propagator.

\subsection{Comparison with literature}

The basic formalism of pseudo-unitary (or symplectic) evolution equations described in this paper is of course contained more or less explicitly in all works on  Quantum Field Theory in curved spacetime, including the standard textbooks \cite{birrell-davies,fulling,wald,parker-toms}.
Surprisingly, however, its point of view is rarely fully exploited.

The (trivial) observation about the existence of two distinguished states on asymptotically stationary spacetimes can be found e.g.\ in~\cite{birrell-davies}, Sect.~3.3. It is rather obvious that they are the preferred states for many actual applications, such as the calculation of the scattering operator.

In the well-known paper \cite{duistermaat} Duistermat and H\"ormander prove the existence of the \emph{Feynman parametrix}, which is unique only up to a smoothing operator.
However, the canonical \emph{in-out Feynman inverse} defined rigorously in our paper is essentially absent from the mathematical literature, with a few recent exceptions~\cite{gerard-wrochna:feynman,vasy:selfadjoint,nakamura,nakamura2,nakamura3}.

Sometimes one considers another generalization of the Feynman propagator: given a Hadamard state with the two-point functions $G_\bullet^{(+/-)}$, one can introduce
\begin{equation}
  \label{idi5.}
  G_\bullet^\Feyn = \im G_\bullet^{(+)} + G^\wedge = -\im G_\bullet^{(-)} + G^\vee,
\end{equation}
which is an inverse of the Klein Gordon operator, and can be called the \emph{Feynman inverse associated to the pair of two-point functions $G_\bullet^{(+)}$, $G_\bullet^{(+)}$}, see e.g.\ \cite{derezinski-siemssen:propagators} and \cite{hollands} App.~D1.
Note, however, that $G_\bullet^\Feyn$ is non-unique and, more importantly, does not satisfy the relation~\eqref{two3.}, which is the basis for perturbative calculations of the scattering operator.

In the more physically oriented literature the in-out Feynman propagator, defined rigorously in our paper, is ubiquitous, even if implicitly.
It essentially appears each time when the functional integration method is applied to compute scattering amplitudes, at least on asymptotically stationary spacetimes.
More precisely, in order to compute perturbatively Feynman diagrams for the scattering operator one needs to associate the Feynman propagator to each line.

It seems that this point is not sufficiently appreciated in a part of mathematically oriented literature.
Let us quote from App.~B of \cite{hollands}, which as we mentioned above, expresses the Hadamardist philosophy:
``[the effective action] depends upon a choice of state [\textellipsis].
Here, the choice of state would enter the precise choice of the formal path-integral measure $[\mathcal{D}\phi]$.''
In reality, typically the path-integral formalism yields a unique prescription for Feynman diagrams (which then need to be renormalized).
This prescription naturally involves two states: the ``in state'' and the ``out state''.
It also determines uniquely the Feynman propagator, which should be associated to each line of the Feynman diagram.

The formula \eqref{two3.}, which gives a physical meaning to the in-out Feynman propagator, can be found in the physical literature in various places, see, for instance, Equation (4.7) of \cite{fulling} and the following equations.

Until quite recently, the question of the self-adjointness of the Klein--Gordon operator was almost absent in the mathematical literature.
There were probably two reasons for this.
First, the question seemed difficult.
The Klein--Gordon operator is unbounded from below and above, and the positivity is usually the major tool in functional analysis.
The second reason is that this question at the first sight appeared not interesting physically.
Indeed, the space $L^2(M)$ is not the space of states of any reasonable quantum system and $\e^{-\ri \tau K}$ seems to have no meaning as a quantum evolution.

However, in physical literature many researchers tacitly assume that the Klein--Gordon operator is self-adjoint, see e.g.\ \cite{rumpf1} or Sect.~9 of \cite{fulling}.
The Feynman and anti-Feynman propagator are important for applications and, formally, assuming (\ref{KG2}), they can be computed from
\begin{align}
  \label{KG2-}
   & (K\pm\im0)^{-1} =\pm \im\int_0^\infty \e^{\mp\im t K} \dif t,
\end{align}
Note that (\ref{KG2-}) presupposes that one can interpret $K$ as a self-adjoint operator, so that $(z-K)^{-1}$ can be defined for $z\nin\RR$.

\section{Preliminaries}
\label{sec:Preliminaries}

In this section we collect various basic mathematical definitions and facts that are useful in our paper.
Most readers will find them rather obvious -- nevertheless, they should be recorded.

If $A$ is an operator, then $\Ran(A)$, $\Ker(A)$, $\cD(A)$ and $\sigma(A)$ denote the range, the nullspace, the domain and the spectrum of $A$.
$B(\cW)$ denotes the space of bounded operators on a Banach space $\cW$.

\subsection{Scales of Hilbert spaces}
\label{sub:Scales of Hilbert spaces}

Suppose that $\cW$ is a Hilbert space and $A$ a positive invertible operator on $\cW$.
Then one defines $A^{-\alpha} \cW$ as the domain of $A^\alpha$ for $\alpha\geq0$ and as its anti-dual for $\alpha < 0$.
We thus obtain a \emph{scale of nested Hilbert spaces} $A^{-\alpha} \cW$, $\alpha \in \RR$, with $A^0 \cW = \cW$ and $A^{-\alpha} \cW$ continuously and densely embedded in $A^{-\beta}\cW$ for $\alpha \geq \beta$.
By restriction/extension, the operator $A^{\beta}$ can be interpreted as a unitary from $A^{-\alpha}\cW$ to $A^{-\alpha+\beta} \cW$.
Often we simplify notation by writing $\cW_\alpha$ for $A^{-\alpha} \cW$, so that $\cW_0 = \cW$ and $\cW_1 = \Dom(A)$.

In practice, the starting point of a construction of a scale of Hilbert spaces is often not an operator $A$ but a nested pair of Hilbert spaces.
More precisely, suppose that $(\cW,\cV)$ is a pair of Hilbert spaces, where $\cV$ is densely and continuously embedded in $\cW$.
Then there exists a unique invertible positive self-adjoint operator $A$ on $\cW$ with the domain $\cV$ such that
\begin{equation*}
  \cinner{v}{v}_{\cV} = \cinner{Av}{Av}_{\cW}, \quad v \in \cV.
\end{equation*}
We can then define the scale $\cW_\alpha = A^{-\alpha} \cW, \alpha \in \RR$.
Note that $\cW = \cW_0$ and $\cV = \cW_1$.

We will often use the following facts:
\begin{proposition}  Consider a scale of Hilbert spaces
    $\cW_\alpha$, $\alpha\in\mathbb{R}$.
  \mbox{}
  \begin{enumerate}
    \item $A\in B(\cW_{\alpha},\cW_{\beta})$ implies $A^*\in B(\cW_{-\beta},\cW_{-\alpha})$.
    \item Let $\alpha_0\leq\alpha_1$, $\beta_0\leq\beta_1$.
          If $A\in B(\cW_{\alpha_0},\cW_{\beta_0})$ can be restricted to an operator in $ B(\cW_{\alpha_1},\cW_{\beta_1})$, then for $\tau\in[0,1]$ it can be restricted to an operator in $B(\cW_{(1-\tau)\alpha_0+\tau\alpha_1},\cW_{(1-\tau)\beta_0+\tau\beta_1})$.
  \end{enumerate}
  \label{interpolation}
\end{proposition}

\subsection{One-parameter groups}
\label{sub:One-parameter groups}

Let $\cW$ be a Banach space.
Recall that a \emph{one-parameter group on $\cW$} is a homomorphism
\begin{equation*}
  \RR \ni t \mapsto R(t) \in B(\cW).
\end{equation*}
It is well-known that to every strongly continuous one-parameter group $R(t)$ one can uniquely associate a densely defined operator $B$ called the \emph{generator of $R(t)$}, so that $R(t) = \e^{-\im t B}$.
It can be shown that $\Dom(B)$ is preserved by $R(t)$ and the following equation is true:
\begin{equation}
  \label{evol}
  (\partial_t + \im B) R(t) w = 0,
  \quad
  w \in \Dom(B).
\end{equation}

Suppose now that $\cW$ is a Hilbert space.
A unitary group on $\cW$ is always of the form $\e^{-\im t B}$, where $B$ is a self-adjoint operator.
Let $\langle B \rangle \defn (B^2+1)^{\frac12}$.
Note that $\e^{-\im t B}$ preserves the scale $\cW_\alpha \defn \langle B \rangle^{-\alpha} \cW$ for $\alpha \geq 0$, and can be uniquely extended by continuity to~$\cW_\alpha$ for $\alpha \leq 0$.
For any~$\alpha$ we have
\begin{equation}
  \label{eq:evo-generator}
  (\partial_t + \im B) \e^{-\im t B} w = 0,
  \quad
  w \in \cW_{1+\alpha},
\end{equation}
where the left-hand side of~\eqref{eq:evo-generator} is understood as an element of~$\cW_{\alpha}$.

In practice, two choices of $\alpha$ are especially useful: $\alpha = 0$ corresponds precisely to \eqref{evol}, and $\alpha = -\frac12$ means that \eqref{eq:evo-generator} is considered on the form domain of $B$.
We will use both points of view when considering a natural setup for non-autonomous evolutions, see Thm.~\ref{thm:evolution}.

If in addition $B$ is invertible, then we can slightly modify the scale of Hilbert spaces by setting $\cW_\alpha \defn \abs{B}^{-\alpha} \cW$.
Note that $B$ is then unitary from $\cW_\alpha$ to $\cW_{\alpha-1}$.

\subsection{Hilbertizable spaces}
\label{sub:Hilbertizable spaces}

\begin{definition}
  \label{def:Hilbertizable}
  Let $\cW$ be a complex\footnote{Analogous definitions and results are valid for \emph{real} Hilbertizable spaces.
  } topological vector space.
  We say that it is \emph{Hilbertizable} if it has the topology of a Hilbert space for some scalar product $\cinner{\cdot}{\cdot}_\bullet$ on $\cW$.
  We will then say that $\cinner{\cdot}{\cdot}_\bullet$ is \emph{compatible with (the Hilbertizable structure of)} $\cW$.
  The subscript $\bullet$ serves as a placeholder for a name of a scalar product.
  The Hilbert space $\bigl( \cW, \cinner{\cdot}{\cdot}_\bullet \bigr)$ will be occasionally denoted $\cW_\bullet$.
  The corresponding norm will be denoted $\norm{\,\cdot\,}_\bullet$.
\end{definition}

In what follows $\cW$ is a Hilbertizable space.
Let $\cinner{\cdot}{\cdot}_1$, $\cinner{\cdot}{\cdot}_2$ be two scalar products compatible with $\cW$.
Then there exist constants $0 < c \leq C$ such that
\begin{equation*}
  c \cinner{w}{w}_1 \leq \cinner{w}{w}_2 \leq C \cinner{w}{w}_1.
\end{equation*}

Let $R$ be a linear operator on $\cW$.
We say that it is \emph{bounded} if for some (hence for all) compatible scalar products $\cinner{\cdot}{\cdot}_\bullet$ there exists a constant $C_\bullet$ such that
\begin{equation*}
  \norm{Rw}_\bullet \leq C_\bullet \norm{w}_\bullet.
\end{equation*}

Suppose that $A$ is a (densely defined) operator on $\cW$.
We say that it is \emph{similar to self-adjoint} if there exists a compatible scalar product $\cinner{\cdot}{\cdot}_\bullet$ such that $A$ is self-adjoint with respect to $\cinner{\cdot}{\cdot}_\bullet$.
Note that for such operators the spectral theorem  can be applied.
In particular, for any (complex-valued) Borel function $f$ on the spectrum of $A$ we can define $f(A)$.

Let $Q$ be a sesquilinear form on $\cW$.
We say that it is \emph{bounded} if for some (hence for all) compatible scalar products $\cinner{\cdot}{\cdot}_\bullet$ there exists $C_\bullet$ such that
\begin{equation*}
  \abs{\cinner{v}{Qw}} \leq C_\bullet \norm{v}_\bullet \norm{w}_\bullet, \quad v,w \in \cW.
\end{equation*}
Note that on Hilbertizable spaces we do not have a natural identification of sesquilinear forms with operators.

\subsection{Interpolation between Hilbertizable spaces}
\label{sub:Interpolation between Hibertizable spaces}

\begin{definition}
  \label{def:Hilbertizable pair}
  A pair of Hilbertizable spaces $(\cW, \cV)$, where $\cV$ is densely and continuously embedded in $\cW$, will be called a \emph{nested Hilbertizable pair}.
\end{definition}

After fixing scalar products $\cinner{\cdot}{\cdot}_{\cV,\bullet}$ and $\cinner{\cdot}{\cdot}_{\cW,\bullet}$ compatible with $\cV$, resp.\ $\cW$, we can interpolate between the Hilbert spaces $\cV_\bullet$ and $\cW_\bullet$ obtaining a scale of Hilbert spaces $\cW_{\alpha,\bullet}$, $\alpha \in \RR$, with $\cV_\bullet = \cW_{1,\bullet}$ and $\cW = \cW_{0,\bullet}$.
By  complex interpolation,
for $\alpha \in [0,1]$ they do not depend on the choice of scalar products $\cinner{\cdot}{\cdot}_{\cW,\bullet}$ and $\cinner{\cdot}{\cdot}_{\cV,\bullet}$ as Hilbertizable spaces.
Therefore, the family of Hilbertizable spaces $\cW_\alpha$, $\alpha \in [0,1]$, is uniquely defined.

If $R \in B(\cW)$ and its restriction to $\cV$ belongs to $B(\cV)$, then $R$ restricts to $B(\cW_\alpha)$ for $0 \leq \alpha \leq 1$.

\subsection{From complex to real spaces and back}
\label{sub:From complex to real spaces and back}

To pass from a complex space to a real one, it is useful to have the notion of a conjugation:

Let $\cW$ be a complex space.
An antilinear involution $v \mapsto \bar v$ on $\cW$ will be called a \emph{conjugation}.
In the context of Hilbertizable spaces we always assume that conjugations are bounded.
For an operator $R$ on $\cW$ we set
\begin{equation*}
  \bar R v \defn \bar{R \bar v}, \quad R^\T \defn \bar R^*.
\end{equation*}
If $R$ satisfies $\bar{R} = \pm R$, it will be called \emph{real} resp.\ \emph{anti-real}.
The \emph{real subspace of $\cW$} is defined as
\begin{equation}
  \label{real}
  \cW_\RR \defn \{w \in \cW \mid \bar w = w \}.
\end{equation}

Conversely, to pass from a real space to a complex one, suppose now that $\cY$ is a real space.
Then $\cY \otimes_\RR \CC = \CC \cY$ will denote the \emph{complexification} of~$\cY$ (i.e., for every $w \in \cW$ we can write $w = w_R + \im w_I$ with $w_R, w_I \in \cY$), and we have the natural conjugation $\bar{v_R + \im v_I} = v_R - \im v_I$.

\subsection{Complexification of (anti-)symmetric forms}
\label{sub:Complexification of bilinear forms}

Let $\cY$ be a real space.

Every symmetric form $q$ on $\cY$, and thus in particular every scalar product, extends to a Hermitian form on $\CC\cY$:
\begin{equation*}
  \begin{split}
    \cinner[\big]{v_R + \im v_I}{q (w_R + \im w_I)} &\defn \rinner{v_R}{q w_R} + \rinner{v_I}{q w_I} \\&\quad - \im\rinner{v_I}{q w_R} + \im\rinner{v_R}{q w_I}.
  \end{split}
\end{equation*}
Note the property $\bar{\cinner{v}{qw}} = \cinner{\bar v}{q\bar w}$.

Extending an antisymmetric form $\omega$ on $\cY$ to a Hermitian form on $\CC\cY$ is slightly different:
\begin{equation}
  \label{qiu}
  \begin{split}
    \cinner[\big]{v_R + \im v_I}{Q (w_R + \im w_I)} &\defn \rinner{v_I}{\omega w_R} - \rinner{v_R}{\omega w_I} \\&\quad + \im\rinner{v_R}{\omega w_R} + \im\rinner{v_I}{\omega w_I}.
  \end{split}
\end{equation}
Note the property $\bar{\cinner{v}{Qw}} = -\cinner{\bar v}{Q\bar w}$, which differs from the symmetric case above.

\subsection{Realification of Hermitian forms}
\label{sub:Realification of Hermitian forms}

Let $Q$ be a Hermitian form on $\cW$.

We say that a conjugation $\bar{\,\cdot\vphantom{1}\,}$ \emph{preserves} $Q$ if
\begin{equation*}
  \bar{\cinner{v}{Q w}} = \cinner{\bar v}{Q \bar w}.
\end{equation*}
In that case,
\begin{equation*}
  \Re{} \cinner{v}{Qw}, \quad v,w \in \cW_\RR,
\end{equation*}
is a symmetric form on $\cW_\RR$.
Note that $\Im{} \cinner{v}{Qw} = 0$ on $\cW_\RR$.

Similarly, we say that a conjugation $\bar{\,\cdot\vphantom{1}\,}$ \emph{anti-preserves} $Q$ if
\begin{equation}
  \label{antipreserves}
  \bar{\cinner{v}{Q w}} = -\cinner{\bar v}{Q \bar w}.
\end{equation}
In that case,
\begin{equation*}
  \Im{} \cinner{v}{Qw}, \quad v,w \in \cW_\RR,
\end{equation*}
is an antisymmetric form on $\cW_\RR$.
Note that $\Re{}\cinner{v}{Qw} = 0$ on $\cW_\RR$.

\subsection{Involutions}
\label{sub:Involutions}

\begin{definition}
  \label{def:complementary}
  We say that a pair $(\cZ_\bullet^{(+)},\cZ_\bullet^{(-)})$ of subspaces of a vector space $\cW$ is \emph{complementary} if
  \begin{equation*}
    \cZ_\bullet^{(+)} \cap \cZ_\bullet^{(-)} = \{0\},
    \quad
    \cZ_\bullet^{(+)}+\cZ_\bullet^{(-)} = \cW.
  \end{equation*}
\end{definition}

\begin{definition}
  An operator $S_\bullet$ on $\cW$ is called an \emph{involution}, if $S_\bullet^2 = \one$.
\end{definition}
We can associate various objects with $S_\bullet$:
\begin{equation}
  \label{projo-}
  \Pi^{(\pm)}_\bullet \defn \frac12 (\one \pm S_\bullet),
  \quad
  \cZ^{(\pm)}_\bullet \defn \Ran(\Pi^{(\pm)}_\bullet).
\end{equation}
$(\Pi_\bullet^{(+)}, \Pi_\bullet^{(-)})$ is a pair of complementary projections and $(\cZ_\bullet^{(+)},\cZ_\bullet^{(-)})$ is the corresponding pair of complementary subspaces.

A possible name for $\cZ_\bullet^{(+)}$ is the \emph{positive space}, and for $\cZ_\bullet^{(-)}$ is the \emph{negative space} (associated with~$S_\bullet$).
We will however prefer names suggested by QFT: $\cZ_\bullet^{(+)}$ will be called the \emph{particle space}, and $\cZ_\bullet^{(-)}$ the \emph{antiparticle space}.

If $\cW$ is Hilbertizable, we will usually assume that $S_\bullet$ is bounded.
Then so are $\Pi_\bullet^{(+)}$ and $\Pi_\bullet^{(-)}$, moreover, $\cZ_\bullet^{(+)}$ and $\cZ_\bullet^{(-)}$ are closed.

\subsection{Pairs of involutions}

Suppose that $S_1$ and $S_2$ are two bounded involutions on $\cW$.
Let
\begin{align*}
  \Pi_i^{(\pm)} \defn \frac12 (\one \pm S_i),\quad \cZ_i^{(\pm)} \defn \Ran(\Pi_i^{(\pm)}),\quad i=1,2,
\end{align*}
be the corresponding pairs of complementary projections and subspaces.
The following operator can be defined by many distinct expressions:
\begin{subequations}
  \label{eq:Upsilon_operator}
  \begin{align}
    \Upsilon
    & \defn \one - (\Pi_1^{(+)} - \Pi_2^{(+)})^2
    = \one - (\Pi_1^{(-)} - \Pi_2^{(-)})^2 \label{eq:Upsilon_operator1}
    \\
    & \mathrel{\phantom{\defn}\mathllap{=}} \Pi_1^{(+)} \Pi_2^{(+)} + \Pi_2^{(-)} \Pi_1^{(-)}
    = \Pi_2^{(+)} \Pi_1^{(+)} + \Pi_1^{(-)} \Pi_2^{(-)} \label{eq:Upsilon_operator2}
    \\
    & \mathrel{\phantom{\defn}\mathllap{=}} (\Pi_1^{(+)} + \Pi_2^{(-)}) (\Pi_2^{(+)} + \Pi_1^{(-)})
    = (\Pi_2^{(+)} + \Pi_1^{(-)}) (\Pi_1^{(+)} + \Pi_2^{(-)}) \label{eq:Upsilon_operator3}
    \\
    \begin{split}
      & \mathrel{\phantom{\defn}\mathllap{=}} (\Pi_1^{(+)} \Pi_2^{(+)} + \Pi_2^{(-)}) (\Pi_2^{(+)} \Pi_1^{(+)} + \Pi_1^{(-)})
      \\
      & \mathrel{\phantom{\defn}\mathllap{=}} (\Pi_2^{(+)} \Pi_1^{(+)} + \Pi_1^{(-)}) (\Pi_1^{(+)} \Pi_2^{(+)} + \Pi_2^{(-)})
    \end{split}
    \label{eq:Upsilon_operator4}
    \\
    & \mathrel{\phantom{\defn}\mathllap{=}} \frac14(2+S_1S_2+S_2S_1)=\frac14(S_1+S_2)^2.
  \end{align}
\end{subequations}
Observe that $\Upsilon$ commutes with $\Pi_1^{(+)}$, $\Pi_1^{(-)}$, $\Pi_2^{(+)}$ and $\Pi_2^{(-)}$.

\begin{proposition}
  \label{prop:pair_proj}
  The following conditions are equivalent:
  \begin{enumerate}
    \item[(i)] $\Upsilon$ is invertible.
    \item[(ii)] $\Pi_1^{(+)} + \Pi_2^{(-)}$ and $\Pi_2^{(+)} + \Pi_1^{(-)}$ are  invertible.
    \item[(iii)] $\Pi_1^{(+)} \Pi_2^{(+)} + \Pi_2^{(-)}$ and $\Pi_2^{(+)} \Pi_1^{(+)} + \Pi_1^{(-)}$ are invertible.
  \end{enumerate}
  Moreover, if one of the above holds, then the pairs $(\cZ_1^{(+)},\cZ_2^{(-)})$ as well as $(\cZ_2^{(+)},\cZ_1^{(-)})$ are complementary.
\end{proposition}
\begin{proof}
  (i)$\iff$(ii) and (i)$\iff$(iii) follow from~\eqref{eq:Upsilon_operator3} and~\eqref{eq:Upsilon_operator4} by the following  easy fact:
  If $R,S,T$ are maps such that $R = ST = TS$, then
  $R$ is bijective if and only if both $T$ and $S$ are bijective.

  The last implication follows from the next proposition.
\end{proof}

In the setting of the above proposition we can use $\Upsilon$ to construct two pairs of complementary projections:
\begin{proposition}
  \label{prop:compl-proj}
  Suppose that $\Upsilon$ is invertible.
  Then
  \begin{alignat*}{2} \Lambda_{12}^{(+)} & \defn \Pi_1^{(+)} \Upsilon^{-1} \Pi_2^{(+)} & \quad & \text{is the projection onto $\cZ_1^{(+)}$ along $\cZ_2^{(-)}$}, \\ \Lambda_{21}^{(-)} & \defn \Pi_2^{(-)} \Upsilon^{-1} \Pi_1^{(-)} & \quad & \text{is the projection onto $\cZ_2^{(-)}$ along $\cZ_1^{(+)}$}, \\ \Lambda_{21}^{(+)} & \defn \Pi_2^{(+)} \Upsilon^{-1} \Pi_1^{(+)} & \quad & \text{is the projection onto $\cZ_2^{(+)}$ along $\cZ_1^{(-)}$}, \\ \Lambda_{12}^{(-)} & \defn \Pi_1^{(-)} \Upsilon^{-1} \Pi_2^{(-)} & \quad & \text{is the projection onto $\cZ_1^{(-)}$ along $\cZ_2^{(+)}$}.
  \end{alignat*}

  In particular,
  \begin{equation*}
    \Lambda_{12}^{(+)} + \Lambda_{21}^{(-)} = \one,
    \quad
    \Lambda_{21}^{(+)} + \Lambda_{12}^{(-)} = \one.
  \end{equation*}
\end{proposition}
\begin{proof}
  First we check that $\Lambda_{12}^{(+)}$ is a projection:
  \begin{align*}
    \bigl( \Lambda_{12}^{(+)} \bigr)^2 & = \Pi_1^{(+)} \Upsilon^{-1} \Pi_2^{(+)} \Pi_1^{(+)} \Upsilon^{-1} \Pi_2^{(+)} \\ & = \Pi_1^{(+)} \Upsilon^{-1} (\Pi_2^{(+)} \Pi_1^{(+)} + \Pi_1^{(-)} \Pi_2^{(-)}) \Upsilon^{-1} \Pi_2^{(+)} = \Lambda_{12}^{(+)}.
  \end{align*}
  Moreover,
  \begin{align*}
    \Lambda_{12}^{(+)} & = \Pi_1^{(+)} (\Pi_2^{(+)}+\Pi_1^{(-)}) \Upsilon^{-1}=\Upsilon^{-1} (\Pi_1^{(+)}+\Pi_2^{(-)}) \Pi_2^{(+)}.
  \end{align*}
  But $ (\Pi_2^{(+)}+\Pi_1^{(-)}) \Upsilon^{-1}$ and $\Upsilon^{-1} (\Pi_1^{(+)}+\Pi_2^{(-)})$ are invertible.
  Hence $\Ran(\Lambda_{12}^{(+)})=\Ran (\Pi_1^{(+)})$ and $\Ker(\Lambda_{12}^{(+)}) =\Ker(\Pi_2^{(+)})=\Ran(\Pi_2^{(-)})$.
  This proves the statement of the proposition about $\Lambda_{12}^{(+)}$.
  The remaining statements are proven analogously.
\end{proof}

\section{Evolutions on Hilbertizable spaces}
\label{sec:Evolutions}

In this section we investigate the concept of an evolution (family) in the Hilbertizable setting and without the additional pseudo-unitary structure, which will be added in later sections.
Already in the present setting we can try to define abstract versions of ``forward/backward'', ``Pauli--Jordan'', ``particle/antiparticle'', ``Feynman/anti-Feynman'' propagators and derive their basic properties.
The existence of the abstract version of the Feynman propagator (inverse) will depend on a certain property called ``asymptotic complementarity'', which in general is not guaranteed to hold.
In the next section we will see that this property automatically holds under some natural assumptions typical of QFT.

Throughout the section, $-\infty \leq t_- < t_+ \leq +\infty$.
For brevity, we write $I \defn \mathopen{]}t_-,t_+\mathclose{[}$.
Without limiting the generality, we will always assume that $0\in I$.
Moreover, we will always assume that either $t_\pm$ are both finite or both infinite.
The case where only one of $t_\pm$ is infinite can be deduced from the other cases.

In typical situations to define an evolution one first introduces a time-dependent family of operators $t\mapsto-\im B(t)$ that generate $R(t,s)$.
The necessary and sufficient conditions for an operator $B$ to generate an \emph{autonomous} evolution or, what is equivalent, a strongly continuous one-parameter group are well-known and relatively simple.
In the \emph{non-autonomous} case, there exist various relatively complicated theorems describing sufficient conditions.
Unfortunately, it seems that a complete theory on this subject is not available.

In the first part of this section we will avoid discussing the topic of generators of a non-autonomous evolutions apart from heuristic remarks.
We will treat the evolution as given.
The \emph{Cauchy data operator} $M \defn \partial_t+\ri B(t)$ will be just a heuristic concept, without a rigorous meaning.
However, ``bisolutions'' and ``inverses of $M$'' will be rigorously defined.
They will be the main topic of this section.

In Subsect.~\ref{Almost unitary evolutions on Hilbertizable spaces} we describe a possible approach to the generation of non-autonomous evolutions based on a theorem of Kato.

\subsection{Concept of an evolution}

\begin{definition}
  \label{def:evolution}
  Let $\cW$ be a Banach space.
  We say that the two-parameter family
  \begin{equation}
    \label{eq:evolution}
    I \times I \ni (t,s) \mapsto R(t,s) \in B(\cW)
  \end{equation}
  is a \emph{strongly continuous evolution (family) on $\cW$} if for all $r,s,t \in I$, we have the identities
  \begin{equation}
    \label{eq:evo_idents}
    R(t,t) = \one, \quad R(t,s) R(s,r) = R(t,r),
  \end{equation}
  and the map~\eqref{eq:evolution} is strongly continuous.
\end{definition}

One can also consider evolutions parametrized by the closed interval $I^\cl \defn [t_-,t_+]$ instead of~$I$, with the obvious changes in the definition.

Note also that Def.~\ref{def:evolution} involves both forward and backward evolution, since we do not assume $t \geq s$ in $R(t,s)$.
In other words, this definition is a generalization of a one-parameter group instead of a one-parameter semigroup.

\subsection{Generators of evolution}
\label{sub:Generators of evolution}

Until the end of this section we consider a strongly continuous evolution $R(t,s)$, $t,s \in I$, on a Hilbertizable space $\cW$.

If $R(t,s) = R(t-s,0)$ for all $t,s,t-s,0 \in I$, we say that the evolution is \emph{autonomous}.
An autonomous evolution can always be extended to~$\RR\times\RR$ in the obvious way.
Setting $R(t) \defn R(t,0)$, we obtain a \emph{strongly continuous one-parameter group}.
As we have already mentioned, we can then write $R(t) = \e^{-\im t B}$, where~$B$ is a certain unique, densely defined, closed operator called the \emph{generator of~$R$}.

For non-autonomous evolutions, the concept of a generator is understood only under some special assumptions.
Heuristically, the operator-valued function $I \ni t \mapsto B(t)$ is called the \emph{(time-dependent) generator of $R(t,s)$} if
\begin{equation*}
  B(t) \defn \bigl(\im\partial_t R(t,s) \bigr) R(s,t).
\end{equation*}
Note that the evolution should satisfy in some sense
\begin{align*}
   \im\partial_t R(t,s) v & = B(t) R(t,s) v, \\
  -\im\partial_s R(t,s) v & = R(t,s) B(s) v.
\end{align*}

A possible rigorous meaning of the concept of a time-dependent generator will be discussed in Subsect.~\ref{Almost unitary evolutions on Hilbertizable spaces}.

\subsection{Bisolutions and inverses of the Cauchy data operator}
\label{sub:Bisolutions and inverses of the evolution equation}

Introducing the (heuristic) generator $B(t)$, we can consider the (still heuristic) \emph{Cauchy data operator}
\begin{equation*}
  M \defn \partial_t + \im B(t).
\end{equation*}
 Let $C_\mathrm{c}(I,\cW)$ denote
  continuous compactly supported functions from the open interval $I$ to $\cW$.
We will say that an operator $E$ is a \emph{bisolution} resp.\ an \emph{inverse or Green's operator of $M$} if it is maps  $C_\mathrm{c}(I,\cW)\to C(I,\cW)$  and satisfies (heuristically)
\begin{subequations}
  \label{eq:bisolution-inverse}
  \begin{alignat}{2}
    \bigl(\partial_t + \im B(t)\bigr) E w & = 0,
                                          & \quad
    E \bigl(\partial_t + \im B(t)\bigr) v & = 0, \label{eq:bisolution}
    \\\text{resp.}\quad
    \bigl(\partial_t + \im B(t)\bigr) E w & = w,
                                          & \quad
    E \bigl(\partial_t + \im B(t)\bigr) v & = v. \label{eq:inverse}
  \end{alignat}
\end{subequations}
 for all  $w\in C_\mathrm{c}(I,\cW)$ and for $v$ in an
  appropriate domain inside $C(I,\cW)$. (Note that $M$ is an example of an
  unbounded operator,  hence problems with its domain are not surprising.)
A possible rigorous version of (\ref{eq:bisolution}) and (\ref{eq:inverse}) will be given in Subsect.~\ref{Rigorous concept of a bisolution and inverse}.

Note that all the definitions of inverses and bisolution that we will give below will be rigorous.
Yet, for the time being, they will satisfy the conditions \eqref{eq:bisolution} or \eqref{eq:inverse} only on a heuristic level.

The following definition introduces the two most natural inverses and the most natural bisolution:
\begin{definition}
  \label{def:classical-propagators}
  Define the operators $E^\bullet : C_\mathrm{c}(I,\cW)\to C(I,\cW)$
  \begin{equation}
    \label{kerno}
    (E^\bullet w)(t) \defn \int_I E^\bullet(t,s) w(s) \dif s,
    \quad
    \bullet = \PJ, \vee, \wedge,
  \end{equation}
  by their \emph{temporal integral kernels}
  \bes
  \begin{align*}
    E^\PJ(t,s)    & \defn R(t,s), \\
    E^\vee(t,s)   & \defn \theta(t-s) R(t,s), \\
    E^\wedge(t,s) & \defn -\theta(s-t) R(t,s).
  \end{align*}
  \ees
  $E^\PJ$ is called the \emph{Pauli--Jordan bisolution} and~$E^\vee, E^\wedge$ are called the \emph{forward} resp.\ \emph{backward inverse}.
  Jointly, we call them \emph{classical propagators}.
\end{definition}

Clearly, we have
\begin{equation}
  \label{eq:PJ=vee-wedge}
  E^\PJ = E^\vee - E^\wedge,
\end{equation}
which is analogous to \eqref{idi1}.

If $I$ is finite, then the operators $E^\vee,E^\wedge,E^\PJ$ can be extended to bounded operators on the Hilbertizable space $L^2(I,\cW)=L^2(I)\otimes\cW$.

If $I=\mathbb{R}$, typically they are not bounded on $L^2(\mathbb{R},\cW)$.
However, if $R(t,s)$ is uniformly bounded (which we will typically assume), then $E^\vee,E^\wedge,E^\PJ$ are bounded as operators $\langle t\rangle^{-s}L^2(I,\cW)\to \langle t\rangle^{s}L^2(I,\cW)$ for $s>\frac12$.

\subsection{Bisolutions and inverses associated with involutions}
\label{sub:In/out positive and negative frequency bisolutions}

Let $S_+$ and $S_-$ be two bounded involutions,
\begin{equation*}
  (\Pi_-^{(-)},\Pi_-^{(+)}),
  \quad
  (\Pi_+^{(-)},\Pi_+^{(+)})
\end{equation*}
the corresponding two pairs of complementary projections and
\begin{equation*}
  (\cZ_-^{(-)},\cZ_-^{(+)}),
  \quad
  (\cZ_+^{(-)},\cZ_+^{(+)}).
\end{equation*}
the corresponding  two pairs of complementary subspaces (see
Subsect.~\ref{sub:Involutions}).

For finite $t_+,t_-$, we set
\begin{alignat*}{3}
  S_\pm(t)         & = R(t,t_\pm ) S_\pm R(t_\pm ,t),
  \\
  \Pi_\pm^{(+)}(t) & \defn R(t,t_\pm ) \Pi_\pm^{(+)} R(t_\pm ,t),
  & \qquad
  \Pi_\pm^{(-)}(t) & \defn R(t,t_\pm ) \Pi_\pm^{(-)}R(t_\pm ,t),
  \\
  \cZ_\pm^{(+)}(t) & \defn R(t,t_\pm) \cZ_\pm^{(+)},
  & \qquad
  \cZ_\pm^{(-)}(t) & \defn R(t,t_\pm) \cZ_\pm^{(-)}.
\end{alignat*}

If $t_\pm=\pm\infty$, we assume that there exists $T$ such that on $\mathopen{]}-\infty,-T\mathclose{[}$ and $\mathopen{]}T,\infty\mathclose{[}$ the evolution is autonomous, that is, there exist $B_\pm$ such that
\begin{equation*}
  {\pm t}, {\pm s}>T \;\Rightarrow\; R(t,s)=\e^{-\im (t-s)B_\pm}.
\end{equation*}
We also assume that
\begin{equation}
  \label{porro}
  \e^{-\im (t-s)B_\pm}
  S_\pm=S_\pm \e^{-\im (t-s)B_\pm}.
\end{equation}

We then set
\begin{alignat*}{3}
  S_\pm(t) & = R(t,s) S_\pm R(s,t),
  \\
  \Pi_\pm^{(+)}(t) & \defn R(t,s) \Pi_\pm^{(+)} R(s,t),
  & \qquad
  \Pi_\pm^{(-)}(t) &\defn R(t,s) \Pi_\pm^{(-)} R(s,t),
  \\
  \cZ_\pm^{(+)}(t) & \defn R(t,s) \cZ_\pm^{(+)},
  & \qquad
  \cZ_\pm^{(-)}(t) &\defn R(t,s) \cZ_\pm^{(-)},
\end{alignat*}
where $\pm s>T$ is arbitrary.

\begin{definition}
  \label{def:smo}
  For any $t,s\in I$, we define
  \begin{subequations}
    \label{smoo1}
    \begin{align}
      \label{smoo1+}
      E_\pm^{(+)}(t,s) & \defn R(t,s) \Pi_\pm^{(+)} (s)=\Pi_\pm^{(+)} (t) R(t,s) .
      \\
      \label{smoo1-}
      E_\pm^{(-)}(t,s) & \defn R(t,s) \Pi_\pm^{(-)}(s)= \Pi_\pm^{(-)}(t)R(t,s).
    \end{align}
  \end{subequations}
  Define the operators $E_\pm^{(+)}$ and $E_\pm^{(-)} : C_\mathrm{c}(I,\cW)\to C(I,\cW)$ by their temporal integral kernels \eqref{smoo1}.
  $E_\pm^{(+)}$ are called the \emph{$\Pi_\pm^{(+)}$-in/out bisolutions} and $E_\pm^{(-)}$ are called the \emph{$\Pi_\pm^{(-)}$-in/out bisolutions}.
\end{definition}

Clearly, we have
\begin{equation}
  \label{eq:PJ=(+)-(-)}
  E^\PJ = E_\pm^{(+)} + E_\pm^{(-)},
\end{equation}
which is analogous to \eqref{idi2}.

\begin{definition}
  \label{def-feyn}
  We say that \emph{asymptotic complementarity} holds for $(\cZ_+^{(+)},\cZ_-^{(-)})$ if for some (and thus for all) $t \in I$,
  \begin{equation*}
    \bigl( \cZ_+^{(+)}(t),\cZ_-^{(-)}(t)\bigr)
  \end{equation*}
  is a pair of complementary subspaces of $\cW$.
  Suppose that asymptotic complementarity holds for $(\cZ_+^{(+)},\cZ_-^{(-)})$.
  Then we define
  \begin{subequations}
    \label{pair1}
    \begin{align}
      \Lambda^{\Feyn(+)}(t), & \quad \text{the projection onto $\cZ_+^{(+)}(t)$ along $\cZ_-^{(-)}(t)$},
      \\
      \Lambda^{\Feyn(-)}(t), & \quad \text{the projection onto $\cZ_-^{(-)}(t)$ along $\cZ_+^{(+)}(t)$.}
    \end{align}
  \end{subequations}
   It is the pair of projections associated to the direct sum decomposition $\cW = \cZ_+^{(+)}(t) \oplus \cZ_-^{(-)}(t)$.
  We also define the operator $E^\Feyn : C_\mathrm{c}(I,\cW)\to C(I,\cW)$ as in~\eqref{kerno} by its temporal integral kernel
  \begin{equation*}
    E^\Feyn(t,s) \defn \theta(t-s) R(t,s) \Lambda^{\Feyn(+)}(s) - \theta(s-t) R(t,s) \Lambda^{\Feyn(-)}(s).
  \end{equation*}
  $E^\Feyn$ is called the \emph{$\cZ_+^{(+)}$-out $\cZ_-^{(-)}$-in inverse}.
\end{definition}

The following definition is fully analogous to the previous one:

\begin{definition}
  \label{def-afeyn}
  We say that \emph{asymptotic complementarity} holds for $(\cZ_+^{(-)},\cZ_-^{(+)})$ if for some (and thus for all) $t \in I$,
  \begin{equation*}
    \bigl( \cZ_+^{(-)}(t),\cZ_-^{(+)}(t)\bigr)
  \end{equation*}
  is a pair of complementary subspaces of $\cW$.
  Suppose that asymptotic complementarity holds for $(\cZ_+^{(-)},\cZ_-^{(+)})$.
  Then we define
  \begin{subequations}
    \label{pair1-anti}
    \begin{align}
      \Lambda^{\aFeyn(-)}(t), & \quad \text{the projection onto $\cZ_+^{(-)}(t)$ along $\cZ_-^{(+)}(t)$},
      \\
      \Lambda^{\aFeyn(+)}(t), & \quad \text{the projection onto $\cZ_-^{(+)}(t)$ along $\cZ_+^{(-)}(t)$}.
    \end{align}
  \end{subequations}
   It is the pair of projections associated to the direct sum decomposition $\cW = \cZ_+^{(-)}(t) \oplus \cZ_-^{(+)}(t)$.
  We also define the operator $E^\aFeyn : C_\mathrm{c}(I,\cW)\to C(I,\cW)$ as in~\eqref{kerno} by its temporal integral kernel
  \begin{equation*}
    E^\aFeyn(t,s) \defn \theta(t-s) R(t,s) \Lambda^{\aFeyn(-)}(s) - \theta(s-t) R(t,s) \Lambda^{\aFeyn(+)}(s).
  \end{equation*}
  $E^\aFeyn$ is called the \emph{$\cZ_+^{(-)}$-out $\cZ_-^{(+)}$-in inverse}.
\end{definition}

We clearly have
\begin{align*}
  R(t,s) \Lambda^{\Feyn(\pm)}(s) & = \Lambda^{\Feyn(\pm)}(t) R(t,s), \\ R(t,s) \Lambda^{\aFeyn(\pm)}(s) & = \Lambda^{\aFeyn(\pm)}(t) R(t,s),
\end{align*}

Heuristically, $E^\Feyn$ and $E^\aFeyn$ are inverses in the sense of \eqref{eq:inverse}.
If $I$ is finite, then they are bounded on $L^2(I,\cW)$.

Asymptotic complementarity is trivially satisfied for the pairs $(\cW,\{0\})$ and $(\{0\}, \cW)$.
The corresponding inverses are simply $E^\vee$ and $E^\wedge$, defined in Def.~\ref{def:classical-propagators}.

\begin{remark}
  Let us explain the notation that we are using: $+$ and $-$ without parenthesis are related to $t_+$ or ``out'', resp.\ $t_-$ or ``in''.
  $({+})$ and $({-})$ inside parenthesis denote the ``particle space'', resp.\ the ``antiparticle space''.
\end{remark}

\subsection{Identities involving bisolutions and inverses}

We make all assumptions of Subsect.~\ref{sub:In/out positive and negative frequency bisolutions}.
In addition we suppose that asymptotic complementarity holds for both pairs of subspaces.
Using Prop.~\ref{prop:compl-proj}, we have the following:
\begin{proposition}
  The projections defined in \eqref{pair1} and \eqref{pair1-anti} are given explicitly by
  \begin{alignat*}{2} \Lambda^{\Feyn(+)}(t) & = \Pi_+^{(+)}(t) \Upsilon(t)^{-1} \Pi_-^{(+)}(t), \quad & \Lambda^{\Feyn(-)}(t) & = \Pi_-^{(-)}(t) \Upsilon(t)^{-1} \Pi_+^{(-)}(t), \\\text{and}\quad \Lambda^{\aFeyn(+)}(t) & = \Pi_-^{(+)}(t) \Upsilon(t)^{-1} \Pi_+^{(+)}(t), \quad & \Lambda^{\aFeyn(-)}(t) & = \Pi_+^{(-)}(t) \Upsilon(t)^{-1} \Pi_-^{(-)}(t),
  \end{alignat*}
  where $\Upsilon(t)$ is the invertible operator defined by
  \begin{align*}
    \Upsilon(t) & \defn \frac14 \bigl( 2 + S_-(t) S_+(t) + S_+(t) S_-(t) \bigr).
  \end{align*}
\end{proposition}

Observe that $\Upsilon(t)$ has the properties
\begin{align*}
  \Upsilon(t) & = R(t,s) \Upsilon(s) R(s,t), \\ \Upsilon(t) S_\pm(t) & = S_\pm(t) \Upsilon(t).
\end{align*}

Thus we have defined the inverses $E^\vee$, $E^\wedge$, $E^\Feyn$, $E^\aFeyn$ and the bisolutions $E^\PJ$, $E_-^{(+)}$, $E_-^{(-)}$, $E_+^{(+)}$, $E_+^{(-)}$.
They satisfy the relations \eqref{eq:PJ=vee-wedge} and~\eqref{eq:PJ=(+)-(-)}.
As described in the introduction, in the setting of the Minkowski space and, more generally, of a stationary spacetime, they satisfy several other identities.
These identities do not hold in general.
Instead, we have:

\begin{proposition}
  The following three identities hold:
  \begin{align}
    \label{rwq4z}
    (E_+^{(+)} - E_+^{(-)} - E_-^{(+)} + E_-^{(-)} )(t,s) & = R(t,s)(S_+-S_-)(s);
    \\
    \label{rwq4}
    ( E^\Feyn + E^\aFeyn - E^\vee - E^\wedge )(t,s) & = \frac14 R(t,s) \Upsilon(s)^{-1} [S_-(s), S_+(s)];
  \end{align}
  \begin{equation}\begin{split}
    \label{rwq4a}
    \MoveEqLeft ( E^\Feyn - E^\aFeyn )(t,s) - \frac12 ( E_+^{(+)} - E_+^{(-)} + E_-^{(+)} - E_-^{(-)} )(t,s) \\
    & = \frac{1}{8} R(t,s) \Upsilon(s)^{-1} \big[S_+(s)-S_-(s), [S_+(s),S_-(s)]\big].
  \end{split}\end{equation}
\end{proposition}
\begin{proof}
  (\ref{rwq4z}) is straightforward.

  Obviously, the difference of two inverses is a bisolution.
  The temporal kernels of the following bisolutions have very simple forms:
  \begin{align}
    (E^\Feyn-E^\wedge)(t,s)  & = R(t,s)\Lambda^{\Feyn(+)}(s),  \label{ui1} \\
    (E^\Feyn-E^\vee)(t,s)    & = -R(t,s)\Lambda^{\Feyn(-)}(s), \label{ui2} \\
    (E^\aFeyn-E^\wedge)(t,s) & = R(t,s)\Lambda^{\aFeyn(-)}(s), \label{ui3} \\
    (E^\aFeyn-E^\vee)(t,s)   & = -R(t,s)\Lambda^{\aFeyn(+)}(s).\label{ui4}
  \end{align}

  Taking the sum of (\ref{ui1}) and (\ref{ui4}) we obtain
  \begin{align*}
    ( E^\Feyn + E^\aFeyn- E^\vee - E^\wedge )(t,s)
    & = R(t,s)(\Lambda^{\Feyn(+)}-\Lambda^{\aFeyn(-)})(s) \\
    & = R(t,s)\Upsilon(s)^{-1}(\Pi_+^{(+)}\Pi_-^{(+)}- \Pi_-^{(+)}\Pi_+^{(+)})(s),
  \end{align*}
  which yields (\ref{rwq4}).
  Taking the difference of (\ref{ui1}) and (\ref{ui3}) we obtain the following identities:
  \begin{align*}
    E^\Feyn(t,s) - E^\aFeyn(t,s)
    & = R(t,s)(\Lambda^{\Feyn(+)}-\Lambda^{\aFeyn(-)})(s) \\
    & = \frac14R(t,s)\Upsilon(s)^{-1}\bigl((1+S_+)(1+S_-) -(1-S_+)(1-S_-)\bigr)(s) \\
    & = \frac12R(t,s)\Upsilon(s)^{-1}(S_++S_-)(s) \\
    & = R(t,s)\Big(\frac12(S_++S_-)(s) + \frac18\Upsilon(s)^{-1}(S_++S_--S_-S_+S_--S_+S_-S_+)(s)\Big),
  \end{align*}
  This yields (\ref{rwq4a}).
\end{proof}

Note that for the standard choice of propagators in a stationary QFT the right hand sides of (\ref{rwq4z}), (\ref{rwq4a}) and (\ref{rwq4}) vanish.
It is easy to see that they do not have to vanish in general.

The identities \eqref{rwq4} and~\eqref{rwq4a} simplify in some important situations:
\begin{proposition}
  Assume that asymptotic complementarity holds.
  Further, suppose that for any (and hence for all) $t \in I$
  \begin{equation}
    \label{eq:S-S+_commute}
    S_-(t) S_+(t) = S_+(t) S_-(t).
  \end{equation}
  Then
  \begin{subequations}
    \begin{align}
      \label{eq:Feyn+aFeyn=vee+wedge}
      E^\Feyn + E^\aFeyn & = E^\vee + E^\wedge,                                        \\ \label{eq:Feyn-aFeyn=(+)+(-)}
      E^\Feyn - E^\aFeyn & = \frac12 ( E_+^{(+)} - E_+^{(-)} + E_-^{(+)} - E_-^{(-)}).
    \end{align}
  \end{subequations}
\end{proposition}

Eq.~\eqref{eq:S-S+_commute} is satisfied in a number of interesting situations.
In particular, if the evolution is autonomous, it is natural to assume that $S_+=S_- \nfed S_\bullet$, requiring that it commutes with the generator $B$.
Then $E_\pm^{(+)}$ and $E_\pm^{(-)}$ collapse to two bisolutions:
\begin{align*}
  E_+^{(+)} & = E_-^{(+)} \nfed E^{(+)}, \\ E_+^{(-)} & = E_-^{(-)} \nfed E^{(-)}.
\end{align*}
Thus, in the autonomous case, the identity \eqref{eq:Feyn+aFeyn=vee+wedge} holds and \eqref{eq:Feyn-aFeyn=(+)+(-)} can be rewritten as
\begin{align*}
  E^\Feyn - E^\aFeyn & = E^{(+)} - E^{(-)}.
\end{align*}

\subsection{Almost unitary evolutions on Hilbertizable spaces}
\label{Almost unitary evolutions on Hilbertizable spaces}

So far we discussed generators of an evolution only in a heuristic way.
In this subsection we will describe a setting that allows us to make this concept precise.
In view of our applications, we will introduce generators of evolutions on Hilbertizable spaces that one might call \emph{almost unitary evolutions}.

For the remainder of this section, we consider the scale of Hilbertizable spaces $\cW_\alpha$, $\alpha \in [0,1]$, as in Subsect.~\ref{sub:Interpolation between Hibertizable spaces}.

\begin{theorem}[cf.~Thm.~C.10 of~\cite{derezinski-siemssen:propagators}]\label{thm:evolution}
  Let $\{B(t)\}_{t\in I}$ be a family of densely defined, closed operators on~$\cW_0$.
  Suppose that the following conditions are satisfied:
  \begin{enumerate}[label=(\alph*)]
    \item
          $\cW_1 \subset \Dom \bigl(B(t)\bigr)$ so that $B(t) \in B(\cW_1, \cW_0)$ and $I \ni t \mapsto B(t)\in B(\cW_1, \cW_0)$ is norm continuous.
    \item
          For every $t \in I$, scalar products $\cinner{\cdot}{\cdot}_{0,t}$ and $\cinner{\cdot}{\cdot}_{1,t}$ compatible with~$\cW_0$ resp.~$\cW_1$ have been chosen.
          Denote the corresponding Hilbert spaces~$\cW_{0,t}$ and~$\cW_{1,t}$.
    \item
          $B(t)$ is self-adjoint in the sense of~$\cW_{0,t}$ and the part~$\tilde{B}(t)$ of~$B(t)$ in~$\cW_1$ is self-adjoint in the sense of~$\cW_{1,t}$.
    \item
          For a positive $C \in L_\loc^1(I)$ and all $s,t \in I$,
          \begin{align*}
            \norm{w}_{0,t} & \leq \norm{w}_{0,s} \exp\abs*{\int_s^t C(r) \dif r}, \\ \norm{v}_{ 1,t} & \leq \norm{v}_{1,s} \exp\abs*{\int_s^t C(r) \dif r}.
          \end{align*}
  \end{enumerate}
  Then there exists a unique family of bounded operators $\{R(t,s)\}_{s,t \in I}$, on~$\cW_0$ with the following properties:
  \begin{enumerate}[label=(\roman*)]
    \item
          For all $r,s,t \in I$, we have the identities~\eqref{eq:evo_idents}.
    \item
          $R(t,s)$ is $\cW_0$-strongly continuous.
          It preserves $\cW_1$ and is $\cW_1$-strongly continuous.
          Hence it preserves $\cW_\alpha$, $0\leq\alpha\leq1$, and is $\cW_\alpha$-continuous.
          Moreover,
          \begin{equation*}
            \norm{R(t,s)}_{\alpha,t} \leq \exp\abs*{\int_s^t 2 C(r) \dif r}, \quad s,t \in I.
          \end{equation*}
    \item
          For all $w \in \cW_1$ and $s,t \in I$,
          \begin{align*}
            \im\partial_t R(t,s) w  & = B(t) R(t,s) w, \\
            -\im\partial_s R(t,s) w & = R(t,s) B(s) w,
          \end{align*}
          where the derivatives are in the strong topology of~$\cW_0$.
  \end{enumerate}
  We call $\{R(t,s)\}_{s,t \in I}$ the \emph{evolution} generated by~$B(t)$.
\end{theorem}

Note that, if $t_\pm$ are finite, the above theorem remains true if we everywhere replace $I$ with $I^\cl=[t_-,t_+]$, and $L_\loc^1]t_-,t_+[$ with $L^1[t_-,t_+]$, provided that we consider only the right/left-sided derivatives at $t_-/t_+$.

Sometimes it is convenient to use an easy generalization of Thm.~\ref{thm:evolution}, where the generator is perturbed by a bounded operator.
It is also proven in~\cite{derezinski-siemssen:propagators}.

\begin{theorem}
  \label{thm:evolution-per}
  Suppose that $ \{B_0(t)\}_{t\in I}$ satisfies all the assumptions of Thm.~\ref{thm:evolution} and $I\ni t\mapsto V(t)\in B(\cW_0)$ is a norm continuous family of operators.
  Let $B(t) \defn B_0(t)+V(t)$.
  Then there exists a unique family of bounded operators $\{R(t,s)\}_{s,t \in I}$, on~$\cW_0$ satisfying all properties of Thm.~\ref{thm:evolution}.
\end{theorem}

If we want to pass to the real case, we use the following obvious fact:

\begin{proposition}
  If $\cW_0$ has a conjugation which preserves $\cW_1$, then $R(t,s)$ is real for $t,s \in I$ if and only if its generator $B(t)$ is anti-real for all $t\in I$.
\end{proposition}

\subsection{Rigorous concept of a bisolution and inverse}
\label{Rigorous concept of a bisolution and inverse}

Under the assumptions of Thm.~\ref{thm:evolution} it is possible to propose a rigorous version of a concept of a (left) bisolution and a (left) inverse, and check that they are satisfied by the $E^\bullet$ that we have constructed.

\begin{proposition}
  \mbox{}
  \begin{enumerate}
    \item Let $v \in C_\mathrm{c}(I,\cW_1)\cap C_\mathrm{c}^1(I,\cW_0)$.
          Then for $E^\bullet=E^\PJ,E_-^{(+)}, E_-^{(-)}$, $E_+^{(+)}, E_+^{(-)}$ we have
          \begin{alignat*}{2}
            E^\bullet \bigl(\partial_t + \im B(t)\bigr) v & = 0, \intertext{and for $E^\bullet= E^\wedge,E^\vee,E^\Feyn, E^\aFeyn$ we have} \quad
            E^\bullet \bigl(\partial_t + \im B(t)\bigr) v & = v.
          \end{alignat*}
    \item Let $w \in C_\mathrm{c}(I,\cW_1)$.
          Then for $E^\bullet=E^\PJ$ we have
          \begin{alignat*}{2}
            \bigl(\partial_t + \im B(t)\bigr) E^\bullet w & = 0,
            \intertext{and for $E^\bullet=E^\wedge,E^\vee$ we have}
            \bigl(\partial_t + \im B(t)\bigr) E^\bullet w & = w.
          \end{alignat*}
  \end{enumerate}
  \label{pos8}
\end{proposition}
\begin{proof}
  Let us prove (1) for $E^\Feyn$:
  \begin{align*}
    \Bigl(E^\Feyn \bigl(\partial_s+\ri B(s)\bigr)v\Bigr)(t)
    & = \int_{t_-}^t\Lambda^{\Feyn(+)}(t)R(t,s)\bigl(\partial_s+\ri B(s)\bigr)v(s)\dif s      \\
    & \quad -\int_t^{t_+}\Lambda^{\Feyn(-)}(t)R(t,s)\bigl(\partial_s+\ri B(s)\bigr)v(s)\dif s \\
    & = \int_{t_-}^t\Lambda^{\Feyn(+)}(t)\partial_s\bigl(R(t,s)v(s)\bigr)\dif s               \\
    & \quad -\int_t^{t_+}\Lambda^{\Feyn(-)}(t)\partial_s\bigl(R(t,s)v(s)\bigr)\dif s          \\
    & = \Lambda^{\Feyn(+)}(t)v(t)+\Lambda^{\Feyn(-)}(t)v(t)                                   \\
    & = v(t).
  \end{align*}

  Let us prove (2) for $E^\vee$:
  \begin{align*}
    \Bigl(\bigl(\partial_t+\ri B(t)\bigr)E^\vee w\Bigr)(t)
    & = \bigl(\partial_t+\ri B(t)\bigr)\int_{t_-}^tR(t,s)w(s)\dif s             \\
    & = R(t,t)w(t)+ \int_{t_-}^t\bigl(\partial_t+\ri B(t)\bigr)R(t,s)w(s)\dif s \\
    & = w(t).
    \tag*{$\square$}
  \end{align*}
\end{proof}

\section{Evolutions on pseudo-unitary spaces}
\label{Hilbertizable pseudo-unitary spaces}

\emph{Pre-pseudo-unitary} spaces are Hilbertizable spaces with a distinguished bounded Hermitian form.
They can be viewed as complexifications of \emph{pre-symplectic spaces} -- real spaces with a distinguished bounded antisymmetric form.

In practice, one usually assumes that the Hermitian or pre-symplectic form is non-degenerate.
Then these spaces are called \emph{pseudo-unitary}, resp.\ \emph{symplectic}.
A transformation preserving the structure of a pseudo-unitary, resp.\ symplectic space is called \emph{pseudo-unitary}, resp.\ \emph{symplectic}.

\emph{Krein spaces} constitute an especially well-behaved class of pseudo-unitary spaces.
The Krein structure adds interesting new features to bisolutions and inverses of $M \defn \partial_t+\ri B(t)$.
The most interesting new fact is the automatic validity of asymptotic complementarity if the ``in particle space'' is maximally positive and the ``out antiparticle space'' is maximally negative.
This implies the existence of the Feynman and anti-Feynman inverses.

We will also discuss generators of 1-parameter groups preserving the pseudo-unitary structure, called \emph{pseudo-unitary generators}.
In particular, we will introduce the so-called \emph{stable pseudo-unitary generators}, which possess positive \emph{Hamiltonians}.
They are distinguished both on physical and mathematical grounds.
Especially good properties have \emph{strongly stable pseudo-unitary generators}, whose positive Hamiltonians are bounded away from zero.

We discuss two constructions of a pseudo-unitary evolution $R(t,s)$.
starting from a time-depenent generator $I\ni t\mapsto B(t)$.
The first construction uses a nested pair of Hilbertizable spaces $\cW_1\subset \cW_0$, where $\cW_0$ is equipped with a Hermitian form $Q$, and $B(t):\cW_1\to\cW_0$.
The second construction uses a nested pair of Hilbertizable spaces $\cW_{\frac12}\subset \cW_{-\frac12}$ equipped with a pairing $Q$ and $B(t): \cW_{\frac12}\to\cW_{-\frac12}$.
The pseudo-unitary space $\cW_0$ is obtained by interpolation.
Later on we will use both constructions.

\subsection{Symplectic and pseudo-unitary spaces}

\begin{definition}
  A \emph{pre-symplectic space} is a real vector space $\cY$ equipped with an antisymmetric form $\omega$, called a \emph{pre-symplectic form}
  \begin{equation*}
    \cY \times \cY \ni (v,w) \mapsto \rinner{v}{\omega w} \in \RR.
  \end{equation*}
  If $\omega$ is non-degenerate, then  $\cY$ is called a
    \emph{symplectic} space.
  If the dimension of $\cY$ is infinite, we assume that $\cY$ is Hilbertizable and $\omega$ is bounded.
\end{definition}

\begin{definition}
  We will say that a bounded invertible operator $R$ on a pre-symplectic space $(\cY,\omega)$ \emph{preserves $\omega$} if
  \begin{equation*}
    \rinner{R v}{\omega R w} = \rinner{v}{\omega w}.
  \end{equation*}
  If in addition $\omega$ is non-degenerate, we will say that $R$ is \emph{symplectic}.
\end{definition}

\begin{definition}
  A \emph{pre-pseudo-unitary space} is a complex vector space $\cW$ equipped with a Hermitian form $Q$
  \begin{equation*}
    \cW \times \cW \ni (v,w) \mapsto \cinner{v}{Q w} \in \CC.
  \end{equation*}
  If $Q$ is non-degenerate, then  $\cW$ is called
   a \emph{pseudo-unitary space}.
  If the dimension of $\cW$ is infinite, we assume that $\cW$ is Hilbertizable and $Q$ is bounded.
\end{definition}

\begin{definition}\label{pseudo-uni}
  We will say that a bounded invertible operator $R$ on $(\cW,Q)$ \emph{preserves $Q$} if
  \begin{equation}\label{prepre}
    \cinner{R v}{Q R w} = \cinner{v}{Q w}.
  \end{equation}
  If in addition $Q$ is non-degenerate, we will say that $R$ is \emph{pseudo-unitary}.
\end{definition}

Note that even if one starts from a real (pre-)symplectic space, it is useful to consider its complexification.
In Subsect.~\ref{sub:Complexification of bilinear forms} we described how to pass from the real to complex formalism.
In this section we will treat the complex formalism as the standard one.

In the context of a pre-pseudo-unitary space treated as the complexification of a pre-symplectic space it is natural to consider conjugations that anti-preserve (and not preserve) $Q$, that is $ \bar{\cinner{v}{Q w}} = -\cinner{\bar v}{Q \bar w},$ see (\ref{antipreserves}).

\begin{definition}
  \label{def:conjugation.}
  An antilinear involution $v \mapsto \bar v$ on a pseudo-unitary space $(\cW,Q)$ which anti-preserves $Q$ and such that there exists a compatible scalar product $\cinner{\cdot}{\cdot}_\bullet$ satisfying
  \begin{equation*}
    \cinner{\bar v}{\bar w}_\bullet = \bar{\cinner{v}{w}}_\bullet,
  \end{equation*}
  will be called a \emph{conjugation on $(\cW,Q)$}.
\end{definition}
As in \eqref{real}, given a conjugation we can define the real subspace $\cW_\RR$ of $\cW$.
The restriction of $Q$ to $\cW_\RR$ is clearly a pre-symplectic space.

\subsection{Admissible involutions and Krein spaces}
\label{sub:Admissible involutions and Krein spaces}

Let $(\cW,Q)$ be a pre-pseudo-unitary space.

\begin{definition}
  \label{admissi}
  A (bounded) involution $S_\bullet$ on $\cW$ will be called \emph{admissible} if it preserves $Q$ and the scalar product
  \begin{equation}
    \label{eq:bull_space}
    {\cinner{v}{w}}_\bullet \defn \cinner{v}{Q S_\bullet w} = \cinner{S_\bullet v}{Q w}
  \end{equation}
  is compatible with the Hilbertizable structure of $\cW$.
  Sometimes we will write $\cW_\bullet$ to denote the space $\cW$ equipped with the scalar product~\eqref{eq:bull_space}.
\end{definition}

\begin{definition}
  A pre-pseudo-unitary space is called a \emph{Krein space} if it possesses an admissible involution.
\end{definition}

Clearly, every Krein space is pseudo-unitary.

\begin{proposition}
  If $S_\bullet$ is an admissible involution on $(\cW,Q)$, then $S_\bullet$ is self-adjoint and unitary on $\cW_\bullet$.
\end{proposition}

For any admissible involution $S_\bullet$, we define the corresponding \emph{particle projection} $\Pi_\bullet^{(+)}$ and \emph{particle space} $\cZ_\bullet^{(+)}$, as well as the \emph{antiparticle projection} $\Pi_\bullet^{(-)}$ and \emph{antiparticle space} $\cZ_\bullet^{(-)}$, as in~\eqref{projo-}.
Note the following relations:
\begin{align*}
  {\cinner{v}{w}}_\bullet & = {\cinner{\Pi_\bullet^{(+)} v}{\Pi_\bullet^{(+)} w}}_\bullet + {\cinner{\Pi_\bullet^{(-)} v}{\Pi_\bullet^{(-)} w}}_\bullet, \\
  \cinner{v}{Q w} & = {\cinner{\Pi_\bullet^{(+)} v}{\Pi_\bullet^{(+)} w}}_\bullet - {\cinner{\Pi_\bullet^{(-)} v}{\Pi_\bullet^{(-)} w}}_\bullet.
\end{align*}

Let us make an additional comment on Krein spaces with conjugations.

\begin{proposition}
  \label{prop:S-real}
  Suppose that $(\cW,Q)$ is a Krein space with  conjugation.
  If $S_\bullet$ is an admissible anti-real involution, then $\im S_\bullet$ is real and we have
  \begin{equation*}
    \bar{\Pi_\bullet^{(+)}} = \Pi_\bullet^{(-)}, \quad \bar{\cZ_\bullet^{(+)}} = \cZ_\bullet^{(-)},
  \end{equation*}
  so that $\cW = \cZ_\bullet^{(+)} \oplus \bar{\cZ_\bullet^{(+)}}$.
\end{proposition}

\subsection{Basic constructions in Krein spaces}

Let $(\cW,Q)$ be a Krein space.

\begin{definition}
  Let $\cZ\subset \cW$.
  We define its \emph{$Q$-orthogonal complement} as follows: \[\cZ^{\perp Q} \defn \{w\in\cW\mid\cinner{w}{Qv}=0,\; v\in\cZ\}.
  \]
  If $\cinner{\cdot}{\cdot}_\bullet$ is a scalar product, we also have the \emph{$\bullet$-orthogonal complement}
  \[\cZ^{\perp \bullet} \defn \{w\in\cW\mid\cinner{w}{v}_\bullet=0,\; v\in\cZ\}.\]
\end{definition}
\begin{proposition}
  \label{basic0}
  \begin{enumerate}
    \item If $\cZ$ is a closed subspace, then so is $\cZ^{\perp Q}$, and
          $(\cZ^{\perp Q})^{\perp Q}=\cZ$.
    \item If $\cZ_1,\cZ_2$ are complementary in $\cW$, then so are
          $\cZ_1^{\perp Q},\cZ_2^{\perp Q}$.
    \item If $\cinner{v}{Qw}=\cinner{v}{S_\bullet w}_\bullet$ (equivalently, if $S_\bullet$ is admissible), then $\cZ^{\perp Q}=S_\bullet\cZ^{\perp\bullet}$.
  \end{enumerate}
\end{proposition}

\begin{definition}
  Let $A\in B(\cW)$.
  We define its $Q$-adjoint as follows: \[\cinner{A^{*Q}v}{Qw}=\cinner{v}{QAw},\quad v,w\in\cW.
  \]
  If $\cinner{\cdot}{\cdot}_\bullet$ is a scalar product, we also have the \emph{$\bullet$-adjoint of $A$}
  \[\cinner{A^{*\bullet}v}{w}_\bullet=\cinner{v}{Aw}_\bullet,\quad v,w\in\cW.\]
\end{definition}
\begin{proposition}
  \mbox{}
  \begin{enumerate}
    \item If $\cinner{v}{Qw}=\cinner{v}{S_\bullet w}_\bullet$ (equivalently, if $S_\bullet$ is admissible), then $A^{* Q}=S_\bullet A^{*\bullet}
            S_\bullet$.
    \item Let $
            (\Pi^{(+)},\Pi^{(-)})$ be a pair of complementary projections.
          Then $(\Pi^{(+)*Q}, \Pi^{(-)*Q})$ is also a pair of complementary projections and
          \begin{align*}
            \cR(\Pi^{(\pm)*Q})= \cN(\Pi^{(\mp)*Q})= \cR(\Pi^{(\mp)})^{\perp Q}= \cN(\Pi^{(\pm)})^{\perp Q}.
          \end{align*}
  \end{enumerate}
\end{proposition}

\begin{proposition}
  Let $\cZ_\bullet^{(+)}$ be a closed subspace of $\cW$.
  Set $\cZ_\bullet^{(-)} \defn \cZ_\bullet^{(+)\perp Q}$.
  The following conditions are equivalent:
  \begin{enumerate}
    \item $\cZ_\bullet^{(+)}$ satisfies
          \begin{align}
            \label{pos1}
            v\in\cZ_\bullet^{(+)}                       & \Rightarrow\cinner{v}{Qv}\geq0, \\
            \label{pos2}
            \text{and }\quad      v\in\cZ_\bullet^{(-)} & \Rightarrow\cinner{v}{Qv}\leq0.
          \end{align}
    \item
          $\cZ_\bullet^{(+)}$ is a maximal closed subspace of $\cW$ with the property (\ref{pos1}).
    \item The spaces $\cZ_\bullet^{(+)}$ and $\cZ_\bullet^{(-)}$ are complementary, and if $(\Pi_\bullet^{(+)},\Pi_\bullet^{(-)})$ is the corresponding pair of projections, then
          $    S_\bullet \defn  \Pi_\bullet^{(+)}-\Pi_\bullet^{(-)}$
          is an admissible involution.
  \end{enumerate}
  \label{pos3}
\end{proposition}

\begin{definition}
  If $\cZ_\bullet^{(+)}$ satisfies the conditions of Prop.~\ref{pos3}, then it is called a \emph{maximally positive subspace}.
  Analogously we define \emph{maximally negative subspaces}.
\end{definition}

\subsection{Pairs of admissible involutions}
\label{sub:Pairs of admissible involutions}

Let $S_1$, $S_2$ be a pair of admissible involutions on a Krein space $(\cW,Q)$.
We will describe some structural properties of such a pair.

Let $\Pi_i^{(+)},\Pi_i^{(-)},\cZ_i^{(+)},\cZ_i^{(-)}$, $i=1,2$, be defined as in (\ref{projo-}).
Set
\begin{equation*}
  K \defn S_2 S_1, \quad c \defn \Pi_1^{(+)} \frac{\one-K}{\one+K} \Pi_1^{(-)},
\end{equation*}
where $c$ is interpreted as an operator from $\cZ_1^{(-)}$ to $\cZ_1^{(+)}$.

\begin{proposition}
  $K$ is pseudo-unitary and invertible.
  $K$ is positive and  $\norm{c}<1$  with respect to ${\cinner{\cdot}{\cdot}}_1$ and ${\cinner{\cdot}{\cdot}}_2$.
  We have
  \begin{align}
    \label{chazar0}
    S_1KS_1=S_2 K S_2                                          & = K^{-1},                  \\ \label{chazar}
    S_1\frac{\one-K}{\one+K} S_1= S_2\frac{\one-K}{\one+K} S_2 & = - \frac{\one-K}{\one+K}.
  \end{align}
\end{proposition}
\begin{proof}
  $K$ is pseudo-unitary as the product of two pseudo-unitary transformations.
  The inequality
  \begin{align*}
    {\cinner{v}{K v}}_1 & = \cinner{S_1 v}{Q S_2 S_1 v} = {\cinner{S_1 v}{S_1 v}}_2 \geq a {\cinner{S_1 v}{S_1 v}}_1 = a {\cinner{v}{v}}_1
  \end{align*}
  with $a>0$ shows the positivity of $K$ with respect to $\cinner{\cdot}{\cdot}_1$ and its invertibility.
  This implies $\|\frac{\one-K}{\one+K}\|<1$.
  Hence $\|c\|<1$.

  The identities \eqref{chazar0} and \eqref{chazar} are direct consequences of the definition of $K$ and $S_1^2=S_2^2=\one$.
\end{proof}

\begin{proposition}
  Using the decomposition $\cW = \cZ_1^{(+)} \oplus \cZ_1^{(-)}$ we have
  \begin{subequations}
    \begin{alignat}{2}\label{chazar1}
      \frac{\one-K}{\one+K} & =
      \begin{bmatrix}
        0   & c \\
        c^* & 0
      \end{bmatrix}
      ,                                    \\
      K                     & = \mathrlap{
        \begin{bmatrix}
          (\one+cc^*)(\one-cc^*)^{-1} & -2c(\one-c^*c)^{-1}         \\
          -2c^*(\one-cc^*)^{-1}       & (\one+c^*c)(\one-c^*c)^{-1}
        \end{bmatrix}
        ,}
      \label{chazar2}
      \\
      \Pi_1^{(+)}           & =
      \begin{bmatrix}
        \one & 0 \\
        0    & 0
      \end{bmatrix}
      ,
                            & \quad
      \Pi_2^{(+)}           & =
      \begin{bmatrix}
        (\one-cc^*)^{-1}     & c(\one-c^*c)^{-1}     \\
        -c^*(\one-cc^*)^{-1} & -c^*c(\one-c^*c)^{-1}
      \end{bmatrix}
      ,
      \label{chazar22}
      \\
      \Pi_1^{(-)}           & =
      \begin{bmatrix}
        0 & 0    \\
        0 & \one
      \end{bmatrix}
      ,
                            & \quad
      \Pi_2^{(-)}           & =
      \begin{bmatrix}
        -cc^*(\one-cc^*)^{-1} & -c(\one-c^*c)^{-1} \\
        c^*(\one-cc^*)^{-1}   & (\one-c^*c)^{-1}
      \end{bmatrix}
      ,
      \label{chazar23}
      \\
      S_1                   & =
      \begin{bmatrix}
        \one & 0     \\
        0    & -\one
      \end{bmatrix}
      ,
                            & \quad
      S_2                  = &
      \begin{bmatrix}
        (\one+cc^*)(\one-cc^*)^{-1} & 2c(\one-c^*c)^{-1}           \\
        -2c^*(\one-cc^*)^{-1}       & -(\one+c^*c)(\one-c^*c)^{-1}
      \end{bmatrix}
      .
      \label{chazar25}
    \end{alignat}
  \end{subequations}
  Moreover, if $S_1$ is an admissible involution and $\norm{c} < 1$, then $S_2$ given as in \eqref{chazar25} is an admissible involution.
\end{proposition}
\begin{proof}
  \eqref{chazar} implies \eqref{chazar1}.
  From the definition of $c$ (or \eqref{chazar1}) we obtain
  \begin{align*}
    K & =
          \begin{bmatrix}
            \one & -c \\ -c^* & \one
          \end{bmatrix}
    \begin{bmatrix}
      \one & c \\ c^* & \one
    \end{bmatrix}
    ^{-1}                          \\  & =
    \begin{bmatrix}
      \one & -c \\ -c^* & \one
    \end{bmatrix}
       \begin{bmatrix}
      (\one-cc^*)^{-1} & -c(\one-c^*c)^{-1} \\ -c^*(\one-cc^*)^{-1} & (\one-c^*c)^{-1}
    \end{bmatrix}
       .
  \end{align*}
  This yields \eqref{chazar2}.

  From $S_2 = K S_1$ we obtain~\eqref{chazar22}, \eqref{chazar23} and \eqref{chazar25}.
\end{proof}

The involutions $S_1$ and $S_2$ correspond to the pairs of complementary subspaces $(\cZ_1^{(+)},\cZ_1^{(-)})$, resp.\ $(\cZ_2^{(+)},\cZ_2^{(-)})$.
The following proposition implies the existence of two other direct sum decompositions.
This fact plays an important role in the construction of the (in-out) Feynman inverse.

\begin{proposition}
  \label{prop:asymp_compl_proof}
  The pairs of subspaces $(\cZ_1^{(+)}, \cZ_2^{(-)})$ and $(\cZ_2^{(+)}, \cZ_1^{(-)})$ are complementary.
  Here are the corresponding projections:
  \begin{alignat*}{3}
    \Lambda_{12}^{(+)} & =
    \begin{bmatrix}
      \one & c \\ 0 & 0
    \end{bmatrix}
    & = \Pi_1^{(+)} \Upsilon^{-1} \Pi_2^{(+)} & \quad & \text{ projects onto $\cZ_1^{(+)}$ along $\cZ_2^{(-)}$},
    \\
    \Lambda_{21}^{(-)}     & =
    \begin{bmatrix}
      0 & -c   \\
      0 & \one
    \end{bmatrix}
    & = \Pi_2^{(-)} \Upsilon^{-1} \Pi_1^{(-)} & \quad & \text{ projects onto $\cZ_2^{(-)}$ along $\cZ_1^{(+)}$},
    \\
    \Lambda_{21}^{(+)}     & =
    \begin{bmatrix}
      \one & 0 \\
      -c^* & 0
    \end{bmatrix}
    & = \Pi_2^{(+)} \Upsilon^{-1} \Pi_1^{(+)} & \quad & \text{ projects onto $\cZ_2^{(+)}$ along $\cZ_1^{(-)}$},
    \\
    \Lambda_{12}^{(-)}     & =
    \begin{bmatrix}
      0   & 0    \\
      c^* & \one
    \end{bmatrix}
    & = \Pi_1^{(-)} \Upsilon^{-1} \Pi_2^{(-)} & \quad & \text{ projects onto $\cZ_1^{(-)}$ along $\cZ_2^{(+)}$},
  \end{alignat*}
  where
  \begin{equation*}
    \Upsilon^{-1} =
    \begin{bmatrix}
      \one-cc^* & 0         \\
      0         & \one-c^*c
    \end{bmatrix}
    = \frac{4}{(2 +  S_2 S_1 + S_1 S_2 )}=\frac{4}{(\one+K)(\one+K^{-1})}.
  \end{equation*}
\end{proposition}
\begin{proof}
  We apply Prop.~\ref{prop:pair_proj}.
\end{proof}

We can reformulate Prop.~\ref{prop:asymp_compl_proof} as follows.
\begin{proposition}
  Let $\cZ_1$ be an maximally positive subspace and $\cZ_2$ an maximally negative space.
  Then they are complementary.
\end{proposition}
\begin{proof}
  By Prop.~\ref{pos3} there exist admissible involutions $S_1$ and $S_2$ such that $\cZ_1=\cZ_1^{(+)}$ and  $\cZ_2=\cZ_2^{(-)}$.
  Hence, it suffices to apply Prop.~\ref{prop:asymp_compl_proof}.
\end{proof}

\subsection{Pseudo-unitary generators}

Let $(\cW,Q)$ be a pre-pseudo-unitary space.

\begin{definition}
  \label{def:symgen}
  We say that a densely defined operator $B$ on $\cW$ \emph{infinitesimally preserves $Q$} if $B$ is the generator of a one-parameter group $\e^{-\im t B}$ on $\cW$ and
  \begin{equation}
    \label{eq:energy_space}
    \cinner{v}{Q B w} = \cinner{B v}{Q w},
    \quad
    v,w \in \Dom( B).
  \end{equation}
  If in addition $Q$ is non-degenerate, then we will say that $B$ is a \emph{pseudo-unitary generator}.
  The quadratic form defined by \eqref{eq:energy_space} will be called the \emph{energy} or \emph{Hamiltonian quadratic form of $B$ on $\Dom(B)$}.
\end{definition}

\begin{proposition}
  \label{symgen}
  Let $B$ be a generator of a one-parameter group on $\cW$.
  Then $\e^{-\im t B}$, $t \in \RR$, preserves~$Q$ if and only if $B$ infinitesimally preserves $Q$.
\end{proposition}
\begin{proof}
  Let us show $\Leftarrow$.
  Assume first that $v,w\in\Dom(B)$.
  Then
  \[
    \frac{\dif}{\dif t}\cinner{\e^{-\ri tB}v}{Q\e^{-\ri tB}w} = \ri\cinner{B\e^{-\ri tB}v}{Q\e^{-\ri tB}w}- \ri\cinner{\e^{-\ri tB}v}{QB\e^{-\ri tB}w} = 0.
  \]
  Hence
  \begin{equation}
    \label{bound}
    \cinner{\e^{-\ri tB}v}{Q\e^{-\ri tB}w}=\cinner{v}{Qw}.
  \end{equation}
  By the density of $\Dom(B)$ and the boundedness of $Q$ and $\e^{-\ri tB}$, (\ref{bound}) extends to the whole $\cW$.

  In the proof of $\Rightarrow$ we use the above arguments in the reverse order (with the exception of the density argument, which is not needed).
\end{proof}

The following proposition describes a large class of pseudo-unitary transformations and pseudo-unitary generators on Krein spaces.

\begin{proposition}
  \label{def:compatible_S-B}
  Suppose that $(\cW,Q)$ is a Krein space and $S_\bullet$ is an admissible involution with the corresponding scalar product $\cinner{\cdot}{\cdot}_\bullet$.
  If $B$ is a densely defined operator on $\cW$, self-adjoint in the sense of $\cW_\bullet$ and commuting with $S_\bullet$, then $B$ is a pseudo-unitary generator on $(\cW,Q)$ in the sense of Def.~\ref{def:symgen}.
\end{proposition}
\begin{proof}
  Clearly, $\e^{-\ri tB}$  is a unitary operator on $\cW_\bullet$ commuting with $S_\bullet$.
  Therefore, it is a pseudo-unitary transformation (see Def. \ref{pseudo-uni}).
\end{proof}

\begin{definition}
  \label{def:1}
  A densely defined operator $B$ on a Krein space $(\cW,Q)$ is called a \emph{stable pseudo-unitary generator} if it is similar to self-adjoint, $\Ker (B) = \{0\}$, and $\sgn(B)$ is an admissible involution.
  $B$ is called a \emph{strongly stable pseudo-unitary generator} if in addition it is invertible.
\end{definition}

In other words, a stable pseudo-unitary generator has a positive Hamiltonian and a strongly stable generator has a positive Hamiltonian bounded away from zero.

\subsection{Bisolutions and inverses}
\label{sub:Classical bisolutions and inverses}

Let $(\cW,Q)$ be a Krein space.
Then naturally $L^2(I,\cW)$ is also a Krein space with the Hermitian form
\begin{equation*}
  \cinner{v}{Qw} \defn \int_I \cinner[\big]{v(t)}{Q w(t)} \dif t,
\end{equation*}
and compatible scalar products
\begin{equation*}
  \cinner{v}{w}_\bullet \defn \int_I \cinner[\big]{v(t)}{w(t)}_\bullet \dif t.
\end{equation*}

Let $\{R(t,s)\}_{t,s \in I}$ be a strongly continuous pseudo-unitary evolution on $(\cW,Q)$.
We denote by $B(t)$ the (heuristic) generator of $R(t,s)$.
Recall that in Sect.~\ref{sub:Bisolutions and inverses of the evolution equation} we considered the (heuristic) Cauchy data operator $M=\partial_t + \im B(t)$.
Note that $M=\partial_t + \im B(t)$ is (heuristically) anti-$Q$-Hermitian on $L^2(I,W)$.
We will give a rigorous version of this statement a little later, in \eqref{anti-herm}.

In Sect.~\ref{sub:Bisolutions and inverses of the evolution equation} we introduced various inverses and bisolutions of $M$, considered as operators $C_\mathrm{c}(I,\cW) \to C(I,\cW)$.
In this subsection we add the Krein structure to the picture.

First, as in Subsect.~\ref{sub:Bisolutions and inverses of the evolution equation}, we define the \emph{Pauli--Jordan bisolution} $E^\PJ$, and the \emph{forward and backward inverses} $E^\vee$, resp.\ $E^\wedge$.

\begin{proposition}
  $E^\PJ$ is $Q$-Hermitian and the $Q$-adjoint of $E^\vee$ is contained in $-E^\wedge$.
  More precisely, for $v,w \in C_\mathrm{c}\bigl(I,\cW\bigr)$ we have
  \begin{align*}
    \cinner{v}{Q E^\PJ w} & = \cinner{E^\PJ v}{Q w}, \\ \cinner{v}{Q E^\wedge w} & = -\cinner{E^\vee v}{Q w}.
  \end{align*}
\end{proposition}

Consider now non-classical bisolutions and inverses.
If $t_\pm$ are finite, let us select two arbitrary admissible involutions $S_+,S_-$.
If $t_\pm=\pm\infty$, recall that we assumed that for large $\pm t,\pm s$ we have $R(t,s)=\e^{-\im(t-s)B_\pm}$.
We assume that $B_\pm$ are stable pseudo-unitary generators and we set
\begin{equation}
  S_\pm \defn \sgn\bigl(B_\pm\bigr).
  \label{porro1}
\end{equation}
Note that \eqref{porro1} implies that \eqref{porro} is satisfied and $S_\pm$ are admissible.

Define $\Pi_\pm^{(+)}, \Pi_\pm^{(-)}, \cZ_\pm^{(+)}, \cZ_\pm^{(-)},$ as well as $E_\pm^{(+)}$, $E_\pm^{(-)}$, as in Def.~\ref{def:smo}.
Note that, for any $t\in I$, $\cZ_+^{(+)}(t)$ and $\cZ_-^{(+)}(t)$ are maximally positive and $\cZ_+^{(-)}$ and $\cZ_-^{(-)}$ are maximally negative.
This immediately implies
\begin{proposition}
  $E_\pm^{(+)}$ and $E_\pm^{(-)}$ are $Q$-Hermitian.
  $E_\pm^{(+)}$ is $Q$-positive and $E_\pm^{(-)}$ is
  $Q$-negative.
  More precisely, for $v,w\in C_\mathrm{c}(I,\cW)$ we have
  \begin{alignat*}{2} \cinner{E_\pm^{(+)} w}{Q v} & = \cinner{w}{Q E_\pm^{(+)} v}, & \quad \cinner{v}{Q E_\pm^{(+)} v} & \geq 0, \\ \cinner{E_\pm^{(-)} w}{Q v} & = \cinner{w}{Q E_\pm^{(-)}v}, & \quad \cinner{v}{Q E_\pm^{(-)}v} & \leq 0.
  \end{alignat*}
\end{proposition}
\begin{proof}
  Consider for definiteness the case of finite $t_\pm$.
  It holds
  \begin{align*}
    \cinner{v}{Q E_\pm^{(+)} v} & = \cinner[\bigl]{\Pi_\pm^{(+)}w}{Q \Pi_\pm^{(+)}w}, \\
    \cinner{v}{Q E_\pm^{(-)} v} & = \cinner[\bigl]{\Pi_\pm^{(-)}w}{Q \Pi_\pm^{(-)}w},
  \end{align*}
  where $w=\int_I R(t_\pm,t)v(t)\dif t$.
\end{proof}

What is more remarkable, under the present assumptions by Prop.~\ref{basic0} (2)
asymptotic complementarity holds automatically
for both $\bigl(\cZ_+^{(+)},\cZ_-^{(-)}\bigr)$ and $\bigl(\cZ_+^{(-)},\cZ_-^{(+)}\bigr)$.
Therefore, we can define the inverses $E^\Feyn$ and $E^\aFeyn$, as in Def.~\ref{def-feyn} and~\ref{def-afeyn}.
\begin{proposition}
  The $Q$-adjoint of $E^\Feyn$ is contained in $-E^\aFeyn$.
  More precisely, for $v,w \in C_\mathrm{c}\bigl(I,\cW\bigr)$ we have
  \begin{equation*}
    \cinner{E^\Feyn w}{Q v} = -\cinner{w}{Q E^\aFeyn v}.
  \end{equation*}
  \label{pos6}
\end{proposition}
 \begin{proof}
   By \eqref{prepre} we have
    $R(s,t)^{*Q}=R(t,s)$. Clearly, $\cZ_\pm^{(+)}(t)^{\perp
      Q}=\cZ_\pm^{(-)}(t)$, and hence by Prop. \ref{basic0}
  \[
    \Lambda^{\Feyn(+)}(t)^{*Q}=\Lambda^{\aFeyn(+)}(t),\quad
    \Lambda^{\Feyn(-)}(t)^{*Q}=\Lambda^{\aFeyn(-)}(t).
  \]
  Now,
  \[
    \cinner{E^\Feyn w}{Q v} = \int_I\cinner[\big]{w(t)}{Q E^\Feyn(s,t)^{*Q} v(s)}\dif t\dif s
  \]
  and
  \begin{align*}
   E^\Feyn(s,t)^{*Q}
    & = \theta(s-t) \Lambda^{\Feyn(+)}(t)^{*Q} R(s,t)^{*Q} -\theta(t-s) \Lambda^{\Feyn(-)}(t)^{*Q} R(s,t)^{*Q} \\
    & = \theta(s-t) \Lambda^{\aFeyn(+)}(t) R(t,s)-\theta(t-s)\Lambda^{\aFeyn(-)}(t) R(t,s) \\
    & = \theta(s-t) R(t,s) \Lambda^{\aFeyn(+)}(s) -\theta(t-s)R(t,s)\Lambda^{\aFeyn(-)}(s) \\
    & = -E^\aFeyn(t,s).
    \tag*{$\square$}
  \end{align*}
\end{proof}

With the choice \eqref{porro1}, the bisolutions $E_\pm^{(+)}$ are called the \emph{in/out positive frequency bisolutions}, the bisolutions $E_\pm^{(-)}$ are called the \emph{in/out negative frequency bisolutions}, and the inverses $E^\Feyn$, resp.\ $E^\aFeyn$ are called the \emph{Feynman}, resp.\ the \emph{anti-Feynman inverse}.

\subsection{The Cauchy data operator in the Krein setting}
\label{Evolutions preserving a Hermitian form}

The Cauchy data operator $M$ is the sum of two unbounded operators: $\partial_t$ and $\im B(t)$.
Therefore, it is not easy to choose its domain.
We will discuss two possible approaches to this question.
In this subsection we will describe the ``operator approach''.
Subsects.~\ref{Nested pre-pseudo-unitary pairs} and~\ref{Evolutions on nested pre-pseudo-unitary pairs} will discuss the ``quadratic form approach''.

Suppose that $(\cW_0,Q)$ is a pre-pseudo-unitary space.

\begin{theorem}
  \label{thm:evolution1}
  Suppose that $I\ni t\mapsto B(t)$ is an operator on a Krein space $\cW_0$.
  Suppose that $\cW_1$ is a Hilbertizable space densely and continuously embedded in $\cW_0$ and all the assumptions of Thm.~\ref{thm:evolution} are satisfied.
  In addition, assume that for all $t\in I$ the operators $B(t)$ infinitesimally preserve $Q$.
  Then the evolution $R(t,s)$ on $\cW_0$ preserves $Q$.
\end{theorem}
\begin{proof}
  Recall that among the assumptions of
  Thm.~\ref{thm:evolution-per} there is
  $\cW_1 \subset \Dom( B(t))$.
  Besides, for any $w\in\cW_1$ this theorem implies that
  \begin{align*}
    \im\partial_t R(t,s) w & = B(t) R(t,s) w.
  \end{align*}
  Hence the above theorem follows by repeating the proof of Prop.~\ref{symgen}, where we use $\cW_1$ instead of $\Dom(B)$.
\end{proof}

Suppose that $(\cW_0,Q)$ is a Krein space and $I\ni t\mapsto B(t)$ is a family of pseudo-unitary generators satisfying the assumptions of Thm.~\ref{thm:evolution1}.
We can treat the Cauchy data operator
\begin{equation*}
  M=\partial_t+\im B(t)
\end{equation*}
as a densely defined operator on the Krein space $L^2(I,\cW_0)$ with the domain $ C_\mathrm{c}(I,\cW_1) \cap C_\mathrm{c}^1(I,\cW_0)$.
Now we can give a rigorous meaning to its anti-$Q$-Hermiticity: for $v,w \in C_\mathrm{c}(I,\cW_1) \cap C_\mathrm{c}^1(I,\cW_0)$,
\begin{equation}
  \label{anti-herm}
  \cinner[\Big]{w}{Q\bigl(\partial_t+\im B(t)\bigr) v} = -\cinner[\Big]{\bigl(\partial_t+\im B(t)\bigr) w}{Q v}.
\end{equation}

For the remaining part of this subsection we assume that $I$ is finite and two admissible involutions $S_+,S_-$ have been chosen.

\begin{proposition}
  $E^\Feyn$ and $E^\aFeyn$ extend to bounded operators on $L^2(I,\cW_0)$, their ranges are dense and their nullspaces are  $\{0\}$.
  \label{pos9}
\end{proposition}
\begin{proof}
  The boundedness is obvious.
  By Prop.~\ref{pos8} (1) for any $v\in C_\mathrm{c}(I,\cW_1)\cap C_\mathrm{c}^1(I,\cW_0)$ we have
  \begin{equation*}
    E^\Feyn \bigl(\partial_t + \im B(t)\bigr) v = v.
  \end{equation*}
  Hence $\cR(E^\Feyn)$ contains $ C_\mathrm{c}(I,\cW_1)\cap C_\mathrm{c}^1(I,\cW_0)$, which is dense in $L^2(I,\cW_0)$.
  The same argument shows that $\cR(E^\aFeyn)$ is dense in $L^2(I,\cW_0)$.
  Now
  \[
    \cN(E^\Feyn)=\cR(E^{\Feyn*Q})^{\perp Q}= \cR(E^{\aFeyn})^{\perp Q}=\{0\}.
    \tag*{$\square$}
  \]
\end{proof}

Thus by Prop.~\ref{pos9}, for finite $I$ we can define operators with dense domains
\[M^\Feyn \defn (E^\Feyn)^{-1},\quad M^\aFeyn \defn (E^\aFeyn)^{-1}.\]
They satisfy
\[(M^\Feyn)^{*Q}= M^\aFeyn\]
and $0$ belongs to their resolvent set.

\subsection{Nested pre-pseudo-unitary pairs}
\label{Nested pre-pseudo-unitary pairs}

In this and the following subsection we describe the ``quadratic form approach'' to dynamics on pre-pseudo-unitary spaces.
Such an approach usually requires weaker assumptions.

In this approach the starting point is a nested pair of Hilbertizable spaces equipped with a Hermitian pairing.
The pre-pseudo-unitary space is then obtained by interpolation.

Let us describe this simple construction in detail.
In the next subsection we will describe evolutions on such nested pairs.

\begin{definition}
  \label{nested pre-pseudo-unitary pair}
  Let $\lambda>0$.
  A \emph{nested pre-pseudo-unitary pair} $(\cW_{-\lambda},\cW_{\lambda},Q)$ consists of a pair of Hilbertizable spaces $\cW_{-\lambda}$, $\cW_{\lambda}$, where $\cW_{\lambda}$ is densely and continuously embedded in $\cW_{-\lambda}$ and a \emph{Hermitian pairing}, that is a sesquilinear form
  \begin{equation*}
    \cW_{\lambda} \times \cW_{-\lambda} \ni (v,w) \mapsto \cinner{v}{Qw} \in \CC,
  \end{equation*}
  which is Hermitian on $\cW_{\lambda}$, i.e.,
  \begin{equation}
    \label{invert.}
    \cinner{v}{Qw} = \bar{\cinner{w}{Qv}},
    \quad
    v,w \in \cW_{\lambda},
  \end{equation}
  and bounded, i.e., for some (hence all) compatible norms $\|\cdot\|_{\lambda,\bullet}$
  and  $\|\cdot\|_{-\lambda,\bullet}$ on $\cW_{\lambda}$, resp.\ $\cW_{-\lambda}$ there exists $C_\bullet$ such that
  \begin{equation}
    \label{bund}
    \abs{\cinner{v}{Qw}} \leq C_\bullet\norm{v}_{\lambda,\bullet} \norm{w}_{-\lambda,\bullet}.
  \end{equation}
\end{definition}

In what follows, let $(\cW_{-\frac12},\cW_{\frac12},Q)$ be a nested pre-pseudo-unitary pair.
(The parameter $\frac12$ can be changed to any positive number, it is chosen here in view of our future applications.)
Let $\cW_\lambda$, $\lambda\in[-\frac12,\frac12]$, be the Hilbertizable spaces obtained by interpolation from the nested pair $(\cW_\frac12,\cW_{-\frac12})$.

\begin{proposition}
  For $\lambda\in[-\frac12,\frac12]$, there exists a unique family of Hermitian pairings $Q_\lambda$
  \begin{equation*}
    \cW_{\lambda} \times \cW_{-\lambda} \ni (v,w) \mapsto \cinner{v}{Q_\lambda w}\in\CC.
  \end{equation*}
  such that $Q_\frac12=Q$ and
  if $-\frac12\leq\lambda_1\leq\lambda_2\leq\frac12$, $v\in\cW_{-\lambda_1}\subset\cW_{-\lambda_2}$, $w\in\cW_{\lambda_2}\subset\cW_{\lambda_1}$, then
  \begin{equation}
    \label{order00}
    \cinner{v}{Q_{\lambda_1} w} = \cinner{v}{Q_{\lambda_2} w}.
  \end{equation}
\end{proposition}
\begin{proof}
  By (\ref{bund}), $Q$ can be viewed as a bounded map from $\cW_{-\frac12}$ to the antidual of $\cW_{\frac12}$.
  The antidual of $\cW_{\frac12}$ coincides with $\cW_{-\frac12}$.
  Hence, $Q\in B(\cW_{-\frac12})$.

  The restriction of $Q$ to $\cW_{\frac12}$, by (\ref{bund}) and (\ref{invert.}) is a bounded map to the antidual of $\cW_{-\frac12}$, which is
  $\cW_{\frac12}$.
  Hence $Q\in B(\cW_{\frac12})$.

  By interpolation, that is, Prop.~\ref{interpolation} (2), for $\lambda\in[-\frac12,\frac12]$, the restriction of $Q$ to $\cW_\lambda$ is bounded.
\end{proof}

In what follows we drop the subscript $\lambda$ from $Q_\lambda$, which is allowed because of (\ref{order00}).
In particular, for $\lambda=0$, we obtain a bounded Hermitian form on $\cW_0$:
\begin{equation*}
  \cW_0 \times \cW_0 \ni (v,w) \mapsto \cinner{v}{Q w}\in\CC.
\end{equation*}

\subsection{Evolutions on nested pre-pseudo-unitary pairs}
\label{Evolutions on nested pre-pseudo-unitary pairs}

Recall that in Thm.~\ref{thm:evolution1} we constructed a pre-pseudo-unitary dynamics starting from a pre-pseudo-unitary space $\cW_0$ and its subspace $\cW_1$.
In this subsection we give a slightly different construction of such a dynamics which starts from a nested pre-pseudo-unitary pair $(\cW_{\frac12},\cW_{-\frac12},Q)$.

\begin{definition}
  \label{defini}
  Let $(\cW_{\frac12},\cW_{-\frac12},Q)$ be a nested pre-pseudo-unitary pair.
  If a bounded operator $R$ on $\cW_{-\frac12}$, restricts to a bounded operator on $\cW_{\frac12}$, and \[\cinner{Rv}{QRw}=\cinner{v}{Qw},\quad v\in\cW_{-\frac12},\quad w\in\cW_{\frac12},\] then we say that \emph{$R$ preserves $(\cW_{\frac12},\cW_{-\frac12},Q)$}.
\end{definition}

Applying complex interpolation we obtain

\begin{proposition}
  Suppose that $R$ preserves $(\cW_{\frac12},\cW_{-\frac12},Q)$.
  Then for $0\leq\lambda\leq\frac12$ it restricts to an operator preserving $(\cW_{\lambda},\cW_{-\lambda},Q)$.
  In particular, $R$ preserves $Q$ on $\cW_0$ in the usual sense.
  \label{preser}
\end{proposition}

\begin{definition}
  \label{def:symgen1}
  Suppose that $B$ is an operator on $\cW_{-\frac12}$ with domain containing $\cW_\frac12$.
  We say that $B$ \emph{infinitesimally preserves} $(\cW_{\frac12},\cW_{-\frac12},Q)$ if $B$ is a generator of a group on $\cW_{-\frac12}$, its part $\tilde B$ in $\cW_{\frac12}$ is a generator of a group on $\cW_{\frac12}$, and
  \begin{equation}
    \label{eq:energy_space1}
    \cinner{Bv}{Qw} = \cinner{v}{QBw},
    \quad
    v,w \in \cW_{\frac12}.
  \end{equation}
  The quadratic form defined by \eqref{eq:energy_space1} is called the \emph{energy} or \emph{Hamiltonian quadratic form of $B$ on~$\cW_{\frac12}$}.
\end{definition}

\begin{proposition}
  \label{symgen1}
  Suppose that $B$ is an operator on $\cW_{-\frac12}$ that infinitesimally preserves $(\cW_{-\frac12}, \cW_{\frac12},Q)$.
  Then $\e^{-\im t B}$, $t \in \RR$, preserves $(\cW_{-\frac12}, \cW_{\frac12},Q)$.
\end{proposition}
\begin{proof}
  First we check that
  \begin{equation}\label{preo}
    \cinner{\e^{-\ri tB}v}{Q\e^{-\ri tB}w}=\cinner{v}{Qw},\quad
    v,w\in\cW_{\frac12}.\end{equation}
  Then, by continuity, we extend \eqref{preo} to $v\in\cW_{-\frac12}$.
  \end{proof}

The following theorem can be viewed as an alternative to Thm.~\ref{thm:evolution1}:
\begin{theorem}
  \label{thm:evolution1.}
  Suppose that all assumptions of Thm.~\ref{thm:evolution-per} are satisfied where $\cW_0$ is replaced with $\cW_{-\frac12}$ and $\cW_1$ is replaced with $\cW_{\frac12}$.
  In addition, assume that for all $t\in I$ the operators $B(t)$ infinitesimally preserve $(\cW_{-\frac12},\cW_{\frac12},Q)$.
  Then the evolution $R(t,s)$ preserves $(\cW_{-\frac12},\cW_{\frac12},Q)$.
  In particular, it is pre-pseudo-unitary on $\cW_0$.
\end{theorem}

Let us compare the constructions of Thm.~\ref{thm:evolution1} and of Theorem~\ref{thm:evolution1.}.
In both cases we obtain a (pre-)pseudo-unitary evolution on a (pre-)pseudo-unitary space.
However, in the former case we have a fixed space $\cW_1$ contained in $\Dom(B(t))$ for all $t$.
In the latter case, we do not have information about the domain of the generator of the evolution on $\cW_0$, that is, of the part of $B(t)$ in $\cW_0$.
On the other hand, in practice the assumptions of Theorem~\ref{thm:evolution1.} can be weaker.

\section{Abstract Klein--Gordon operator}
\label{Abstract Klein-Gordon operator}

The usual Klein--Gordon operator acts on, say, $C_\mathrm{c}^\infty(M)$, where $M$ is a Lorentzian manifold, and is given by the expression \eqref{hilbert1a}.
$K$ can be interpreted as a Hermitian operator in the sense of the
Hilbert space $L^2(M)$. (One of the main ideas of our paper is the
usefulness of this interpretation.)
After the identification of $M$ with $I\times \Sigma$, where $I$
corresponds to the time variable and $\Sigma$  describes the spatial
variables, we can identify $L^2(M)$ with $L^2(I,\cK)$, where $\cK=
  L^2(\Sigma)$, see \eqref{hilbert} for more details.
Then, formally, the operator $K$ is given by the expression
\begin{equation}
  K \defn \bigl(D_t+W^*(t)\bigr)\frac{1}{\alpha^2(t)} \bigl(D_t+W(t)\bigr)-L(t),\label{kleingordon}
\end{equation}
where $\alpha(t)$ involves the metric tensor, $W(t)$ consists  mostly of the $0$th component of the potential and $L(t)$ is a magnetic Schr\"odinger operator on $\Sigma$.
We will describe this identification in more detail in Subsect.~\ref{Foliating the spacetime}.

In this section we study (\ref{kleingordon}) in an abstract setting.
We are interested in various inverses and bisolutions of $K$.
We treat $\cK$ as an abstract Hilbert space and $L(t)$, $\alpha(t)$, $W(t)$ as given abstract operators.
The results of this section will be applied to the usual Klein--Gordon operator in Sect.~\ref{Quantum Field Theory on curved space-times}.

In order to study propagators associated with the abstract Klein--Gordon operator, we first introduce a certain scale of Hilbertizable spaces $\cW_\lambda$.
Each member of this scale is the direct sum of two Sobolev-type spaces based on $\cK$ describing the ``configurations'' and the ``momenta''.
The space $\cW_0$ has the structure of a Krein space and will play the central role in quantization.
The Cauchy data for $K$ on $\cW_0$  undergo a certain pseudo-unitary evolution whose generator $B(t)$ is made out of $L(t),W(t),\alpha^2(t)$.
After imposing boundary conditions we can define various propagators, first associated with the Cauchy data operator $M=\partial_t+\ri B(t)$, and then associated with the operator $K$ itself.

Note that the formula (\ref{kleingordon}) does not give a rigorous definition of a unique closed operator.
Actually, the analysis of possible closed realizations of $K$ and the corresponding inverses is quite subtle and depends strongly on whether $I$ is finite or not.

 In the case of a finite $I$ one first needs to impose
  appropriate  boundary
conditions at the initial time $t_-$ and the final time $t_+$.
These conditions  lead to a construction of  $E^\Feyn$ and
  $E^\aFeyn$, which are bounded inverses of $M$. They are then used to
  define $G^\Feyn$ and $G^\aFeyn$, which are bounded inverses of $K$.
Inverting them we obtain a well-posed realization of the Klein--Gordon operator $K$.

The situation is different and less understood if $I=\mathbb{R}$.
The Feynman and anti-Feynman inverse can be constructed, however they are not bounded on $L^2(\mathbb{R},\cK)$.
We conjecture that there exists a distinguished self-adjoint realization $K^{\mathrm{s.a.}}$ of $K$ such that these inverses are the boundary values of the resolvent of
$K^{\mathrm{s.a.}}$ from above and below at zero.
The conjecture can be easily shown in some special cases, e.g., if $K$ is stationary.
We describe some arguments in favor of the conjecture, notably, we sketch a possible construction of the resolvent.
There exist recent papers that show this conjecture if $K$ corresponds to the Klein--Gordon operator on asymptotically Minkowskian spaces satisfying a non-trapping condition.

Throughout the section we need to overcome a number of technical issues.
First, it is convenient to assume that the Hamiltonian used in the construction of the phase space is bounded away from zero, or in physical terms, that the mass is strictly positive.
However, physical systems may have a zero mass.
This is solved by assuming that the phase space is constructed not directly from the Hamiltonian $H(t)$, but from $H_0(t)$ which differs by a constant $b$.

Another problem is the regularity.
Assumption \ref{ass.I} allows us only to perform the basic construction.
We introduce the additional Assumption \ref{assII}$(\rho)$, which for $\rho=0$ coincides with Assumption \ref{ass.I} and for $\rho>0$ guarantees additional regularity.
With strengthened hypotheses we are able to show some desired properties of the propagators and of the Klein--Gordon operator.

Note that the above issue essentially disappears if we assume that everything is smooth, therefore it can be considered as purely academic, of interest only to specialists in operator theory.
Nevertheless, we try to give an honest (if not optimal) treatment of this question.

\subsection{Basic assumptions on the abstract Klein--Gordon quadratic form}

Throughout this section we assume that $\cK$ is a Hilbert space and for $t\in I^\cl$ we are given the following operators on $\cK$:
\begin{enumerate}
  \item self-adjoint $ L(t)$ for which there exists $b\in\mathbb{R}$ such that $L_0(t) \defn L(t)+b$ are positive invertible,
  \item bounded invertible self-adjoint $\alpha(t)$,
  \item operator $W(t)$.
\end{enumerate}

We will say that an  operator-valued function $I\ni t\mapsto A(t)\in B(\cK)$ is \emph{absolutely norm continuous} if there exists $c\in L^1(I)$ such that $c\geq0$ and
\begin{equation*}
  \|A(t)-A(s)\|\leq \int_t^s c(\tau)\dif\tau,\quad t\leq s,\quad t,s\in I.
\end{equation*}

Here is the basic assumption that we will use in this section.
\begin{assumption}
  \mbox{}
  \begin{enumerate}
    \item For any $t\in I$ there exist $0<c_1\leq c_2$ such that
          \begin{equation}
            \label{byby}
            c_1 L_0(0) \leq L_0(t)\leq c_2 L_0(0)
          \end{equation}
          and $I^\cl\ni t\mapsto L_0(0)^{-\frac12} L_0(t)L_0(0)^{-\frac12}\in B(\cK)$ is absolutely norm continuous.
    \item $I^\cl\ni t\mapsto \alpha^2(t)\in B(\cK)$ is absolutely norm continuous.
    \item $I^\cl\ni t\mapsto W(t)L_0(t)^{-\frac12}\in B(\cK)$ is absolutely norm continuous and there exists $a<1$ such that
          \begin{equation*}
            \|\alpha(t)^{-1} W(t)L_0(t)^{-\frac12}\|\leq a.
          \end{equation*}
  \end{enumerate}
  \label{ass.I}
\end{assumption}

For $\beta\in\RR$, $t\in I^\cl$, define the scales of Hilbert spaces
\begin{equation}
  \cK_{\beta,t} \defn L_0(t)^{-\frac\beta2}\cK.
  \label{cK}
\end{equation}
Then by (\ref{byby}) and Subsect.~\ref{sub:Interpolation between Hibertizable spaces}, for $\beta\in[-1,1]$, $\cK_{\beta,t}$ is compatible with $\cK_{\beta,0}$.
Thus we obtain the scale of Hilbertizable spaces
\begin{equation*}
  \cK_\beta,\quad \beta\in[-1,1].
\end{equation*}

Assumption \ref{ass.I} seems insufficient to define the Klein--Gordon operator.
However, we can define the \emph{Klein--Gordon quadratic form} by setting
\begin{align}
  \notag \cinner{f_1}{Kf_2} & =\int_I\cinner[\Big]{\bigl(D_t+W(t)\bigr)f_1(t)}{\frac{1}{\alpha^2(t)} \bigl(D_t+W(t)\bigr)f_2(t)}\dif t \\  & \quad-\int_I\cinner[\big]{L_0(t)^{\frac12}f_1(t)}{L_0(t)^{\frac12}f_2(t)}\dif t+b\int_I\cinner[\big]{f_1(t)}{f_2(t)}\dif t, \label{kleingordon1}
\end{align}
$f_1,f_2\in   C_\mathrm{c}(I,\cK_1)\cap C_\mathrm{c}^1(I,\cK_0)$.
Note that $K$ is a Hermitian form in the sense of the Hilbert space $L^2(I)\otimes\cK\simeq L^2(I,\cK)$.
Formally, it corresponds to the operator given by the expression \eqref{kleingordon}.

Unfortunately, $K$ is not a semibounded form, hence the usual theory of quadratic forms does not apply.
Therefore, it is not easy to interpret $K$ as a closed operator on $ L^2(I,\cK)$.
We will come back to this question in Subsect.~\ref{The abstract Klein-Gordon operator} under more restrictive assumptions.

\subsection{Pseudo-unitary evolution on the space of Cauchy data}

The following analysis is essentially an adaptation of \cite{derezinski-siemssen:static,derezinski-siemssen:propagators} to the abstract setting.

We consider the scale of Hilbertizable spaces
\begin{equation}
  \cW_\lambda=\cK_{\lambda+\frac12}\oplus\cK_{\lambda-\frac12},\quad \lambda\in\left[-\frac12,\frac12\right].
  \label{cW}
\end{equation}
Of special importance are
\begin{alignat*}{3}
  \text{the energy space}\quad          \cW_\en   & \defn \cW_{\frac12} = \cK_1\oplus\cK_0,          \\
  \text{the dynamical space}\quad       \cW_\dyn  & \defn \cW_0 = \cK_{\frac12}\oplus\cK_{-\frac12}, \\
  \text{and the dual energy space}\quad \cW_\en^* & \defn \cW_{-\frac12} = \cK_0\oplus\cK_{-1}.
\end{alignat*}
For $
  \begin{bmatrix}
    u_1 \\u_2
  \end{bmatrix}
  \in\cW_{-\lambda}$ and $
  \begin{bmatrix}
    v_1 \\v_2
  \end{bmatrix}
  \in\cW_{\lambda}$ with $|\lambda|\leq\frac12$ we introduce the pairing defined by the charge operator $Q=
  \begin{bmatrix}
    0 & \one \\\one & 0
  \end{bmatrix}
$.
In other words,
\begin{equation*}
  \cinner[\Big]{%
    \begin{bmatrix}
      u_1 \\u_2
    \end{bmatrix}%
  }{Q
    \begin{bmatrix}
      v_1 \\v_2
    \end{bmatrix}%
   } =\cinner{u_1}{v_2}+\cinner{u_2}{v_1}.
\end{equation*}

Note that $(\cW_{-\frac12},\cW_{\frac12},Q)$ is a nested pseudo-unitary pair, see Def.~\ref{nested pre-pseudo-unitary pair}.
Moreover, $ (\cW_0,Q)$ is a Krein space.
Indeed,
\begin{equation*}
  S_t \defn
  \begin{bmatrix}
    0 & L_0(t)^{-\frac12} \\ L_0(t)^{\frac12} & 0
  \end{bmatrix}
\end{equation*}
is a bounded self-adjoint involution on the Hilbert space
\begin{equation}
  \label{charge1}
  \cW_{0,t}=
  L_0(t)^{-\frac14}\cK\oplus L_0(t)^{\frac14}\cK.
\end{equation}
\eqref{charge1} is compatible with $\cW_0$.
Hence $S_t$ is an admissible involution.

Introduce the \emph{Hamiltonians}
\begin{align*}
  H(t) & =
           \begin{bmatrix}
             L(t) & W^*(t) \\W(t) & \alpha^2(t)
           \end{bmatrix}
  ,\quad H_0(t)=
                 \begin{bmatrix}
                   L_0(t) & W^*(t) \\W(t) & \alpha^2(t)
                 \end{bmatrix}
  .
\end{align*}

\begin{proposition}
  Suppose Assumption \ref{ass.I} holds.
  There exist $0<c_1\leq c_2$ such that on $\cK\oplus\cK$ we have
  \begin{align*}
    c_1
       \begin{bmatrix}
         L_0(t) & 0 \\0 & \one
       \end{bmatrix}
    \leq & H_0(t) \leq c_2
                          \begin{bmatrix}
                            L_0(t) & 0 \\0 & \one
                          \end{bmatrix}
    ,                                        \\ c_1
    \begin{bmatrix}
      \one & 0 \\0 & L_0(t)
    \end{bmatrix}
    \leq & QH_0(t)Q \leq c_2
                            \begin{bmatrix}
                              \one & 0 \\0 & L_0(t)
                            \end{bmatrix}
    .
  \end{align*}
  Therefore,
  \begin{align*}
    \cW_{\frac12,t} \defn H_0(t)^{-\frac12}(\cK\oplus\cK) & \quad\text{ is compatible with }\cW_{\frac12}, \\ \cW_{-\frac12,t} \defn (Q H_0(t)Q)^{\frac12}(\cK\oplus\cK) & \quad\text{ is compatible with }\cW_{-\frac12}.
  \end{align*}
  In particular, $H_0(t)$ is a positive operator on $\cK\oplus\cK$ with the form domain $\cW_{\frac12}$.
\end{proposition}

Let $\cW_{\lambda,t}$, $\lambda\in\mathbb{R}$, denote the scale of Hilbert spaces defined by interpolation from the nested pair $ \cW_{\frac12,t}$, $ \cW_{-\frac12,t}$.
Clearly, $\cW_{\lambda,t}$ are compatible with the Hilbertizable spaces $\cW_{\lambda}$ for $\lambda\in[-\frac12,\frac12]$.

Introduce the \emph{generators}
\begin{align*}
  B(t) & \defn QH(t)=
                      \begin{bmatrix}
                        W(t) & \alpha^2(t) \\L(t) & W^*(t)
                      \end{bmatrix}
  ,\quad B_0(t) \defn QH_0(t)=
                               \begin{bmatrix}
                                 W(t) & \alpha^2(t) \\L_0(t) & W^*(t)
                               \end{bmatrix}
  .
\end{align*}

\begin{proposition}
  \label{kps1}
  Suppose Assumption \ref{ass.I} holds.
  Then for any $\lambda\in\mathbb{R}$, $B_0(t)$ is a unitary operator from $ \cW_{\frac12+\lambda,t}$ to $ \cW_{-\frac12+\lambda,t}$.
  Besides, it is a self-adjoint operator in the sense of $\cW_{-\frac12+\lambda,t}$ with the domain $ \cW_{\frac12+\lambda,t}$.
  Therefore,
  \begin{equation}
    \cW_{\lambda,t}=|B_0(t)|^{-\lambda+\frac12}\cW_{\frac12,t}=|B_0(t)|^{-\lambda-\frac12}\cW_{-\frac12,t}.
    \label{propi}
  \end{equation}
\end{proposition}
\begin{proof}
  We drop $0$ and $(t)$ from $H_0(t),B_0(t)$.
  First note that $H$ is bounded from $\cK_{\frac12}\oplus\cK_0$ to $\cK_{-\frac12}\oplus\cK_0$.
  Hence $B$ is bounded from $\cW_{\frac12}=\cK_{\frac12}\oplus\cK_0$ to $\cW_{-\frac12}=\cK_0\oplus\cK_{-\frac12}$.
  Now,
  \begin{equation*}
    \cinner{Bu}{Bv}_{-\frac12,t} = \cinner{QHu}{(QHQ)^{-1}QHv} = \cinner{u}{Hv}=\cinner{u}{v}_{\frac12,t}.
  \end{equation*}
  This proves the unitarity of $B$ from $ \cW_{\frac12,t}$ to $ \cW_{-\frac12,t}$.

  Let $u,v\in \cW_{\frac12}$.
  Then
  \begin{align*}
    \cinner{u}{Bv}_{-\frac12,t}=
    \cinner{u}{(QHQ)^{-1}QHv} = \cinner{QHu}{(QHQ)^{-1}v}= \cinner{Bu}{v}_{-\frac12,t}
  \end{align*}
  proves the Hermiticity in the sense of $\cW_{-\frac12,t}$ with the domain $ \cW_{\frac12,t}$.
  Clearly, an invertible Hermitian operator is self-adjoint.

  Now we obtain
  \begin{equation*}
    \cW_{-\frac12,t}=|B|\cW_{\frac12,t},
  \end{equation*}
  from which (\ref{propi}) and all the remaining statements of the proposition follow.
\end{proof}

\begin{proposition}
  \label{exid}
  Suppose Assumption \ref{ass.I} holds.
  Then $I\ni t\mapsto B(t):\cW_{\frac12}\to\cW_{-\frac12}$ satisfies the assumptions of Thm ~\ref{thm:evolution1.}.
  Therefore, it defines an evolution $R(t,s)$ on $\cW_\lambda$, $-\frac12\leq\lambda\leq\frac12$, pseudo-unitary in the sense of $(\cW_{-\lambda},\cW_{\lambda},Q)$.
  In particular, the evolution $R(t,s)$ is pseudo-unitary on $(\cW_0,Q)$.
\end{proposition}
\begin{proof}
  First we check that $t\mapsto B_0(t)\in B\bigl(\cW_{\frac12},\cW_{-\frac12}\bigr)$ is norm continuous.
  Besides, $B_0(t)$ is self-adjoint in the sense of the the scalar products $\cinner{\cdot}{\cdot}_{-\frac12,t}$ and $\cinner{\cdot}{\cdot}_{\frac12,t}$.

  There exists $c\in L^1(I)$ such that for $t,s\in I^\cl$,
  \[
    \left\|H_0(s)^{-\frac12}\bigl(H_0(t)-H_0(s)\bigr) H_0(s)^{-\frac12} \right\|\leq 2\int_s^t c(u)\dif u.
  \]
  Therefore,
  \[
    \left\|H_0(s)^{-\frac12}H_0(t) H_0(s)^{-\frac12} \right\| \leq \exp\Big(2\int_s^t c(u)\dif u\Big),
  \]
  hence $\| v\|_{\frac12,t}= \left\| H_0(t)^{\frac12} H_0(s)^{-\frac12}\right\|\| v\|_{\frac12,s} \leq \exp\left(\int_t^s c(\tau)\dif\tau\right)\|v\|_{\frac12,s}$.

  Similarly, there exists $c\in L^1(I)$ such that for $t,s\in I^\cl$,
  \[
    \left\|(QH_0(s)Q)^{-\frac12}\bigl(QH_0(t)Q-QH_0(s)Q\bigr) (QH_0(s)Q)^{-\frac12} \right\|\leq 2\int_s^t c(u)\dif u.
  \]
  Therefore,
  \begin{align*}\MoveEqLeft
    \left\|(QH_0(t)Q)^{\frac12}(QH_0(s) Q)^{-1} (QH_0(t)Q)^{\frac12} \right\| \\
    & = \left\|(QH_0(s)Q)^{-\frac12}QH_0(t) Q (QH_0(s)Q)^{-\frac12} \right\|
    \leq \exp\Big(2\int_s^t c(u)\dif u\Big),
  \end{align*}
  hence $\| v\|_{-\frac12,s}= \left\| H_0(s)^{-\frac12} H_0(t)^{\frac12}\right\|\| v\|_{-\frac12,t} \leq \exp\left(\int_t^s c(\tau)\dif\tau\right)\|v\|_{-\frac12,t}$.

  Thus the assumptions of Thm.~\ref{thm:evolution} are satisfied and $B_0(t)$ defines a dynamics on $\cW_\lambda$ for $-\frac12\leq\lambda\leq\frac12$.

  The perturbation $B(t)-B_0(t)$ is bounded.
  Therefore, the assumptions of Thm~\ref{thm:evolution-per} are satisfied, and $B(t)$ also defines a dynamics.

  Finally, both $B_0(t)$ and $B(t)$ infinitesimally preserve the pseudo-unitary nested pair $(\cW_{\frac12},\cW_{-\frac12},Q)$.
  Hence the assumptions of Thm~\ref{thm:evolution1.} hold
  and $R(t,s)$ is pseudo-unitary in the sense of  $(\cW_{\frac12},\cW_{-\frac12},Q)$.
\end{proof}

\subsection{Propagators on a finite interval}
\label{Propagators for a finite interval}

Assume that $I=\mathopen{]}t_-,t_+\mathclose{[}$ is finite.
Suppose Assumption \ref{ass.I} holds.
Let $R(t,s)$ be the corresponding pseudo-unitary evolution on the Krein space $\cW_0$, whose existence is guaranteed by Prop.~\ref{exid}.
Thus we are now in the setting of Subsect.~\ref{sub:Classical bisolutions and inverses}.

First we define the propagators for the Cauchy data.
The classical propagators $E^\PJ$, $E^\vee$, $E^\wedge$ are introduced as in Def.~\ref{def:classical-propagators}.
Then we choose two admissible involutions $S_+$, $S_-$ on $\cW_0$.
We then define the non-classical propagators $E_\pm^{(+)}$, $E_\pm^{(+)}$, $E^\Feyn$, $E^\aFeyn$ as in Defs.~\ref{def:smo}, \ref{def-feyn} and~\ref{def-afeyn}.
(Recall that the asymptotic complementarity is automatically satisfied.)
All the propagators for the Cauchy data are bounded operators on $L^2(I,\cW_0)$.

Clearly, we can write
\begin{align}
  \label{uppp}
  E^\bullet =
  \begin{bmatrix}
    E^\bullet_{11} & E^\bullet_{12} \\
    E^\bullet_{21} & E^\bullet_{22}
  \end{bmatrix}.
\end{align}
We define the propagators for the abstract Klein--Gordon operator by selecting the upper right element of the matrix of \eqref{uppp} and possibly by multiplying it by a conventional factor:
\bes\label{pkc}
\begin{align}
  \label{pkc1.}
  G^\bullet & \defn \im E_{12}^\bullet,\quad \bullet=\PJ,\vee,\wedge,\Feyn,\aFeyn; \\
  G_\pm^{(+)} & \defn E_{\pm,12}^{(+)},\quad G_\pm^{(-)} \defn - E_{\pm,12}^{(-)}.
\end{align}
\ees

\begin{theorem}
  (\ref{pkc}) are bounded operators on $L^2(I,\cK)$.
  They satisfy
  \bes
  \begin{align*}
    G^{\PJ*} & =-G^\PJ; \\ G^{\vee*} & =G^\wedge;\\ G^{\Feyn*} & =G^\aFeyn;\\ G_\pm^{(+)*} & = G_\pm^{(+)}\geq0,\\ G_\pm^{(-)*} & = G_\pm^{(-)}\geq0.
  \end{align*}
  \ees
\end{theorem}

We expect that typically $G^\bullet$ with $\bullet=\vee,\wedge,\Feyn,\aFeyn$ have a zero nullspace and a dense range.
(This will be proven below under some additional assumptions.)
If this is the case, we can define
\begin{equation}
  \label{pkc1}
  K^\bullet \defn G^{\bullet-1},\quad \bullet=\vee,\wedge,\Feyn,\aFeyn.
\end{equation}
Note that $K^\bullet $ can be viewed as well-posed realizations of the Klein--Gordon operator on a slab with appropriate (non-self-adjoint) boundary conditions.
Clearly,
\bes
\begin{align*}
  K^{\vee*} & =K^\wedge; \\ K^{\Feyn*} & =K^\aFeyn.
\end{align*}
\ees

\subsection{Propagators on the real line}
\label{Propagators on the real line}

Let $I=\RR$.
We assume that for $\pm t>T$ the operators $L(t)$, $\alpha(t)$ and $W(t)$ do not depend on $t$.
Therefore, the generators $B(t)$ for $\pm t>T$ do not depend on $t$, so that they can be denoted $B_\pm$.
We also assume that $B_\pm$ are stable and set $S_\pm \defn \sgn(B_\pm)$.
We define the propagators $E^\bullet$ undestood as operators from $L_{\mathrm{c}}^2(\mathbb{R},\cW_0)$ to $L_\loc^2(\mathbb{R},\cW_0)$.
Then we define the operators $G^\bullet$ just as in (\ref{pkc}), interpreted as operators $L_{\mathrm{c}}^2(\mathbb{R},\cK)$ to $L_\loc^2(\mathbb{R},\cK)$.
Obvious analogs of Prop.~\ref{pro1.} and~\ref{analog} hold for $I=\mathbb{R}$.

The above propagators play a central role in Quantum Field Theory on the spacetime $M$.
In fact, they correspond to the positive energy Fock representations of incoming and outgoing quantum fields, as will be sketched in Sect.~\ref{Quantum Field Theory on curved space-times}.

We actually believe that this choice is also distinguished for very different reasons by purely mathematical arguments.
It probably corresponds to a distinguished (maybe unique) self-adjoint realization of the Klein--Gordon operator.
We formulate our expectation in the following conjecture.
\begin{conjecture}
  For a large class of asymptotically stationary and stable abstract Klein--Gordon operators the following holds:
  \begin{enumerate}
    \item There exists a distinguished self-adjoint operator $K^{\mathrm{s.a.}}$ on $L^2(\mathbb{R},\cK)$ such that on $\cD(K^{\mathrm{s.a.}})$ the quadratic form (\ref{kleingordon1}) coincides with $\cinner{f_1}{K^{\mathrm{s.a.}}f_2}$.
    \item For $s>\frac12$ the following statements hold in the sense of $\langle t\rangle^{-s} L^2(\RR,\cK)\to\langle t\rangle^{s}L^2(\RR,\cK)$:
          \begin{align*}
            \slim_{\epsilon\searrow0} (K^{\mathrm{s.a.}}-\ri\epsilon)^{-1} & =G^{\Feyn},  \\
            \slim_{\epsilon\searrow0} (K^{\mathrm{s.a.}}+\ri\epsilon)^{-1} & =G^{\aFeyn}.
          \end{align*}
  \end{enumerate}
  \label{coject}
\end{conjecture}

We will discuss arguments in favor of this conjecture in Subsect.~\ref{Resolvent}.

\subsection{Perturbation of the evolution by the spectral parameter}

Now we are going to compute the resolvent of well-posed realizations of the Klein--Gordon operator.
To this end in this subsection we introduce the perturbed evolution $R_z(t,s)$.

Let $z\in\mathbb{C}$.
We define
\begin{align*}
  Z & \defn
            \begin{bmatrix}
              0 & 0 \\\one & 0
            \end{bmatrix}
  ,                         \\ B_z(t) & \defn
  \begin{bmatrix}
    W(t) & \alpha^2(t) \\L(t)-z & W^*(t)
  \end{bmatrix}
     =B(t)-zZ .
\end{align*}
Note that $Z$ is a bounded operator on each $\cW_\lambda$: \[ \|Zv\|_\lambda=\|L(t)^{\frac12(\lambda-\frac12)}v_1\|\leq \| L(t)^{\frac12(\lambda+\frac12)}v_1\|+\|L(t)^{\frac12(\lambda-\frac12)}v_2\|=\|v\|_\lambda.
\]

\begin{proposition}
  Suppose Assumption \ref{ass.I} holds.
  Then $t\mapsto B_z(t)$ generates an evolution $R_z(t,s)$ on $\cW_\lambda$, $\lambda\in[-\frac12,\frac12]$.
  We have
  \begin{equation}
    \label{pseu-}
    R_{\bar z}(t,s)^*QR_z(t,s)=Q, \quad R_z(t,s)^*=QR_{\bar z}(s,t)Q .
  \end{equation}
  In particular, if $z\in\mathbb{R}$, then $R_z(t,s)$ is pseudo-unitary on $\cW_0$.
\end{proposition}

Formally the equation
\begin{equation*}
  (z+K)f=0
\end{equation*}
is equivalent to
\begin{align*}
  \ri\partial_t
               \begin{bmatrix}
                 u_1 \\u_2
               \end{bmatrix}
   & =B_z(t)
            \begin{bmatrix}
              u_1 \\u_2
            \end{bmatrix}
  ,                        \\ u_1=f, & \quad u_2=\frac{1}{\alpha^2(t)}\bigl(D_t+W(t)\bigr)f.
\end{align*}
Therefore, the evolution $R_z(t,s)$ can be used to construct the resolvent of realizations of $K$.

\subsection{Resolvent for finite intervals}
\label{Resolvent of realizations of the Klein--Gordon}

Assume again that $I$ is finite.
We are back in the setting of Subsect.~\ref{Propagators for a finite interval}.
Recall that we impose Assumption \ref{ass.I} and choose admissible involutions $S_+$, $S_-$.
Let $(\cZ_+^{(+)},\cZ_+^{(-)})$ and $(\cZ_-^{(+)},\cZ_-^{(-)})$ be the corresponding particle/antiparticle spaces.
For $z\in\CC$ and $t\in I$, set
\begin{equation}
  \label{resol}
  \cZ_{\pm,z}^{(+)}(t) \defn R_z(t,t_\pm)\cZ_+^{(+)},\quad  \cZ_{\pm,z}^{(-)}(t) \defn R_z(t,t_\pm)\cZ_{\pm}^{(-)}.
\end{equation}
Let
\begin{align*}
  \RS \defn \bigl\{ z\in\CC \;\big| & \text{ for some (hence all) $t\in I$ the pair of subspaces} \\  & \text{$\bigl(\cZ_{+,z}^{(+)}(t),\cZ_{-,z}^{(-)}(t)\bigr)$ is complementary} \bigr\}.
\end{align*}
By Prop.~\ref{basic0} (2),
\begin{align*}
  \RS \defn \bigl\{ z\in\CC \;\big| & \text{ for some (hence all) $t\in I$ the pair of subspaces} \\  & \text{$\bigl(\cZ_{+,z}^{(-)}(t),\cZ_{-,z}^{(+)}(t)\bigr)$ is complementary} \bigr\}.
\end{align*}
For $z\in \RS$ we define
\begin{align*}
  \Lambda_z^{\Feyn(+)}(t), & \quad\text{ the projection onto $\cZ_{+,z}^{(+)}(t)$ along $\cZ_{-,z}^{(-)}(t)$}, \\ \Lambda_z^{\Feyn(-)}(t), & \quad\text{ the projection onto $\cZ_{-,z}^{(-)}(t)$ along $\cZ_{+,z}^{(+)}(t)$};\\ \Lambda_z^{\aFeyn(-)}(t), & \quad\text{ the projection onto $\cZ_{+,z}^{(-)}(t)$ along $\cZ_{-,z}^{(+)}(t)$},\\ \Lambda_z^{\aFeyn(+)}(t), & \quad\text{ the projection onto $\cZ_{-,z}^{(+)}(t)$ along $\cZ_{+,z}^{(-)}(t)$}.
\end{align*}
Set
\begin{align*}
  E_z^\Feyn(t,s) & \defn \theta(t-s) R_z(t,s) \Lambda_z^{\Feyn(+)}(s) - \theta(s-t) R_z(t,s) \Lambda_z^{\Feyn(-)}(s), \\ E_z^\aFeyn(t,s) & \defn \theta(t-s) R_z(t,s) \Lambda_z^{\aFeyn(-)}(s) - \theta(s-t) R_z(t,s) \Lambda_z^{\aFeyn(-)}(s);\\ G_z^\Feyn(t,s) & \defn \ri E_{z,12}^\Feyn(t,s),\\ G_z^\aFeyn(t,s) & \defn \ri E_{z,12}^\aFeyn(t,s).
\end{align*}

\begin{proposition}
  For $ z\in \RS$ the operators $G_z^\Feyn$, $G_{\bar{z}}^\aFeyn$ are bounded and satisfy the resolvent equation:
  \begin{align}
    G_z^\Feyn-G_w^\Feyn   & =(z-w) G_z^\Feyn G_w^\Feyn, \label{popo1}
    \\
    G_z^\aFeyn-G_w^\aFeyn & =(z-w)
    G_z^\aFeyn G_w^\aFeyn.
  \end{align}
  Besides,
  \begin{align}
    ( Q E_z^\Feyn)^*  & = -QE_{\bar z}^\aFeyn,\label{assua}
    \\
    (    G_z^\Feyn)^* & =   G_{\bar z}^\aFeyn.\label{assu}
  \end{align}
  \label{popopo}
\end{proposition}
\begin{proof}
  The boundedness is obvious.
  Let us prove (\ref{popo1}).
  A straightforward computation yields
  \bes\begin{align*}\MoveEqLeft
    (z-w) \bigl(E_z^\Feyn ZE_w^\Feyn \bigr)(t,s) = \int_{t_-}^{t_+} E_z^\Feyn(t, \tau) Z E_w^\Feyn(\tau, s) \dif\tau \\
    & = \phantom{{}+{}} \theta(t-s)(z-w)\int_s^t\Lambda_z^{\Feyn(+)}(t)R_z(t,\tau)ZR_w(\tau,s)\Lambda_w^{\Feyn(+)}(s)\dif\tau \\
    &\quad + \theta(s-t)(z-w)\int_t^s\Lambda_z^{\Feyn(-)}(t)R_z(t,\tau)ZR_w(\tau,s)\Lambda_w^{\Feyn(-)}(s)\dif\tau \\
    &\quad - \theta(t-s)(z-w)\int_{t_-}^s\Lambda_z^{\Feyn(+)}(t)R_z(t,\tau)ZR_w(\tau,s)\Lambda_w^{\Feyn(-)}(s)\dif\tau \\
    &\quad - \theta(s-t)(z-w)\int_{t_-}^t\Lambda_z^{\Feyn(+)}(t)R_z(t,\tau)ZR_w(\tau,s)\Lambda_w^{\Feyn(-)}(s)\dif\tau \\
    &\quad - \theta(s-t)(z-w)\int_s^{t_+}\Lambda_z^{\Feyn(-)}(t)R_z(t,\tau)ZR_w(\tau,s)\Lambda_w^{\Feyn(+)}(s)\dif\tau \\
    &\quad - \theta(t-s)(z-w)\int_t^{t_+}\Lambda_z^{\Feyn(-)}(t)R_z(t,\tau)ZR_w(\tau,s)\Lambda_w^{\Feyn(+)}(s)\dif\tau.
  \intertext{By the fundamental theorem of calculus this equals}
    & = \phantom{{}+{}} \theta(t-s) \Lambda_z^{\Feyn(+)}(t)\bigl(R_z(t,s)-R_w(t,s)\bigr)\Lambda_w^{\Feyn(+)}(s) \\
    &\quad + \theta(s-t)\Lambda_z^{\Feyn(-)}(t)\bigl(R_w(t,s)-R_z(t,s)\bigr)\Lambda_w^{\Feyn(-)}(s) \\
    &\quad - \theta(t-s)\Lambda_z^{\Feyn(+)}(t)\bigl(R_z(t,t_-)R_w(t_-,s)-R_z(t,s)\bigr)\Lambda_w^{\Feyn(-)}(s)\\
    &\quad - \theta(s-t)\Lambda_z^{\Feyn(+)}(t)\bigl(R_z(t,t_-)R_w(t_-,s)-R_w(t,s)\bigr)\Lambda_w^{\Feyn(-)}(s)\\
    &\quad - \theta(s-t)\Lambda_z^{\Feyn(-)}(t)\bigl(-R_z(t,t_+)R_w(t_+,s)+R_z(t,s)\bigr)\Lambda_w^{\Feyn(+)}(s)\\
    &\quad - \theta(t-s)\Lambda_z^{\Feyn(-)}(t)\bigl(-R_z(t,t_+)R_w(t_+,s)+R_w(t,s)\bigr)\Lambda_w^{\Feyn(+)}(s).
  \intertext{We rearrange this, obtaining}
    & = \phantom{{}+{}} \theta(t-s) \Lambda_z^{\Feyn(+)}(t)R_z(t,s)\bigl(\Lambda_w^{\Feyn(+)}(s)+\Lambda_w^{\Feyn(-)}(s)\bigr) \\
    &\quad - \theta(t-s)\bigl(\Lambda_z^{\Feyn(-)}(t)+ \Lambda_z^{\Feyn(+)}(t)\bigr) R_w(t,s) \Lambda_w^{\Feyn(+)}(s) \\
    &\quad - \theta(s-t) \Lambda_z^{\Feyn(-)}(t)R_z(t,s)\bigl(\Lambda^{\Feyn(+)}_w(s)+\Lambda_w^{\Feyn(-)}(s)\bigr) \\
    &\quad + \theta(s-t)\bigl(\Lambda_z^{\Feyn(-)}(t)+ \Lambda_z^{\Feyn(+)}(t)\bigr) R_w(t,s) \Lambda_w^{\Feyn(-)}(s) \\
    &\quad - R_z(t,t_-)\Lambda_z^{\Feyn(+)}(t_-)\Lambda_w^{\Feyn(-)}(t_-)R_w(t_-,s) \numberthis\label{last1} \\
    &\quad + R_z(t,t_+)\Lambda_z^{\Feyn(-)}(t_+)\Lambda_w^{\Feyn(+)}(t_+)R_w(t_+,s) \numberthis\label{last2}
  \intertext{which simplifies to}
    & = \phantom{{}-{}} \theta(t-s) \Lambda_z^{\Feyn(+)}(t)R_z(t,s) -\theta(t-s) R_w(t,s) \Lambda_w^{\Feyn(+)}(s) \\
    &\quad - \theta(s-t) \Lambda_z^{\Feyn(-)}(t)R_z(t,s) +\theta(s-t) R_w(t,s) \Lambda_w^{\Feyn(-)}(s) \\
    & = E_z^\Feyn(t,s)-E_w^\Feyn(t,s).
  \end{align*}\ees
  since, for any $z,w$,
  \begin{align*}
    \Ran\bigl(\Lambda_w^{\Feyn(-)}(t_-)\bigr) & = \cZ_-^{(-)} = \Ker\bigl(\Lambda_z^{\Feyn(+)}(t_-)\bigr) \\
    \Ran\bigl(\Lambda_w^{\Feyn(+)}(t_+)\bigr) & = \cZ_+^{(+)} = \Ker\bigl(\Lambda_z^{\Feyn(-)}(t_+)\bigr).
  \end{align*}
  Thus we have proven that
  \begin{equation}
    (z-w)\bigl(E_z^\Feyn ZE_w^\Feyn\bigr)(t,s)=E_z^\Feyn(t,s)-E_w^\Feyn(t,s).
    \label{popo}
  \end{equation}
  Taking the $1,2$ component of (\ref{popo}) we obtain (\ref{popo1}).

  To prove (\ref{assua}) we use (\ref{pseu-}).
  This implies (\ref{assu}).
\end{proof}

Clearly, if we can define $K^\Feyn$, $K^\aFeyn$ by (\ref{pkc1}) as operators with dense domains, then we have
\[
  G_z^\Feyn=(z+K^\Feyn)^{-1},\quad G_z^\aFeyn=(z+K^\aFeyn)^{-1}
\]
and
\[
  \text{resolvent set of $ K^\Feyn$} \,=\, \bar{ \text{resolvent set of $ K^{\aFeyn}$}} \,\subset\, \RS.
\]

\subsection{Resolvent for $I=\mathbb{R}$}
\label{Resolvent}

In this subsection we will give some arguments supporting Conj.~\ref{coject}.
We will sketch a construction of a family of operators which we expect to be the resolvent of the (putative) distinguished self-adjoint realization of the abstract Klein--Gordon operator for $I=\mathbb{R}$.
We impose the assumptions of Subsect.~\ref{Propagators on the real line} (but it is likely that more assumptions are needed).
Our analysis will not be complete.

First note that we cannot repeat the constructions of Subsect.~\ref{Resolvent of realizations of the Klein--Gordon} without major changes.
In fact, for $I=\mathbb{R}$ in the definitions (\ref{resol}) one should take $t\to\pm\infty$.
However, for $z\not\in\mathbb{R}$ the evolution $R_z(t,s)$ blows up on a part of $\cW_0$ and decays on another part as $t\to\pm\infty$.

We need to define the projections onto the (distorted) positive/negative part of the spectrum of $B_{\pm,z}$ (defined analogously to $B_\pm$ in Subsect.~\ref{Propagators on the real line}).
This is straightforward for $B_\pm$, because they are self-adjoint.
However, $B_{\pm,z}$ in general are not self-adjoint.

For simplicity, we will assume that $B_\pm $ are strongly stable, so that their spectrum has a gap around $0$.
Then for small $\Re(z)$ the operators $B_{\pm,z}$ are ``bisectorial'', which is sufficient for the construction of these projections.

This is described in Proposition 7.2 of \cite{derezinski-siemssen:static}, which implies the following proposition:

\begin{proposition}
  There exists $\zeta_0>0$ such that the strip $\{\zeta\in\mathbb{C} \mid -\zeta_0\leq\Re(\zeta)\leq\zeta_0\}$ is contained in the resolvent set of $B_{\pm,z}$.
  Moreover, the operators
  \begin{align*}
    \Pi_{\pm, z}^{(+)} & \defn \lim_{\tau\to\infty}\frac12\Big(\one+\frac1{\pi\ri} \int_{-\ri\tau}^{\ri\tau}(B_{\pm,z}-\zeta)^{-1}\dif\zeta\Big), \\ \Pi_{\pm, z}^{(-)} & \defn \lim_{\tau\to\infty}\frac12\Big(\one-\frac1{\pi\ri} \int_{-\ri\tau}^{\ri\tau}(B_{\pm,z}-\zeta)^{-1}\dif\zeta\Big)
  \end{align*}
  constitute a pair of complementary projections commuting with $B_{\pm,z}$ such that
  \begin{align*}
    \sigma\bigl(B_{\pm,z}\Pi_{\pm,z}^{(+)}\bigr) & =\sigma\bigl(B_{\pm,z}\bigr)\cap\{w\in\mathbb{C} \mid \Re(w)\geq0\}, \\ \sigma\bigl(B_{\pm,z}\Pi_{\pm,z}^{(-)}\bigr) & =\sigma\bigl(B_{\pm,z}\bigr)\cap\{w\in\mathbb{C} \mid \Re(w)\leq0\}.
  \end{align*}
\end{proposition}
Set
\begin{align*}
  \cZ_{\pm,z}^{(\pm)}(t) & \defn \Ran\Big(\lim_{\tau\to\pm\infty}R_z(t,\tau) \Pi_{\pm,z}^{(\pm)}R_z(\tau,t)\Big),\quad \Im(z)\geq0 \\ \cZ_{\pm,z}^{(\mp)}(t) & \defn \Ran\Big(\lim_{\tau\to\pm\infty}R_z(t,\tau) \Pi_{\pm,z}^{(\mp)}R_z(\tau,t)\Big),\quad \Im(z)\leq0.
\end{align*}

We can now complement Conj.~\ref{coject} with an additional conjecture about the resolvent of $K^{\mathrm{s.a.}}$:

\begin{conjecture}
  \label{conject}
  We expect that for $\Im(z)\geq0$ for some (hence all) $t\in \mathbb{R}$ the pair of subspaces $\bigl(\cZ_{+,z}^{(+)}(t),\cZ_{-,z}^{(-)}(t)\bigr)$ is complementary.
  Let $\bigl(\Lambda_z^{\Feyn(+)}(t),\Lambda_z^{\Feyn(-)}(t)\bigr)$ be the pair of projections corresponding to this pair of spaces.

  Equivalently, we expect that for $\Im(z)\leq0$ for some (hence all) $t\in \mathbb{R}$ the pair of subspaces $\bigl(\cZ_{+,z}^{(-)}(t),\cZ_{-,z}^{(+)}(t)\bigr)$ is complementary.
  Let $\bigl(\Lambda_z^{\aFeyn(-)}(t),\Lambda_z^{\aFeyn(+)}(t)\bigr)$ be the pair of projections corresponding to this pair of spaces.

  We also introduce
  \begin{alignat*}{2}
    E_z(t,s) & \defn \theta(t-s) R_z(t,s) \Lambda_z^{\Feyn(+)}(s) - \theta(s-t) R_z(t,s) \Lambda_z^{\Feyn(-)}(s),\quad && \Im(z)>0; \\
    E_z(t,s) & \defn \theta(t-s) R_z(t,s) \Lambda_z^{\aFeyn(-)}(s) - \theta(s-t) R_z(t,s) \Lambda_z^{\aFeyn(+)}(s),\quad && \Im(z)<0; \\
    G_z(t,s) & \defn \ri E_{z,12}(t,s),\quad && \Im(z)\neq0.
  \end{alignat*}

  Then we conjecture that $G_z$ defines for $z\in\mathbb{C}\backslash\mathbb{R}$ a bounded operator on $L^2(\mathbb{R},\cK)$ with a dense range and a trivial nullspace such that
  \begin{align}
    G_z-G_w & =(z-w) G_zG_w, \label{popo1.}
    \\           G_z^* & =G_{\bar z}.\label{popo2.}
  \end{align}
  Thus $G_z$ is the resolvent of a self-adjoint operator, which we can call $K^\mathrm{s.a}$ and treat as the distinguished self-adjoint realization of $K$.
\end{conjecture}

Let us sketch some arguments in favor of the above conjecture.

First of all, in the stationary case, that is, if $B(t)$ does not depend on $t$, the conjecture is true (with minor additional assumptions) following the arguments of \cite{derezinski-siemssen:static} and \cite{derezinski-siemssen:propagators}.

The conjecture is also true if $\cK=\mathbb{C}$.
In fact, the operator $K$ is then essentially the well-known 1-dimensional magnetic Schr\"odinger operator.
It should not be difficult to generalize this to the case of a finite dimensional $\cK$, or bounded $L(t)$ and $W(t)$.

Unfortunately, for a generic spacetime, $L(t)$ is unbounded and non-stationary.
We know in this case that asymptotic complementarity holds for real $z$, so we can expect it to hold in a neighborhood of $\mathbb{R}$.
$E_z$  are well defined
as quadratic forms on, say, $C_\mathrm{c}(\mathbb{R},\cW_0)$
and (\ref{popo2.}) is easily checked.
However, we do not know how to control the norm of $E_z(t,s)$ for large $t,s$, and hence to show the boundedness of $E_z$.

We expect that
\begin{align}
  \lim_{t\to+\infty}\e^{\mp\ri tB_{\pm,z}}\Pi_{\pm,z}^{(\pm)} & =0,\quad \Im(z)\geq 0;\label{pas1}
  \\
  \lim_{t\to-\infty}\e^{\mp\ri tB_{\pm,z}}\Pi_{\pm,z}^{(\mp)} & =0,\quad \Im(z)\leq 0.\label{pas2}
\end{align}
(This follows under a slightly stronger assumption from Proposition 7.2 of \cite{derezinski-siemssen:static}.)
Now the resolvent equation (\ref{popo1.}) should follow by the same calculation as in the proof of Prop.~\ref{popopo}, except that the terms (\ref{last1}) and (\ref{last2}) should be zero by (\ref{pas1}) and (\ref{pas2}).

\subsection{Additional regularity}

In order to have better properties of the Klein--Gordon form, and in particular, to guarantee that it defines an operator, it will be useful to introduce a family of assumptions more restrictive than Assumption \ref{ass.I}.
This family will depend on a parameter $\rho\geq0$.

\begin{assumption}[$\rho$]\mbox{}
  \begin{enumerate}
    \item
          For any $t\in I$ there exist $0<c_1\leq c_2$ such that
          \begin{equation*}
            c_1L_0(0)^{1+\rho}\leq L_0(t)^{1+\rho}\leq c_2 L_0(0)^{1+\rho}
          \end{equation*}
          and $I^\cl\ni t\mapsto L_0(0)^{-\frac{1+\rho}{2}} L_0(t)^{1+\rho}L_0(0)^{-\frac{1+\rho}{2}}\in B(\cK)$ is absolutely norm continuous.
    \item $I^\cl\ni t\mapsto L_0(t)^{\frac\rho2}\alpha^2(t) L_0(t)^{-\frac\rho2} \in B(\cK)$ is absolutely norm continuous.
    \item
          $I^\cl\ni t\mapsto L_0(t)^{\frac\rho2} W(t)L_0(t)^{\frac{-1-\rho}{2}}$, $L_0(t)^{-\frac\rho2} W(t)L_0(t)^{\frac{-1+\rho}{2}}\in B(\cK)$ are absolutely norm continuous and there exists $a<1$ such that
          \begin{align*}
            \|L(t)^{-\frac\rho2}\alpha(t)^{-1}W(t)L_0(t)^{\frac{-1+\rho}{2}}\|\leq a, \\ \|L(t)^{\frac\rho2}\alpha(t)^{-1}W(t)L_0(t)^{\frac{-1-\rho}{2}}\|\leq a.
          \end{align*}
  \end{enumerate}
  \label{assII}
\end{assumption}

Note that Assumption \ref{ass.I} coincides with Assumption \ref{assII}($0$).
If $0\leq \rho'\leq\rho$, then Assumption \ref{assII}($\rho$) implies Assumption \ref{assII}($\rho'$).

\begin{proposition}
  Suppose Assumption \ref{assII}($\rho$) holds.
  Then the following is true:
  \begin{enumerate}
    \item
          For any $t\in I^\cl$ and for $|\beta|\leq 1+\rho$, the Hilbert spaces $\cK_{\beta,t}$ and $\cK_{\beta,0}$ are compatible.
          Hence we can extend the scale of Hilbertizable spaces $\cK_\beta$ to $|\beta|\leq 1+\rho$ \item We can extend the scale of Hilbertizable spaces $\cW_\lambda$ to $|\lambda|\leq\frac12+\rho$.
          Besides, for any $t\in I^\cl$ and such $\lambda$ the Hilbert spaces $\cW_{\lambda,t}$ are compatible with $\cW_\lambda$.
    \item
          For any $-\rho-\frac12\leq\lambda\leq\rho-\frac12$
          the function $I\ni t\mapsto B(t):\cW_{\lambda+1}\to\cW_{\lambda}$ satisfies the assumptions of Thm.~\ref{thm:evolution-per}, and hence defines an evolution
          $R(t,s)$ on $\cW_\lambda$, $|\lambda|\leq\frac12+\rho$.
          For $-\rho-\frac12\leq\lambda\leq\rho-\frac12$ this evolution satisfies
          \begin{equation}
            \label{podi}
            \bigl(\partial_t+\ri B(t)\bigr)R(t,s)u=0,\quad u\in\cW_{\lambda+1}.
          \end{equation}
  \end{enumerate}
\end{proposition}
\begin{proof}
  Assumption \ref{assII}($\rho$) implies immediately that the Hilbertizable spaces $\cK_{\pm(1+\rho),t}$ and $\cK_{\pm(1+\rho),0}$ coincide.
  Therefore (1) follows by the Kato--Heinz inequality.

  Let us drop $(t)$.
  We have
  \begin{equation}
    \begin{bmatrix}
      L_0^{\frac\theta2} & 0 \\0 & L_0^\frac{-1+\theta}{2}
    \end{bmatrix}
    \begin{bmatrix}
      W & \alpha^2 \\L_0 & W^*
    \end{bmatrix}
    \begin{bmatrix}
      L_0^{\frac{-1-\theta}{2}} & 0 \\0 & L_0^{-\frac{\theta}{2}}
    \end{bmatrix}
    =
    \begin{bmatrix}
      L_0^{\frac\theta2} W L_0^{\frac{-1-\theta}2} & L_0^{\frac\theta2}\alpha^2 L_0^{-\frac\theta2} \\\one & L_0^{\frac{-1+\theta}2}W^* L_0^{-\frac\theta2}
    \end{bmatrix}
    .
    \label{qaw1}
  \end{equation}
  (\ref{qaw1}) is bounded for $\theta=\rho$ and $\theta=-\rho$.
  Hence by interpolation it is bounded for $-\rho\leq \theta\leq\rho$.
  Hence $B_0$ is bounded from $\cW_{\theta+\frac12}$ to $\cW_{\theta-\frac12}$.

  Let $-\rho\leq \theta\leq\rho$.
  (\ref{qaw1}) can be represented as
  \begin{align*}
    \begin{bmatrix}
      L_0^{\frac\theta2} \alpha L_0^{-\frac{\theta}2} & 0 \\0 & \one
    \end{bmatrix}
    \begin{bmatrix}
      L_0^{\frac\theta2} \alpha^{-1} W L_0^{\frac{-1-\theta}2} & \one \\\one & L_0^{\frac{-1+\theta}2}W^*\alpha^{-1} L_0^{-\frac\theta2}
    \end{bmatrix}
    \begin{bmatrix}
      \one & 0 \\0 & L_0^{\frac\theta2} \alpha L_0^{-\frac{\theta}2}
    \end{bmatrix}
  \end{align*}
  The two extreme terms are invertible.
  Besides, \[\|L_0^{\frac\theta2} \alpha^{-1} W L_0^{\frac{-1-\theta}2}\|\leq a<1.
  \]
  Hence the middle term is also invertible.
  This proves the invertibility of $B_0(t)$ as a map from $\cW_{\theta+\frac12}$ to $\cW_{\theta-\frac12}$.

  Let $-\rho-\frac12\leq\theta\leq \rho+\frac12$.
  Then for some $\theta_0\in[-\frac12,\frac12]$ and $n\in\mathbb{Z}$ we have $\theta=n+\theta_0$.
  Then $\cW_{\theta,t}=B_0(t)^n\cW_{\theta_0,t}$.
  But as Hilbertizable spaces $\cW_{\theta_0,t}=\cW_{\theta_0}$.
  We have just proved that $B_0(t)^n$ are bounded invertible in the sense $\cW_{\theta_0}\to\cW_{\theta_0+n}$ for $\theta_0+n\leq \rho+\frac12$.
  Hence $\cW_{\theta_0+n,t}=\cW_{\theta_0+n}$.
  Hence (2) is true.

  (3) is proven in a similar way as Prop.~\ref{exid}.
\end{proof}

Note that Assumption \ref{assII}($\rho$) is especially important for $\rho=\frac12$.
Then (\ref{podi}) holds with $\lambda=0$, so that $B(t)$ can be interpreted as an (unbounded) pseudo-unitary generator on the Krein space $\cW_0$ with the domain $\cW_1$.

\subsection{The abstract Klein--Gordon operator}
\label{The abstract Klein-Gordon operator}

Analysis of differential operators with variable coefficients of low regularity is a rather technical and complicated subject, even if these coefficients are scalar.
In our case these coefficients have values in unbounded operarators, hence it is not surprising that defining Klein--Gordon operators with low regularity conditions is difficult and messy.

For completeness, in the following theorems we give conditions which allow us to define abstract Klein--Gordon operators.
Note that we do not attempt to be optimal.

\begin{proposition}
  Suppose Assumption \ref{assII}($\frac12$) holds.
  Let $\bullet=\vee,\wedge,\Feyn,\aFeyn$.
  \begin{enumerate}
    \item
          Let
          \begin{equation}
            \label{sfs}
            f\in C_\mathrm{c}(I,\cK_{\frac32})\cap C_\mathrm{c}^1(I,\cK_{\frac12})
            ,\quad\alpha^{-2}(-\ri\partial_t+W)f\in
            C_\mathrm{c}^1(I,\cK_{-\frac12}) .
          \end{equation}
          Then
          \begin{equation}
            G^\bullet Kf=f.
            \label{vece2.}
          \end{equation}
    \item If the space of $f$ satisfying (\ref{sfs}) is dense in $\cK$, then
          $\Ker(G^\bullet)=\{0\}$ and $\Ran(G^\bullet)$ is dense.
    \item
          If in addition $I\ni t\mapsto L(t)^{-\frac14}\bigl(\partial_t\alpha(t)^{-2}\bigr) L(t)^{-\frac14}$, $L(t)^{-\frac14}\bigl(\partial_t W(t)\bigr) L(t)^{-\frac34} \in B(\cK)$ are continuous, then (\ref{sfs}) is equivalent to
          \begin{equation*}
            f\in      C_\mathrm{c}(I,\cK_{\frac32})\cap C_\mathrm{c}^1(I,\cK_{\frac12})
            \cap C_\mathrm{c}^2(I,\cK_{-\frac12}).
          \end{equation*}
  \end{enumerate}
  \label{pro1.}
\end{proposition}
\begin{proof}
  (1):
  By Prop.~\ref{pos8} (1) we know that if
  $u\in C_\mathrm{c}(I,\cW_1)\cap C_\mathrm{c}^1(I,\cW_0)$, then
  \begin{equation}
    E^\bullet\bigl(\partial_t+\ri B(t)\bigr)u=u.
    \label{vece1.}
  \end{equation}
  $      f\in C_\mathrm{c}(I,\cK_{\frac32})\cap C_\mathrm{c}^1(I,\cK_{\frac12})$ implies that $\alpha^{-2}(-\ri\partial_t+W)f\in
    C_\mathrm{c}(I,\cK_{\frac12})$.
  Therefore,
  \begin{equation}
    \label{vece}
    u \defn
    \begin{bmatrix}
      f \\
      -\alpha^{-2}(-\ri\partial_t+W)f
    \end{bmatrix}
    \in
    C_\mathrm{c}(I,\cW_1)\cap C_\mathrm{c}^1(I,\cW_0).
  \end{equation}
  Hence we can apply (\ref{vece1.}) to $u$.
  We have
  \begin{align*}
    \notag Q( -\ri \partial_t+ B) & =
                                      \begin{bmatrix}
                                        L & -\ri\partial_t+ W^* \\-\ri\partial_t+W & \alpha^2
                                      \end{bmatrix}
    \\ = &
    \begin{bmatrix}
      \one & (-\ri\partial_t + W^*)\alpha^{-2} \\0 & \one
    \end{bmatrix}
      \begin{bmatrix}
      - K & 0 \\0 & \alpha^2
    \end{bmatrix}
      \begin{bmatrix}
      \one & 0 \\ \alpha^{-2}(-\ri \partial_t + W) & \one
    \end{bmatrix}
      .
  \end{align*}
  For $u$ given by (\ref{vece}) we have
  \[ Q(-\ri\partial_t+B) u=
    \begin{bmatrix}
      -Kf \\0
    \end{bmatrix}
    .
  \]
  We obtain
  \begin{align*}
    u=E^\bullet
    \bigl(\partial_t+\ri B(t)\bigr)u
    =
    \begin{bmatrix}
      -  E_{12}^\bullet Kf \\ - E_{22}^\bullet Kf
    \end{bmatrix}
    ,
  \end{align*}
  which yields (\ref{vece2.}).

  By (\ref{vece2.})  the range of $G^\bullet$ contains (\ref{sfs}) contains.
  The same argument shows that the range of $G^{\bullet*}$ contains (\ref{sfs}).
  If the space of (\ref{sfs}) is dense, we have
  \begin{equation*}
    \Ker(G^\bullet)=\Ran(G^{\bullet*})^\perp=\{0\}.
  \end{equation*}
  This proves (2).

  (3) follows by checking that $\partial_t\alpha^{-2}\partial_tf$ and $\partial_t\alpha^{-2} Wf$ are in $C_\mathrm{c}(I,\cK_{-\frac12})$.
\end{proof}

Thus, if $I$ is finite, under the assumptions of Prop.~\ref{pro1.} (3)  we can define  an unbounded  closed operator $K^\bullet$
in the sense of $L^2(I,\cK)$ with $0$ belonging to its resolvent set.
It corresonds to the quadratic form (\ref{kleingordon1}) with the appropriate boundary conditions.
For instance, the boundary conditions of $K^\Feyn$ are
\begin{align*}
  f\in\cD( K^\Feyn) & \;\Rightarrow\;
                                        \begin{bmatrix}
                                          f(t_\mp) \\ -\alpha^{-2}(-\ri\partial_t+W)f(t_\mp)
                                        \end{bmatrix}
  \in\cZ_{\mp}^{(\mp)},                                                                  \\[2ex] f\in\cD( K^{\aFeyn}) & \;\Rightarrow\;
  \begin{bmatrix}
    f(t_\mp) \\ -\alpha^{-2}(-\ri\partial_t+W)f(t_\mp)
  \end{bmatrix}
          \in\cZ_{\mp}^{(\pm)} .
\end{align*}

Clearly, the above construction is indirect.
It does not mean that the expression (\ref{kleingordon}) is well defined in the sense of the Hilbert space $L^2(I,\cK)$.
Under more stringent conditions, such as those described in the following proposition, one can directly interpret (\ref{kleingordon}) as an operator:

\begin{proposition}
  Suppose that Assumption \ref{assII}($1$) holds.
  In addition, assume that $I^\cl\ni t\mapsto \bigl(\partial_t\alpha^{-2}(t)\bigr)L_0(t)^{-\frac12}$, $\bigl( \partial_t W(t)\bigr)L_0(t)^{-1}\in B(\cK)$ are norm continuous families.
  Then $K$, as defined in (\ref{kleingordon}), maps
  \begin{equation}
    C_\mathrm{c}(I,\cK_2)\cap C_\mathrm{c}^1(I,\cK_1)\cap C_\mathrm{c}^2(I,\cK_0)\label{qeqe}
  \end{equation}
  into $C_\mathrm{c}(I,\cK_0)$.
  Hence (\ref{qeqe}) can serve as a dense domain of the operator $K$.
  \label{analog}
\end{proposition}
\begin{proof}
  We rewrite $K$ as
  \begin{align*}
    K & =\frac{1}{\alpha^2(t)}D_t^2+W^*(t)\frac{1}{\alpha^2(t)}D_t+ \frac{1}{\alpha^2(t)}W(t)D_t+ \Big(\partial_t\frac{\ri}{\alpha^2(t)}\Big)D_t \\  & +W^*(t)\frac{1}{\alpha^2(t)}W(t)+ \Big(\partial_t\frac{\ri}{\alpha^2(t)}\Big)W(t)-\frac{\ri}{\alpha^2(t)} \partial_t W(t)-L(t).
    \tag*{$\square$}
  \end{align*}
\end{proof}

Now in the $I=\RR$ case we can strengthen Conj.~\ref{conject}:
\begin{conjecture}
  \label{conject2}
  Impose the assumptions of Conj.~\ref{conject} and Prop.~\ref{analog}, Then the operator $K$ with the domain, say, \eqref{qeqe} is essentially self-adjoint and its closure $K^{\mathrm{s.a.}}$ satisfies Conj.~\ref{conject}.
\end{conjecture}

\section{Bosonic quantization}
\label{Bosonic quantization}

In this section we describe the basics of quantization used in bosonic QFT.
It involves two steps.
First, we select a classical phase space, and we associate with it an algebra of Canonical Commutation Relations.
Second, we choose a representation of this algebra.

We will describe four formalisms used in Step 1 as summarized in the table in Subsect.~\ref{Intro Bosonic quantization}.
We also describe how to define Fock representations, which is the most common realization of Step 2.

Note that, unlike in the introduction and the next section, in this section we do not put ``hats'' on quantized operators to reduced the notational burden.

\subsection{Real (or neutral) formalism}
\label{Bosonic quantization: real (or neutral) formalism}

\subsubsection{Canonical commutation relations}
\label{Canonical commutation relations}

Suppose that $\cY$ is a real vector space equipped with an antisymmetric form $\omega$, i.e., $(\cY,\omega)$ is a pre-symplectic space.

Let $\CCR(\cY)$ denote the complex unital $*$-algebra generated by $\Phi(w)$, $w\in\cY$, satisfying
\begin{enumerate}
  \item $\Phi(w)^* = \Phi(w)$,
  \item the map $\cY \ni w \mapsto \Phi(w)$ is linear,
  \item the canonical commutation relations hold,
        \begin{equation*}
          \bigl[ \Phi(v), \Phi(w) \bigr] = \im \rinner{v}{\omega w},
          \quad
          v,w \in \cY.
        \end{equation*}
\end{enumerate}

Let $\cW \defn \CC\cY$ be the complexification of $\cY$ and $Q$ the corresponding Hermitian form, as described in~\eqref{qiu}:
\begin{equation*}
  \cinner{v}{Q w} \defn \im \rinner{\bar{v}}{\omega w},\qquad v,w\in\cW.
\end{equation*}
We extend $\Phi$ to $\cW$, so that it is complex antilinear:
\begin{equation*}
  \Phi(w_R+\im w_I) \defn \Phi(w_R)-\im\Phi(w_I), \quad w_R,w_I\in\cY.
\end{equation*}
Then we have, for all $v,w \in \cW$,
\begin{align*}
  \Phi^*(w) \defn \Phi(w)^* & = \Phi(\bar w), \\ \bigl[ \Phi(v), \Phi^*(w) \bigr] & = \cinner{v}{Q w}.
\end{align*}

\subsubsection{Fock representation}

Assume in addition that $\cW$ is Krein.
Let $S_\bullet$ be an admissible anti-real involution on $\cW$, see Subsect.~\ref{sub:Admissible involutions and Krein spaces}.
Let $\Pi_\bullet^{(+)}$, $\cZ_\bullet^{(\pm)}=\Pi_\bullet^{(+)}\cW$ be the corresponding particle projection and space, see Subsect.~\ref{sub:Involutions}.

It is well-known that the two-point function
\begin{equation*}
  \omega_\bullet\bigl(\Phi(v)\Phi^*(w)\bigr) = \cinner{\Pi_\bullet^{(+)} v}{Q\Pi_\bullet^{(+)} w}, \quad v,w \in \cW,
\end{equation*}
uniquely determines a centered pure quasi-free state on the algebra
$\CCR(\cY)$.

 To describe the  corresponding GNS
representation first note that the charge $Q$ defines
on $\cZ_\bullet^{(+)}$ a positive definite scalar product
\begin{equation}\label{hilbo}
  \cinner{z_1}{z_2} \defn \cinner{z_1}{Qz_2},\quad z_1,z_2\in\cZ_\bullet^{(+)}.\end{equation}
Hence for $z \in \cZ_\bullet^{(+)}$ we can introduce
the standard
annihilation, resp.\ creation operators  $a_\bullet(z)$
and $a_\bullet^*(z)$ acting on
 the
bosonic Fock space $\Gamma_\s(\cZ_\bullet^{(+)})$, see e.g.\ \cite{derezinski-gerard}.
The state $\omega_\bullet$ is given by the vacuum $\Omega_\bullet \in \Gamma_\s(\cZ_\bullet^{(+)})$:
\begin{equation*}
  \omega_\bullet(\,\cdot\,) = \cinner{\Omega_\bullet}{\,\cdot\;\Omega_\bullet},
\end{equation*}
and the representation is given by
\begin{align*}
  \Phi_\bullet(w) & \defn a_\bullet(\Pi_\bullet^{(+)} w) + a_\bullet^*(\Pi_\bullet^{(+)} \bar w), \\ \Phi_\bullet^*(w) & \defn a_\bullet^*(\Pi_\bullet^{(+)} w) + a_\bullet(\Pi_\bullet^{(+)} \bar w),\quad w \in \cW.
\end{align*}
Note that if $z \in \cZ_\bullet^{(+)}$, then
\begin{alignat*}{2} \Phi_\bullet(z) & = \Phi_\bullet^*(\bar z) = a_\bullet(z), \\ \Phi_\bullet(\bar z) & = \Phi_\bullet^*(z) = a_\bullet^*(z).
\end{alignat*}


\subsubsection{Two-component representation}

Let $\cX$ be a real Hilbertizable space and $\cX^*$ is its
dual.  The pairing of $\cX$ and $\cX^*$
  will be denoted $\rinner{\cdot}{\cdot}$.
We equip $\cX^*\oplus\cX$ with the symplectic form
\begin{equation*}
  (u_2,v_2)\omega(u_1,v_1)= \rinner{u_2}{v_1} - \rinner{v_2}{u_1} ,\qquad(u_1,v_1),(u_2,v_2)\in\cX^*\oplus\cX.
\end{equation*}
Then $\cY \defn \cX^*\oplus\cX$ is a symplectic space and $\CC\cY$ is a Krein space.

Consider the $*$-algebra generated by $\phi(u)$, $u\in\cX^*$, $\pi(v)$, $v\in \cX$ satisfying
\begin{enumerate}
  \item $\phi(u)^* = \phi(u)$, $\pi(v)^* = \pi(v)$; \item the maps $\cX^* \ni u \mapsto \phi(u)$, $\cX \ni v \mapsto \pi(v)$ are linear; \item the CCR in the two-component form hold:
        \begin{alignat*}{3}
          \bigl[ \phi(u_1), \phi(u_2) \bigr] & = \bigl[ \pi(v_1), \pi(v_2) \bigr] = 0;
          \quad
          && u_1,u_2\in\cX^*, v_1,v_2\in\cX.
          \\
          \bigl[ \phi(u), \pi(v) \bigr]      & = \im \rinner{u}{v},
          \quad
          && u\in\cX^*, v\in\cX.
        \end{alignat*}
\end{enumerate}
We can pass to the formalism of Subsubsect.~\ref{Canonical commutation relations} by setting
\begin{align*}
  \Phi(v,u) & \defn\pi(v)+ \phi(u),\quad u\in\cX^*, v\in\cX.
\end{align*}
Indeed,
\begin{equation*}
  \big[ \Phi(v_2,u_2), \Phi(v_1,u_1)\big] =\im\bigl( \rinner{u_2}{v_1}- \rinner{v_2}{u_1}\bigr).
\end{equation*}

We can extend $\phi,\pi$ to $\mathbb{C}\cX^*\ni u\mapsto\phi(u)$ and $\mathbb{C}\cX\ni v\mapsto \pi(v)$ by antilinearity.
Then we can replace the symplectic form by the Hermitian form
\begin{align*}
  \cinner[\big]{(\im v_2,u_2)}{Q(\im v_1,u_1)} & = \rinner{\bar{u_2}}{\im v_1}+ \rinner{\bar{\im v_2}}{ u_1} \\ & =\im\bigl( \rinner{\bar{u_2}}{v_1}- \rinner{\bar{ v_2}}{u_1} \bigr).
\end{align*}

Every anti-real admisssible involution $S_\bullet$ on $\mathbb{C}\cX^*\oplus \mathbb{C}\cX$ determines a Fock state $\omega_\bullet$.
Let
\begin{align*}
  \Pi_\bullet^{(+)}=
  \begin{bmatrix}
    \Pi_{\bullet 11}^{(+)} & \Pi_{\bullet 12}^{(+)} \\ \Pi_{\bullet 21}^{(+)} & \Pi_{\bullet 22}^{(+)}
  \end{bmatrix}
  .
\end{align*}
be the corresponding ``particle projection''.
The conditions $\Pi_\bullet^{(+)*}Q=Q\Pi_\bullet^{(+)}$ and $\Pi_\bullet^{(+)}+\bar{\Pi_\bullet^{(+)}}=\one$ yield
\[
  \Pi_{\bullet 22}^{(+)*} = \Pi_{\bullet 11}^{(+)},
  \quad
  \Pi_{\bullet 12}^{(+)*} = \Pi_{\bullet 12}^{(+)}
  ,\quad
  \Pi_{\bullet 21}^{(+)*} = \Pi_{\bullet 21}^{(+)},
\]
\[
  \Pi_{\bullet 22}^{(+)} + \bar{\Pi_{\bullet 22}^{(+)}} = \Pi_{\bullet 11}^{(+)} + \bar{\Pi_{\bullet 11}^{(+)}} = \one,
  \quad
  \Pi_{\bullet 12}^{(+)} + \bar{\Pi_{\bullet 12}^{(+)}} = \Pi_{\bullet 21}^{(+)} + \bar{\Pi_{\bullet 21}^{(+)}} = 0.
\]
The two-point correlation functions of the state $\omega_\bullet$ are
given by
\begin{subequations}\label{recipe1}\begin{align}
    \omega_\bullet\bigl(\pi(v)\pi(v')\bigr) &= \rinner[\big]{v}{\Pi_{\bullet21}^{(+)}v'}, \\
    \omega_\bullet\bigl(\phi(u)\phi(u')\bigr) &= \rinner[\big]{u}{\Pi_{\bullet12}^{(+)}u'}, \\
    \omega_\bullet\bigl(\phi(u)\pi(v)\bigr) &= \ri\rinner[\big]{u}{\Pi_{\bullet11}^{(+)}v}, \\
    \omega_\bullet\bigl(\pi(v)\phi(u)\bigr) &= -\ri\rinner[\big]{v}{\Pi_{\bullet22}^{(+)}u} = -\ri\rinner{u}{\bar{\Pi_{\bullet 11}^{(+)}}v}.
\end{align}\end{subequations}


\subsection{Complex (or charged) formalism}
\label{Bosonic quantization: complex (or charged) formalism}

\subsubsection{Charged canonical commutation relations}
\label{Charged canonical commutation relations}

Suppose that $\cW$ is a complex vector space equipped with a Hermitian form
\begin{equation*}
  \cinner{v}{Qw}, \quad v,w\in\cW.
\end{equation*}

Let $\CCR(\cW)$ denote the complex unital $*$-algebra generated by $\Psi(w)$ and $\Psi^*(w)$, $w\in\cW$, such that
\begin{enumerate}
  \item $\Psi^*(w) = \Psi(w)^*$,
  \item the map $\cW \ni w \mapsto \Psi^*(w)$ is linear,
  \item the canonical commutation relations in the complex form hold:
    \begin{equation*}
      \bigl[\Psi(v),\Psi^*(w)\bigr] = \cinner{v}{Q w},
      \quad
      \bigl[\Psi^*(v),\Psi^*(w)\bigr] = 0,  \quad v,w \in \cW.
    \end{equation*}
\end{enumerate}

We can pass from the complex to real formalism as follows.
Set
\begin{equation*}
  \Phi_\mathrm{R}(w) \defn \frac{1}{\sqrt2}\bigl(\Psi(w)+\Psi^*(w)\bigr), \quad \Phi_\mathrm{I}(w) \defn \frac{1}{\im\sqrt2}\bigl(\Psi(w)-\Psi^*(w)\bigr),
\end{equation*}
Then
\begin{align*}
  \big[\Phi_\mathrm{R}(w_2), \Phi_\mathrm{R}(w_1)\big] & =\im\Im\cinner{w_2}{Qw_1}, \\ \big[\Phi_\mathrm{R}(w_2), \Phi_\mathrm{I}(w_1)\big] & =0,\\ \big[\Phi_\mathrm{I}(w_2), \Phi_\mathrm{I}(w_1)\big] & =\im\Im\cinner{w_2}{Qw_1}.
\end{align*}
Therefore,  according to the real formalism of Subsect. \ref{Canonical
  commutation relations}, the phase space is $\cW\oplus\cW$ considered as a real space, and on each $\cW$ we put the symplectic form $\Im\cinner{\cdot}{Q\cdot}$.

\subsubsection{Fock representations}

Assume, in addition, that $(\cW,Q)$ is Krein.
Let $S_\bullet$ be an admissible involution on $\cW$ and introduce $\Pi_\bullet^{(\pm)}$, $\cZ_\bullet^{(\pm)}:=\Pi_\bullet^{(\pm)}\cW$\ as in Subsect.~\ref{sub:Involutions}.

Then we have a unique centered pure quasi-free state on $\CCR(\cW)$ defined by
\begin{align*}
  \omega_\bullet\bigl(\Psi(v) \Psi^*(w)\bigr) & = \cinner{v}{Q \Pi_\bullet^{(+)} w}, \\ \omega_\bullet\bigl(\Psi^*(v) \Psi(w)\bigr) & = -\cinner{w}{Q \Pi_\bullet^{(-)} v}, \\ \omega_\bullet\bigl(\Psi^*(v) \Psi^*(w)\bigr) & = \omega_\bullet\bigl(\Psi(v) \Psi(w)\bigr) = 0.
\end{align*}

Let us describe explicitly the GNS representation of $\omega_\bullet$.
 The space $\cZ_\bullet^{(+)}$ is a Hilbert space with the scalar
product \eqref{hilbo}.
The space $\bar{\cZ_\bullet^{(-)}}$  has
  the scalar product
  \begin{equation}\label{hilbo1}
  \cinner{\bar{z_1}}{\bar{z_2}} \defn -\cinner{z_2}{Qz_1}=-\bar{\cinner{z_1}{Qz_2}},\quad
  \bar{z_1},\bar{z_2}\in\bar{\cZ_\bullet^{(-)}},\end{equation}
and is also a Hilbert space.
The state $\omega_\bullet$ is represented by the Fock vacuum
$\cinner{\Omega}{\,\cdot\;\Omega}$ in the doubled Fock space
\begin{equation*}
  \Gamma_\s\bigl( \cZ_\bullet^{(+)} \oplus \bar{\cZ_\bullet^{(-)}} \bigr) \simeq \Gamma_\s(\cZ_\bullet^{(+)}) \otimes \Gamma_\s(\bar{\cZ_\bullet^{(-)}}).
\end{equation*}
Denote the creation and annihilation operators by~$a_\bullet^*$ and~$a_\bullet$.
The fields $\Psi$ in the GNS representation given by~$\omega_\bullet$ will be denoted by~$\Psi_\bullet$.
We have
\begin{align*}
  \Psi_\bullet(w) & \defn a_\bullet(\Pi_\bullet^{(+)} w) + a_\bullet^*(\bar{\Pi_\bullet^{(-)} w}), \\ \Psi_\bullet^*(w) & \defn a_\bullet^*(\Pi_\bullet^{(+)} w) + a_\bullet(\bar{\Pi_\bullet^{(-)} w}).
\end{align*}

The operator $\dGamma(S_\bullet)$ plays the role of a charge:
\begin{proposition}
  We have a $U(1)$ group of symmetries
  \begin{align*}
    \e^{\im s \dGamma(S_\bullet)} \Psi_\bullet(w) \e^{-\im s \dGamma(S_\bullet)}   & = \e^{-\im s} \Psi_\bullet(w),
    \\
    \e^{\im s \dGamma(S_\bullet)} \Psi_\bullet^*(w) \e^{-\im s \dGamma(S_\bullet)} & = \e^{\im s} \Psi_\bullet^*(w).
  \end{align*}
\end{proposition}


\subsubsection{Two-component representations}

Suppose $\cV$ is a complex Hilbertizable space and $\cV^*$ is its
antidual. The pairing of $\cV$ and $\cV^*$ is denoted $\cinner{\cdot}{\cdot}$.
We equip $\cV^*\oplus\cV$ with the Hermitian form
\[
  \cinner[\big]{(v_2,u_2)}{Q(v_1,u_1)} = \cinner{u_2}{v_1}+ \cinner{v_2}{u_1} .
\]
Then $\cV^*\oplus\cV$ is a Krein space.

Consider the $*$-algebra generated by $\psi(u)$, $\psi^*(u)$, $u\in\cV^*$, $\eta(v)$, $\eta^*(v)$, $v\in \cV$ satisfying
\begin{enumerate}
  \item
    $\psi(u)^* = \psi^*(u)$, $\eta(v)^* = \eta^*(v)$; \item the maps $\cV^* \ni u \mapsto \psi(u)$, $\cV \ni v \mapsto \eta(v)$ are antilinear; \item the CCR in the complex version of the two-componenet form hold (we write only non-zero commutators):
    \begin{align*}
      \bigl[ \psi(u), \eta^*(v) \bigr] & = \im
      \cinner{u}{v},                           \\
      \bigl[ \psi^*(u), \eta(v) \bigr] & = \im
      \cinner{v}{u},
    \end{align*}
    for $u\in\cV^*, v\in\cV$.
\end{enumerate}
To pass to the formalism of Subsubsect.~\ref{Charged canonical commutation relations}, we set
\begin{align*}
  \Psi(v,u) & \defn \im\eta(v)+\psi(u), \\
  \Psi^*(v,u) & \defn -\im\eta^*(v)+\psi^*(u),\qquad u\in\cV^*, v\in\cV.
\end{align*}
Indeed,
\[
  \big[ \Psi(v_2,u_2), \Psi^*(v_1,u_1)\big] = \cinner{u_2}{v_1}+ \cinner{v_2}{u_1}.
\]

Every admisssible involution $S_\bullet$ on $\cV^*\oplus \cV$ determines a Fock state $\omega_\bullet$.
Let
\[
  \Pi_\bullet^{(\pm)}=
  \begin{bmatrix}
    \Pi_{\bullet 11}^{(\pm)} & \Pi_{\bullet 12}^{(\pm)} \\ \Pi_{\bullet 21}^{(\pm)} & \Pi_{\bullet 22}^{(\pm)}
  \end{bmatrix}
  .
\]
be the corresponding ``particle and antiparticle projection''.
The conditions $\Pi_\bullet^{(\pm)*}Q=Q\Pi_\bullet^{(\pm)}$ and $\Pi_\bullet^{(+)}+\Pi_\bullet^{(-)}=\one$ yield
\[
  \Pi_{\bullet 22}^{(\pm)*} = \Pi_{\bullet 11}^{(\pm)},
  \quad
  \Pi_{\bullet 12}^{(\pm)*} = \Pi_{\bullet 12}^{(\pm)},
  \quad
  \Pi_{\bullet 21}^{(\pm)*} = \Pi_{\bullet 21}^{(\pm)},
\]
\[
  \Pi_{\bullet 22}^{(+)} + \Pi_{\bullet 22}^{(-)} = \Pi_{\bullet 11}^{(+)} + \Pi_{\bullet 11}^{(-)} = \one,
  \quad
  \Pi_{\bullet 12}^{(+)} + \Pi_{\bullet 12}^{(-)} = \Pi_{\bullet 21}^{(+)} + \Pi_{\bullet 21}^{(-)} = 0.
\]
The two-point correlation functions of the state $\omega_\bullet$ are given by
\begin{subequations}\label{recipe}\begin{align}
  \omega_\bullet\bigl(\eta(v)\eta^*(v')\bigr) &= \omega_\bullet\bigl(\eta^*(v') \eta(v)\bigr) = \cinner[\big]{v}{\Pi_{\bullet21}^{(+)}v'}, \\
  \omega_\bullet\bigl(\psi(u)\psi^*(u')\bigr) &= \omega_\bullet\bigl(\psi^*(u') \psi(u)\bigr) = \cinner[\big]{u}{\Pi_{\bullet12}^{(+)}u'}, \\
  \omega_\bullet\bigl(\psi(u)\eta^*(v)\bigr) &= \phantom{-}\ri \cinner[\big]{u}{\Pi_{\bullet11}^{(+)}v}, \qquad \omega_\bullet\bigl(\eta^*(v)\psi(u)\bigr) = -\ri \cinner[\big]{u}{\Pi_{\bullet11}^{(-)}v}, \\
  \omega_\bullet\bigl(\eta(v)\psi^*(u)\bigr) &= -\ri \cinner[\big]{v}{\Pi_{\bullet22}^{(+)}u}, \qquad \omega_\bullet\bigl(\psi^*(u)\eta(v)\bigr) = \phantom{-} \ri \cinner[\big]{v}{\Pi_{\bullet22}^{(-)}u},
\end{align}\end{subequations}


\section{Klein--Gordon equation and Quantum Field Theory on curved space-times}
\label{Quantum Field Theory on curved space-times}

In this section we formulate the main results of this paper in the setting of QFT on curved spacetimes, more precisely, on a globally hyperbolic manifolds equipped with electromagnetic and scalar potentials.
We describe the role various propagators play in the theory of quantum fields satisfying the Klein--Gordon equation.

The usual presentations of this topic make the assumption that all coefficients in the Klein--Gordon equation are smooth.
The results that we have obtained in the previous sections allow us to consider systems with much lower regularity.

\subsection{Half-densities on a pseudo-Riemannian manifold}
\label{Half-densities on a pseudo-Riemannian manifold}

Let $M$ be a manifold.
A \emph{half-density} on $M$ is an assignment of a complex function to every coordinate patch satisfying the following condition: if $x\mapsto f(x)$ and $x'\mapsto f'(x')$ are two such assignments, then
\begin{equation*}
  \Big|\frac{\partial x'}{\partial x}\Big|^{\frac12}f'(x')=f(x),
\end{equation*}
where $ \Big|\frac{\partial x'}{\partial x}\Big|$ denotes the Jacobian of the change of coordinates.
The space of square integrable half-densities will be denoted $L^2(M)$.
Thus if we choose coordinates $x$ and the support of a function $f$ is contained in the corresponding coordinate patch, then
\begin{equation}
  \cinner{f_1}{f_2}=\int \bar{f_1(x)}f_2(x)\dif x,\label{half1}
\end{equation}
where $\dif x=\dif x^1\cdots\dif x^d$ is the Lebesgue measure.
Note that the scalar product \eqref{half1} is independent of coordinates.

Suppose that $M$ is a pseudo-Riemannian manifold with a metric tensor, which in coordinates $x=[x^\mu]$ is given by the matrix $g(x)=[g_{\mu\nu}(x)]$.
Let $|g|(x) \defn \big|\det[g_{\mu\nu}(x)]\big|$.
$M$ has a distinguished density, which in the coordinates $x$ is given by $\sqrt{|g|(x)}\dif x$.
The space of scalar functions, square integrable with respect to $\sqrt{|g|(x)}\dif x$, is denoted $L^2(M,\sqrt{|g|})$.
Obviously \[L^2(M,\sqrt{|g|})\ni f\mapsto|g|^{\frac14}f\in L^2(M)\] is a unitary map.

\subsection{Klein--Gordon equation on spacetime and the conserved current}

Suppose a pseudo-Riemannian manifold is equipped with a vector field $[A^\mu(x)]$, called the electromagnetic potential, and the scalar potential $Y(x)$.

In this and the following subsection we develop the basic formalism thinking of $A,Y$ as smooth functions.
In Subsect.~\ref{Classical propagators} we translate this fomalism to the setting of Sect.~\ref{Abstract Klein-Gordon operator}, which alows for low regularity.

The Klein--Gordon operator, written first in the scalar and then in the half-density formalism in any local coordinates is
\begin{align*}
  K & \defn - \abs{g}^{-\frac12} (D_\mu - A_\mu) \abs{g}^\frac12 g^{\mu\nu} (D_\nu - A_\nu) - Y, \\ K & \defn - \abs{g}^{-\frac14} (D_\mu - A_\mu) \abs{g}^\frac12 g^{\mu\nu} (D_\nu - A_\nu)\abs{g}^{-\frac14} - Y.
\end{align*}
For functions $u,v\in C^\infty (M)$ introduce the \emph{current}, which
again we write first in the scalar, then in the half-density formalism:
\bes
\begin{align*}
  j^\mu(x;\bar u,v) \defn & -\bar{u(x)} g^{\mu\nu}(x)|g|^{\frac12}(x)\bigl(D_\nu-A_\nu(x)\bigr)v(x)                    \\
                          & -\bar{\bigl( D_\nu -A_\nu(x)\bigr)u(x)}g^{\mu\nu}(x)|g|^{\frac12}(x)v(x),
  \\
  j^\mu(x;\bar u,v) \defn & -\bar{u(x)} g^{\mu\nu}(x)|g|^{\frac14}(x)\bigl(D_\nu-A_\nu(x)\bigr)|g|^{-\frac14}(x)v(x)             \\
                          & -\bar{\bigl( D_\nu -A_\nu(x)\bigr)|g|^{-\frac14}(x)u(x)}g^{\mu\nu}(x)|g|^{\frac14}(x)v(x).
\end{align*}
\ees
We check that if $u$, $v$ solve the Klein--Gordon equation, that is
\begin{equation*}
  Ku=Kv=0,
\end{equation*}
then the current $j^\mu(x;\conj{u},v)$ is conserved, that is
\begin{equation*}
  \partial_\mu j^\mu(x;\conj{u},v)=0.
\end{equation*}

Let $M$ be globally hyperbolic, see e.g.\ \cite{bar}.
If $\Omega\subset M$, then $J^\vee(\Omega)$ denotes the \emph{future shadow}, and $J^\wedge(\Omega)$ the \emph{past shadow} of $\Omega$, that is, the set of all points in $M$ that can be reached from $\Omega$ by future/past directed causal paths.
A set $\Theta\subset M$ is called \emph{space compact} if there exists a compact $\Omega\subset M$ such that $\Theta\subset J^\vee(\Omega)\cup J^\wedge(\Omega)$.
$C_\sc(M)$ denotes the set of continuous functions on $M$ with  a space compact support.

Let $\Sol_\sc$ denote the set of smooth  space compact solutions to the
Klein--Gordon equation. For $u,v\in\Sol_\sc$,
\begin{equation}
  \label{symplec}
  \cinner{u}{Qv} \defn   \int_{\Sigma} j^\mu(x;\bar u,v)\dif\mathrm{s}_\mu(x)
\end{equation}
does not depend on the choice of a Cauchy surface $\Sigma$, where $\dif\mathrm{s}_\mu(x)$ denotes the natural measure on $\mathcal{S}$ times the normal vector.
(\ref{symplec}) defines a Hermitian form on $\Sol_\sc$ called the \emph{charge}.

\subsection{Foliating the spacetime}
\label{Foliating the spacetime}

Let us fix a diffeomorphism $I\times\Sigma\to M$, where $I$ is, as usual, $[t_-,t_+]$ or $\mathbb{R}$, and $\Sigma$ is a manifold.
In other words, equip $M$ with a \emph{time function} $t=x^0\in I$ such that all the leaves of the foliation $\Sigma_t = \{t\} \times \Sigma$ are identified with a fixed manifold $\Sigma$.
The generic notation for a point of~$\Sigma$ will be $\vec x$.

The restriction of~$g$ to the tangent space of~$\Sigma_t$ defines a time-dependent family of metrics on~$\Sigma$, denoted $h(t)=h=[h_{ij}]$.
We make the assumption that all $h$ are Riemannian, or, what is equivalent, that the covector $\dif t$ is always timelike.
We set $|h|=\det h$.
In coordinates, the metric can be written as
\begin{align*}
  g_{\mu\nu} \dif x^\mu \dif x^\nu & = -\alpha^2 \dif t^2 + h_{ij} (\dif x^i + \beta^i \dif t) (\dif x^j + \beta^j \dif t), \\ g^{\mu\nu} \partial_\mu \partial_\nu & = -\frac{1}{\alpha^2} (\partial_t - \beta^i \partial_i)^2 + h^{ij} \partial_i \partial_j .
\end{align*}
for some $\alpha(x)>0$ and $[\beta^i(x)]$.
We have $|g|=\alpha^2|h|$.
The Klein--Gordon operator in the half-density formalism can now be written
\begin{align}
  \notag K & = \abs{g}^{-\frac14} (D_0-\beta^i D_i-A_0+\beta^iA_i) \frac{\abs{g}^\frac12}{\alpha^2} (D_0-\beta^i D_i-A_0+\beta^iA_i)\abs{g}^{-\frac14} \\\notag & \quad - \abs{h}^{-\frac14} (D_i - A_i) \abs{h}^\frac12 h^{ij} (D_j - A_j)\abs{h}^{-\frac14} - Y\\  & = \bigl(D_t+W^*(t)\bigr)\frac{1}{\alpha^2(t)}\bigl(D_t+W(t)\bigr)-L(t),\label{hilbert1}
\end{align}
where
\bes
\begin{align}
  \label{beltrami0}
  L(t) & \defn \abs{g}^{-\frac14} (D_i - A_i) \abs{g}^\frac12 h^{ij} (D_j - A_j)\abs{g}^{-\frac14} + Y \\ \label{beltrami}
       & \mathrel{\phantom{\defn}\mathllap{=}} \abs{h}^{-\frac14} \Big(D_i -\frac\ri{2\alpha}\alpha_{,i}- A_i\Big)
  \abs{h}^\frac12 h^{ij} \Big(D_j
  + \frac\ri{2\alpha}\alpha_{,i}- A_j\Big)\abs{h}^{-\frac14} + Y                                       \\
  W(t) & \defn \beta^iD_i-A_0+\beta^iA_i+\frac{\ri}{4\abs{g}}\abs{g}_{,0}
  -\frac{\ri}{4\abs{g}}\beta^i\abs{g}_{,i}
\end{align}
\ees

The temporal component of the current, again in the half-density formalism, is
\begin{align*}
  \notag j^0(x;\bar u,v) & = \bar{ u(x)}\frac{1}{\alpha^2(x)}\bigl(D_0-W(x)\bigr)v(x) \\  &\quad +\bar{\frac{1}{\alpha^2(x)}\bigl(D_0-W(x)\bigr)u(x)}v(x).
\end{align*}

We use the half-density formalism to define the spaces $L^2(M)$ and $L^2(\Sigma)$.
We define $L^2(I)$ using the Lebesgue measure.
We have
\begin{equation}
  \label{hilbert}
  L^2(M)\simeq L^2(I,\Sigma)= L^2(I)\otimes L^2(\Sigma),
\end{equation}

We treat $\Sigma$ as equipped with the metric $h(t)$.
The operator $L(t)$ is a Hermitian operator on $C_\mathrm{c}^\infty(\Sigma)$ in the sense of the Hilbert space $L^2(\Sigma)$.
We have written it in two ways: \eqref{beltrami0} looks simpler, but the expression \eqref{beltrami} is manifestly covariant with respect to a change of coordinates on $\Sigma$.

For  brevity, we will write $\cK=L^2(\Sigma)$.
Note that $K$ in \eqref{hilbert1} has the form of an abstract Klein--Gordon operator considered in Sect.~\ref{Abstract Klein-Gordon operator}.

Choosing $\Sigma_t$ for the Cauchy surface we can rewrite \eqref{symplec} as
\begin{align*}
  \cinner{u}{Qv} & =\int_\Sigma \bar{u(t,\vec x)}\frac{1}{\alpha^2(t,\vec x)}\bigl(D_0+W(t,\vec x)\bigr) v(t,\vec x)\dif\vec x \\ & \quad + \int_\Sigma \bar{\bigl(D_0+W(t,\vec x)\bigr) u(t,\vec x)}\frac{1}{\alpha^2(t,\vec x)} v(t,\vec x) \dif\vec x,\quad u,v\in\cK.
\end{align*}

\subsection{Classical propagators}
\label{Classical propagators}

After identifying $M\simeq I\times\Sigma$, (\ref{hilbert}) and (\ref{hilbert1}) show that we are in the setting of Sect.~\ref{Abstract Klein-Gordon operator} devoted to the abstract Klein Gordon
operator.
From now on we impose Assumption \ref{ass.I}.
We introduce the formalism of Sect.~\ref{Abstract Klein-Gordon operator}, such as the Hilbertizable spaces $\cK_\beta,$ $\beta\in[-1,1]$ and $\cW_\lambda$, $\lambda\in[-\frac12,\frac12]$, as in \eqref{cK} and \eqref{cW}, the generator of the evolution $B(t)$ and the evolution itself $R(t,s)$.
In particular, the space $\cW_0$ is a Krein space equipped with the form $Q$.

It is easy to construct the classical propagators in this setting.
First we define $G^\bullet$, $\bullet=\vee,\wedge,\PJ$ as operators $C_\mathrm{c}(I,\cK) \to C(I,\cK)$ as in \eqref{pkc1.}.
By the Schwartz kernel theorem they possess distributional kernels, which we denote $G^\bullet(x,y) =G^\bullet(t,\vec x;s,\vec y) $,

It is obvious that
\begin{align*}
  \supp G^\wedge\subset\{(t,\vec x;s,\vec y)\in M\times M \mid t\leq s\}, \\
  \supp G^\vee\subset \{(t,\vec x;s,\vec y)\in M\times M \mid t\geq s\}.
\end{align*}

One can expect their support to be even smaller, more precisely, that they are \emph{causal}:
\begin{align*}
  \supp G^\wedge\subset\{(x,y)\in M\times M \mid x\in J^\wedge(x)\}, \\
  \supp G^\vee\subset\{(x,y)\in M\times M \mid x\in J^\vee(x)\},\\
  \supp G^\PJ\subset\{(x,y)\in M\times M \mid x\in J(x)\}.
\end{align*}
If $g,A,Y$ are smooth, this is very well-known, proven in numerous sources.
Under Assumption \ref{ass.I} this is presumably also true.
Under slightly more restrictive assumptions it follows from Theorem E1 of \cite{derezinski-siemssen:propagators}, see also \cite{sanchez} for a different approach.

\subsection{Non-classical propagators}
\label{Non-classical propagators.}

If $I$ is finite, we choose two admissible involutions $S_\pm$ on $\cW_0$.

If $I=\mathbb{R}$, we assume that the spacetime is stationary for large times and assume that $B_\pm $ are stable.
We set $S_\pm \defn \sgn\bigl(B_\pm \bigr) $, which are automatically admissible involutions.

With help of these admissible involutions, we define the non-classical propagators $G_\pm^{(+)},$ $G_\pm^{(-)},$ $G^\Feyn,$ $G^\aFeyn$.

If $I$ is finite, the non-classical propagators can be interpreted as bounded operators on $L^2(M)$.
If $G^\Feyn,$ $G^\aFeyn$ have zero nullspaces, then we define
\begin{equation*}
  K^\Feyn \defn G^{\Feyn-1},\quad K^\aFeyn \defn G^{\aFeyn-1},
\end{equation*}
which can be treated as well-posed realizations of the Klein--Gordon operator satisfying $K^{\Feyn*}=K^{\aFeyn}$.

If $I=\mathbb{R}$, then the non-classical propagators can be understood as, say, operators $C_\mathrm{c}(I,L^2(\Sigma)) \to C(I,L^2(\Sigma))$.

Suppose we impose the assumption of Prop.~\ref{analog}.
It is then clear that the operator $K$ is Hermitian (or, as it is often termed, symmetric) on $C_\mathrm{c}^\infty(M)$.
One can ask about the existence of its self-adjoint extensions.
This is the subject of following conjecture, which is essentially a spacetime version of Conj.~\ref{coject} and~\ref{conject2}.

\begin{conjecture}
  For a large class of asymptotically stationary and stable Klein--Gordon operators the following holds:
  \begin{enumerate}
    \item The operator $K$ with the domain $C_\mathrm{c}^\infty(M)$ is essentially self-adjoint in the sense of $L^2(M)$.
          Denote its unique self-adjoint extension by $K^{\mathrm{s.a.}}$.
    \item In the sense of $\langle t\rangle^{-s} L^2(M)\to\langle t\rangle^{s}L^2(M)$, for $s>\frac12$,
          \begin{align*}
            \slim_{\epsilon\searrow0} (K^{\mathrm{s.a.}}-\ri\epsilon)^{-1} & =G^{\Feyn},  \\
            \slim_{\epsilon\searrow0} (K^{\mathrm{s.a.}}+\ri\epsilon)^{-1} & =G^{\aFeyn}.
          \end{align*}
  \end{enumerate}
  \label{coject2}
\end{conjecture}

Note that Conj.~\ref{coject2} is true in the stable stationary case, see \cite{derezinski-siemssen:static} and \cite{derezinski-siemssen:propagators}.
As proven by Vasy \cite{vasy:selfadjoint} and Nakamura--Taira
\cite{nakamura,nakamura2,nakamura3}, it is also true for some classes
of asymptotically Minkowskian spacetimes.  Kami\'{n}ski
described a counterexample to a certain strong version of this conjecture \cite{kaminski}.

\subsection{Charged fields}

Let us now describe the formalism of classical and quantum field theory in our setting.
Note that the formalism described in the introduction used pointlike fields.
In this section pointlike fields may be ill-defined because of insufficient smoothness.
Therefore we prefer to use smeared fields.

Recall that elements of $\cW_0$ can be written as two component vectors:
\begin{equation*}
  w=
  \begin{bmatrix}
    w_1 \\w_2
  \end{bmatrix}
  .
\end{equation*}
The space $\cW_0$ is preserved by the dynamics $R(t,s)$.
We treat time $t=0$ as the ``reference time''.

For any $t\in I$, $u\in\cK_{-\frac12}$, $v\in\cK_{\frac12}$ we define the following functionals on $\cW_0$:
\bes
\begin{align*}
  \rinner{\psi_t(u)}{w} & =\int\bar{u(\vec x)}\bigl( R(t,0)w \bigr)_1 (\vec x)\dif\vec x, \\ \rinner{\psi^*_t(u)}{w} & =\int u(\vec x)\bar{\bigl( R(t,0)w \bigr)_1 (\vec x)}\dif\vec x,\\ \rinner{\eta_t(v)}{w} & =-\ri\int\bar{v(\vec x)}\bigl( R(t,0)w \bigr)_2 (\vec x)\dif\vec x,\\ \rinner{\eta^*_t(v)}{w} & =\ri\int v(\vec x)\bar{\bigl( R(t,0)w \bigr)_2 (\vec x)}\dif\vec x.
\end{align*}
\ees

From the symplectic structure associated with the charge form $Q$ we derive the Poisson brackets between $\psi_t,\psi_t^*,\eta_t,\eta_t^*$.
Below we present only the non-zero cases:
\begin{align*}
  \big\{\psi_t(u),\eta^*_t(v)\big\}= & \int\bar{u(\vec x)}v(\vec x)\dif\vec x, \\ \big\{\psi^*_t(u),\eta_t(v)\big\}= & \int u(\vec x)\bar{v(\vec x)}\dif\vec x.
\end{align*}

The first step of quantization is the replacement of the Poisson bracket by $\im$ times the commutator.
Thus we obtain the commutation relations
\begin{align*}
  \big[\hat\psi_t(u),\hat\eta^*_t(v)\big] &= \im\int\bar{u(\vec x)}v(\vec x)\dif\vec x, \\ \big[\hat\psi^*_t(u),\hat\eta_t(v)\big] &= \im\int u(\vec x)\bar{v(\vec x)}\dif\vec x.
\end{align*}

Then one chooses the in and the out Fock state.
Recall that they are determined by two admissible involutions $S_\pm$.
From $S_\pm$ we obtain two pairs of complementary projections $\Pi_\pm^{(+)},\Pi_\pm^{(-)}$.
Following the recipe \eqref{recipe}, $\Pi_\pm^{(+)}$ and $\Pi_\pm^{(-)}$ are used to define the Fock states $\omega_\pm=\cinner{\Omega_\pm}{\,\cdot\;\Omega_\pm}$.
In the slab geometry case this is done using the fields $\hat\psi_{t_\pm}, \hat\psi_{t_\pm}^*, \hat\eta_{t_\pm}, \hat\eta_{t_\pm}^*$.
In the unrestricted time case it is done using $\hat\psi_{t}, \hat\psi_{t}^*, \hat\eta_{t}, \hat\eta_{t}^*$ with $\pm t>T$.

We can also smear the fields with spacetime functions.
Suppose that, say, $f\in C_\mathrm{c}(I,\cK_{-\frac12}).$
Set
\begin{align*}
  \psi[f] \defn
  \int
  \psi_t\bigl(f(t,\,\cdot\,)\bigr)\dif
  t,\quad
  \psi^*[f] \defn
  \int
  \psi_t^*\bigl(f(t,\,\cdot\,)\bigr)\dif t.
\end{align*}
The Poisson bracket of the fields is known then as the \emph{Peierls bracket}:
\begin{align}\label{peierls}
  \big\{\psi[f_1],\psi^*[f_2]\big\} &
  =-\iint\bar{f_1(x)}
  G^\PJ(x,y)f_2(y)\dif x\dif y.
\end{align}

Similarly, we can smear the quantum fields:
\begin{align}\label{peierls2}
  \hat\psi[f] \defn \int \hat\psi_t\bigl(f(t,\cdot)\bigr)\dif t,\quad \hat\psi^*[f] \defn \int \hat\psi_t^*\bigl(f(t,\cdot)\bigr)\dif t.
\end{align}
The commutator of fields is expressed by the Peierls bracket.
\bes\label{field}
\begin{align}
  \label{field1}
  \big[\hat\psi[f_1],\hat\psi^*[f_2]\big] &
  =-\im\iint\bar{f_1(x)}
  G^\PJ(x,y)f_2(y)\dif x\dif y.
\end{align}
The vacuum expectation values of the products of fields are expressed by the positive/negative frequency bisolutions:
\begin{align}
  \label{field2}
  \cinner[\big]{\Omega_\pm}{\hat\psi[f_1]\hat\psi^*[f_2]\Omega_\pm} & =
  \iint\bar{f_1(x)}
  G_\pm^{(+)}(x,y)f_2(y)\dif x\dif y,                               \\\label{field3}
  \cinner[\big]{\Omega_\pm}{\hat\psi^*[f_2]\hat\psi[f_1]\Omega_\pm} & =
  \iint\bar{f_1(x)}
  G_\pm^{(-)}(x,y)f_2(y)\dif x\dif y.
\end{align}
\ees
Let $\torder{}$ denote the time-ordered product and and $\atorder{}$ the reverse time-ordered product.
Assume in addition the Shale condition for $\omega_+$ and $\omega_-$, so that $\Omega_+$ and $\Omega_-$ can be treated as vectors in the same representation.
Then the vacuum expectation values of the time-ordered and reverse time-order products divided by the overlap of the vacua is expressed by the Feynman, resp.\ anti-Feynman inverses:
\bes
\label{feynman-c-all}
\begin{align}
  \label{feynman-c}
  \frac{\cinner[\big]{\Omega_+}{\torder{\hat\psi[f_1] \hat\psi^*[f_2]}\Omega_-}}{\cinner{\Omega_+}{\Omega_-}}  & = -\im\iint\bar{ f_1(x)}
  G^\Feyn(t,s) f_2(y) \dif x \dif y,
  \\
  \label{feynman-c-}
  \frac{\cinner[\big]{\Omega_-}{\atorder{\hat\psi[f_1] \hat\psi^*[f_2]}\Omega_+}}{\cinner{\Omega_-}{\Omega_+}} & = \im\iint\bar{ f_1(x)}
  G^\aFeyn(t,s) f_2(y) \dif x \dif y.
\end{align}
\ees
Note, however, that as already mentioned in the introduction the RHS of~\eqref{feynman-c-all} is well-defined also if the Shale condition is not satisfied.

Here are identities that can be used to prove the above formulas:
\bes\label{peierls0}\begin{align}
  \psi_t(u)&=\psi_s\big(R(s,t)_{22}u\big)+\ri\eta_s\big(R(s,t)_{12}u\big),\\
  \psi_t^*(u)&=\psi_s^*\big(R(s,t)_{22}u\big)-\ri\eta_s^*\big(R(s,t)_{12}u\big).
\label{peierls3}  \end{align}
\ees
They follow  the definition of $\psi_t$ and $\psi_t^*$ and from
the pseudounitarity of $R(s,t)$.
  Commuting  $\psi_s(v)$ with \eqref{peierls3} we obtain
  \begin{align*}
    \{\psi_s(v),\psi_t(u)\}&=-\ri\cinner{v}{R(s,t)_{12}u}\\&=-\int \bar{v(\vec
      x)}G^\PJ(s,\vec x;t,\vec y)u(\vec y)\dif\vec x\dif\vec y,
    \end{align*} from which the formula for the Peierls bracket \eqref{peierls} follows.

Clearly, the quantized version of \eqref{peierls0}  with all fields
decorated with hats is also true. It implies \eqref{field1}

\begin{remark}
  In most physics literature one uses \emph{pointlike fields}, denoted typically $\hat\psi^*(x),\hat\psi(x)$, $x\in M$, (not to be confused with the spatially smeared fields $\hat\psi_t^*(u),\hat\psi_t(u)$).
  Formally, the smeared-out fields are given by
  \begin{align*}
    \hat\psi[f] & \defn \int \bar{f(x)}\hat\psi(x)\dif x, \\
    \hat\psi^*[f] & \defn \int f(x)\hat\psi^*(x)\dif x.
  \end{align*}
  Smeared-out fields are more typical for the mathematics literature, since they can be interpreted as closed densely defined operators (at least in a linear QFT).
  Nevertheless, pointlike fields are convenient.
  We used them in the introduction.
  Note that identities (\ref{field1}), (\ref{field2}), (\ref{field3}), (\ref{feynman-c}) and (\ref{feynman-c-}) are equivalent to identities (\ref{two1}), (\ref{two2a.}), (\ref{two2b.}), (\ref{two3.}) and (\ref{two4.}).
\end{remark}

\subsection{Neutral fields}

Suppose that the electromagnetic potential $[A^\mu]$ is absent.
Then the Klein--Gordon operator
\begin{align*}
  K & \defn \abs{g}^{-\frac14} \partial_\mu \abs{g}^\frac12 g^{\mu\nu} \partial_\nu \abs{g}^{-\frac14} - Y.
\end{align*}
is real.
In particular, the spaces $\cK_\beta$, $\cW_\lambda$ can be equipped with the usual complex conjugation and their real subspaces $\cK_{\beta,\mathbb{R}}$, $\cW_{\lambda,\mathbb{R}}$ can be defined.
Note that $\cW_{0,\mathbb{R}}$ is a real Krein space equipped with the symplectic form
\begin{equation*}
  \rinner{v}{\omega w} \defn \Im\cinner{v}{Qw}.
\end{equation*}
Also note that the evolution satisfies $R(t,s)=\bar{R(t,s)}$ and thus it can be restricted to $\cW_{0,\mathbb{R}}$.

For any $t\in I$, $u\in\cK_{-\frac12,\mathbb{R}}$, $v\in\cK_{\frac12,\mathbb{R}}$ we define the following functionals on $\cW_{0,\mathbb{R}}$:
\begin{align*}
  \rinner{\phi_t(u)}{w} & =\int u(\vec x)\bigl( R(t,0)w \bigr)_1 (\vec x)\dif\vec x, \\
  \rinner{\pi_t(v)}{w}  & =-\ri\int v(\vec x)\bigl( R(t,0)w \bigr)_2 (\vec x)\dif\vec x.
\end{align*}

From the symplectic structure given by the form $\omega$ we derive the Poisson brackets between $\phi_t,\pi_t$.
Below we present the only non-zero relation:
\begin{align*}
  \big\{\phi_t(u),\pi_t(v)\big\}=\int u(\vec x)v(\vec x)\dif\vec x.
\end{align*}

The first step of quantization is the replacement of the Poisson bracket by $\im$ times the commutator.
Thus we obtain the commutation relations
\begin{align*}
  \big[\hat\phi_t(u),\hat\pi_t(v)\big]=\im\int u(\vec x) v(\vec x)\dif\vec x.
\end{align*}

Then one chooses the in and the out Fock state -- as in the charged case.
In addition, in the slab geometry case, we demand that the two admissible involutions $S_\pm$ on $\cW_0$ are anti-real.
(In the unrestricted time case this is automatic.)
Thus the two pairs of complementary projections $\Pi_\pm^{(+)},\Pi_\pm^{(-)}$ obtained from $S_\pm$ are real.

We can also smear the fields with space-time functions.
Suppose that, say, $f\in C_\mathrm{c}(I,\cK_{-\frac12,\mathbb{R}}).$
Set
\begin{align*}
  \hat\phi[f] \defn
  \int
  \hat\phi_t\bigl(f(t,\cdot)\bigr)\dif
  t.
\end{align*}
Now we have the following identities:
\bes\label{field-r}
\begin{align}
  \label{field1-r}
  \big[\hat\phi[f_1],\hat\phi[f_2]\big]                                                                      &
  =-\im\iint f_1(x)G^\PJ(x,y)f_2(y)\dif x\dif y,                                                                                                                     \\\label{field2-r}
  \cinner[\big]{\Omega_\pm}{\hat\phi[f_1]\hat\phi[f_2]\Omega_\pm}                                                & =
  \iint f_1(x)G_\pm^{(+)}(x,y)f_2(y)\dif x\dif y,
  \\  \label{feynman-c-r}
  \frac{\cinner[\big]{\Omega_+}{\torder{\hat\phi[f_1] \hat\phi[f_2]}\Omega_-}}{\cinner{\Omega_+}{\Omega_-}}  & = -\im\iint f_1(x) G^\Feyn(t,s) f_2(y) \dif x \dif y, \\
  \label{feynman-c-r-}
  \frac{\cinner[\big]{\Omega_-}{\atorder{\hat\phi[f_1] \hat\phi[f_2]}\Omega_+}}{\cinner{\Omega_-}{\Omega_+}} & = \im\iint f_1(x)
  G^\aFeyn(t,s) f_2(y) \dif x \dif y.
\end{align}
\ees
As in the charged case, also here the Shale conditions is required for the LHS of \eqref{feynman-c-r} and \eqref{feynman-c-r-} to be well-defined.

\begin{remark}
  In most physics literature one uses \emph{pointlike fields}, denoted typically $\hat\phi(x)$, $x\in M$.
  Formally, the smeared-out fields are given by
  \begin{align*}
    \hat\phi[f] & \defn \int f(x)\hat\phi(x)\dif x.
  \end{align*}
  for $f\in C_\mathrm{c}^\infty(M,\mathbb{R})$,
\end{remark}


\begin{acknowledgement}
J.D.\ acknowledges the support of National Science Center (Poland) under the grant UMO-2019/35/B/ST1/01651.
He also benefited from the support of Istituto Nazionale di Alta Matematica ``F.~Severi'', through the Intensive Period ``INdAM Quantum Meetings (IQM22)''.
Moreover, he is grateful to C.~G\'{e}rard, M.~Wrochna and W.~Kami\'{n}ski for useful discussions.
\end{acknowledgement}


\begin{thebibliography}{10}
  \providecommand{\url}[1]{{#1}}
  \providecommand{\urlprefix}{URL }
  \expandafter\ifx\csname urlstyle\endcsname\relax
    \providecommand{\doi}[1]{DOI~\discretionary{}{}{}#1}\else
    \providecommand{\doi}{DOI~\discretionary{}{}{}\begingroup
    \urlstyle{rm}\Url}\fi

  \bibitem{bar-fredenhagen}
  B{\"a}r, C., Fredenhagen, K. (eds.): Quantum Field Theory on Curved Spacetimes.
  \newblock No. 786 in Lecture Notes in Physics. Springer (2009).
  \newblock \doi{10.1007/978-3-642-02780-2}

  \bibitem{bar}
  B{\"a}r, C., Ginoux, N., Pf{\"a}ffle, F.: Wave Equations on {Lorentzian}
    Manifolds and Quantization.
  \newblock ESI Lectures in Mathematical Physics. European Mathematical Society
    (2007).
  \newblock \doi{10.4171/037}

  \bibitem{birrell-davies}
  Birrell, N.D., Davies, P.C.W.: Quantum Fields in Curved Space.
  \newblock Cambridge Monographs on Mathematical Physics. Cambridge University
    Press (1984)

  \bibitem{bjorken}
  Björken, J.D., Drell, S.D.: Relativistic Quantum Fields.
  \newblock McGraw-Hill (1965)

  \bibitem{bogoliubov}
  Bogoliubov, N.N., Shirkov, D.V.: Introduction to the Theory of Quantized
    Fields, 3 edn.
  \newblock John Wiley \& Sons (1980)

  \bibitem{brunetti-fredenhagen}
  Brunetti, R., Fredenhagen, K.: Microlocal analysis and interacting quantum
    field theories: {Renormalization} on physical backgrounds.
  \newblock Communications in Mathematical Physics \textbf{208}(3), 623--661
    (2000).
  \newblock \doi{10.1007/s002200050004}

  \bibitem{derezinski-gerard}
  Derezi{\'n}ski, J., G{\'e}rard, C.: Mathematics of Quantization and Quantum
    Fields.
  \newblock Cambridge Monographs on Mathematical Physics. Cambridge University
    Press (2013)

  \bibitem{derezinski-siemssen:static}
  Derezi{\'n}ski, J., Siemssen, D.: Feynman propagators on static spacetimes.
  \newblock Reviews in Mathematical Physics \textbf{30}(3), 1850006 (2018)

  \bibitem{derezinski-siemssen:propagators}
  Derezi{\'n}ski, J., Siemssen, D.: An evolution equation approach to the
    {Klein}--{Gordon} operator on curved spacetime.
  \newblock Pure and Applied Analysis \textbf{1}(2), 215--261 (2019).
  \newblock \doi{10.2140/paa.2019.1.215}

  \bibitem{dewitt:curved}
  DeWitt, B.S.: Quantum field theory in curved spacetime.
  \newblock Physics Reports \textbf{19}(6), 295--357 (1975).
  \newblock \doi{10.1016/0370-1573(75)90051-4}

  \bibitem{duistermaat}
  Duistermaat, J.J., H{\"o}rmander, L.: {Fourier} integral operators. {II}.
  \newblock Acta Mathematica \textbf{128}(1), 183--269 (1972).
  \newblock \doi{10.1007/BF02392165}

  \bibitem{EE}
  Edmunds, D.E., Evans, W.D.: Spectral Theory and Differential Operators.
  \newblock Oxford University Press (2018)

  \bibitem{fewster-verch:spass}
  Fewster, C.J., Verch, R.: Dynamical locality and covariance: What makes a
    physical theory the same in all spacetimes?
  \newblock Annales Henri Poincar{\'e} \textbf{13}(7), 1613--1674 (2012).
  \newblock \doi{10.1007/s00023-012-0165-0}

  \bibitem{fewster-verch:necessary}
  Fewster, C.J., Verch, R.: The necessity of the {Hadamard} condition.
  \newblock Classical and Quantum Gravity \textbf{30}(23), 235027 (2013).
  \newblock \doi{10.1088/0264-9381/30/23/235027}

  \bibitem{friedlander}
  Friedlander, F.G.: The Wave Equation on a Curved Space-Time.
  \newblock Cambridge University Press (1975)

\bibitem{fukuma}  Fukuma, M.,  Sugishita, S., Sakatani, Y.:
  Propagators in de Sitter space. \newblock
  Physical Review D \textbf{88}, 024041 (2013)

  \bibitem{fulling}
  Fulling, S.A.: Aspects of Quantum Field Theory in Curved Space-Time.
  \newblock No.~17 in London Mathematical Society Student Texts. Cambridge
    University Press (1989)

  \bibitem{fulling-narcowich-wald}
  Fulling, S.A., Narcowich, F.J., Wald, R.M.: Singularity structure of the
    two-point function in quantum field theory in curved spacetime, {II}.
  \newblock Annals of Physics \textbf{136}(2), 243--272 (1981).
  \newblock \doi{10.1016/0003-4916(81)90098-1}

  \bibitem{gerard-wrochna:inout}
  G{\'e}rard, C., Wrochna, M.: {Hadamard} property of the in and out states for
    {Klein}-{Gordon} fields on asymptotically static spacetimes.
  \newblock Annales Henri Poincar{\'e} \textbf{18}(8), 2715--2756 (2017).
  \newblock \doi{10.1007/s00023-017-0573-2}

  \bibitem{gerard-wrochna:feynman}
  G{\'e}rard, C., Wrochna, M.: The massive {Feynman} propagator on asymptotically
    {Minkowski} spacetimes.
  \newblock American Journal of Mathematics \textbf{141}(6), 1501--1546 (2019)

  \bibitem{hollands}
  Hollands, S.: Renormalized quantum {Yang}--{Mills} fields in curved spacetime.
  \newblock Reviews in Mathematical Physics \textbf{20}(9), 1033--1172 (2008).
  \newblock \doi{10.1142/S0129055X08003420}

  \bibitem{hollands-wald1}
  Hollands, S., Wald, R.M.: Local {Wick} polynomials and time ordered products of
    quantum fields in curved spacetime.
  \newblock Communications in Mathematical Physics \textbf{223}(2), 289--326
    (2001).
  \newblock \doi{10.1007/s002200100540}

  \bibitem{sanchez}
  Hörmann, G., Sanchez~Sanchez, Y., Spreitzer, C., Vickers, J.A.: Green
    operators in low regularity spacetimes and quantum field theory.
  \newblock Classical and Quantum Gravity \textbf{37}(17), 175009 (2020).
  \newblock \doi{10.1088/1361-6382/ab839a}

  \bibitem{kaminski}Kami\'{n}ski, W.:
  Non-Self-Adjointness of the Klein-Gordon Operator on a Globally
  Hyperbolic and Geodesically Complete Manifold: An Example.
  \newblock Annales Henri Poincaré \textbf{23}, 4409–4427 (2022)

  \bibitem{kato:hyperbolic}
  Kato, T.: Linear evolution equations of "hyperbolic" type.
  \newblock Journal of the Faculty of Science, University of Tokyo, Section I
    \textbf{17}, 241--258 (1970)

  \bibitem{leray}
  Leray, J.: Hyperbolic Differential Equations.
  \newblock Unpublished lecture notes. The Institute for Advanced Study,
    Princeton, N.J. (1953)

  \bibitem{nakamura}
  Nakamura, S., Taira, K.: Essential self-adjointness of real principal type
    operators.
  \newblock Annales Henri Lebesgue \textbf{4}, 1035--1059 (2021).
  \newblock \doi{10.5802/ahl.96}

  \bibitem{nakamura3}
  Nakamura, S., Taira, K.: Essential self-adjointness for the {Klein}-{Gordon}
    type operators on asymptotically static spacetime  (2022)

  \bibitem{nakamura2}
  Nakamura, S., Taira, K.: A remark on the essential self-adjointness for
    {Klein}-{Gordon} type operators  (2022)

  \bibitem{parker-toms}
  Parker, L.E., Toms, D.J.: Quantum Field Theory in Curved Spacetime.
  \newblock Cambridge Monographs on Mathematical Physics. Cambridge University
    Press (2009)

  \bibitem{radzikowski}
  Radzikowski, M.J.: Micro-local approach to the {Hadamard} condition in quantum
    field theory on curved space-time.
  \newblock Communications in Mathematical Physics \textbf{179}(3), 529--553
    (1996).
  \newblock \doi{10.1007/BF02100096}

  \bibitem{rumpf1}
  Rumpf, H., Urbantke, H.K.: Covariant ``in--out'' formalism for creation by
    external fields.
  \newblock Annals of Physics \textbf{114}, 332--355 (1978).
  \newblock \doi{10.1016/0003-4916(78)90273-7}

  \bibitem{vasy:selfadjoint}
  Vasy, A.: Essential self-adjointness of the wave operator and the limiting
    absorption principle on {Lorentzian} scattering spaces.
  \newblock Journal of Spectral Theory \textbf{10}(2), 439--461 (2020)

  \bibitem{wald:backreaction}
  Wald, R.M.: The back reaction effect in particle creation in curved spacetime.
  \newblock Communications in Mathematical Physics \textbf{54}(1), 1--19 (1977).
  \newblock \doi{10.1007/BF01609833}

  \bibitem{wald}
  Wald, R.M.: Quantum Field Theory in Curved Spacetime and Black Hole
    Thermodynamics.
  \newblock Chicago Lectures in Physics. University of Chicago Press (1994)

\end{thebibliography}

\end{document}